\documentclass[12pt,reqno]{amsart}
\usepackage{fullpage}

\newtheorem{theorem}{Theorem}[section]
\newtheorem{lemma}[theorem]{Lemma}

\newtheorem{cor}[theorem]{Corollary}
\newtheorem{conj}{Conjecture}
\usepackage{graphicx}
\usepackage{color}
\usepackage{subfigure}
\usepackage{amssymb}
\usepackage{amsmath}
\usepackage{colonequals}
\usepackage{hyperref}

\theoremstyle{definition}
\newtheorem{definition}[theorem]{Definition}

\newtheorem{question}{Question}

\renewcommand{\subset}{\subseteq}
\renewcommand{\supset}{\supseteq}
\renewcommand{\epsilon}{\varepsilon}
\renewcommand{\nu}{v}

\newcommand{\abs}[1]{\left|#1\right|}                   
\newcommand{\absf}[1]{|#1|}                             
\newcommand{\vnorm}[1]{\left\|#1\right\|}    
\newcommand{\vnormf}[1]{\|#1\|}                         

\newcommand{\sign}[1]{\mbox{sign}#1}
\newcommand{\Z}{\mathbb{Z}}                             
\newcommand{\N}{\mathbb{N}}
\newcommand{\E}{\mathbb{E}}

\newcommand{\R}{\mathbb{R}}

\newcommand{\figwidth}{.4\textwidth}                
\newcommand{\figwidtha}{.8\textwidth}                
\newcommand{\embolden}[1]{\textbf {#1}}

\begin{document}

\title{Euclidean Partitions Optimizing Noise Stability}
%
\thanks{This research is supported by NSF Graduate Research Fellowship DGE-0813964.  Part of this work was carried out while visiting the Quantitative Geometry program at MSRI}
\author{Steven Heilman}
\address{Courant Institute, New York University, New York NY 10012}
\email{heilman@cims.nyu.edu}

%
\begin{abstract}
The Standard Simplex Conjecture of Isaksson and Mossel \cite{isaksson11} asks for the partition $\{A_{i}\}_{i=1}^{k}$ of $\mathbb{R}^{n}$ into $k\leq n+1$ pieces of equal Gaussian measure of optimal noise stability.  That is, for $\rho>0$, we maximize
$$
\sum_{i=1}^{k}\int_{\mathbb{R}^{n}}\int_{\mathbb{R}^{n}}1_{A_{i}}(x)1_{A_{i}}(x\rho+y\sqrt{1-\rho^{2}})
e^{-(x_{1}^{2}+\cdots+x_{n}^{2})/2}e^{-(y_{1}^{2}+\cdots+y_{n}^{2})/2}dxdy.
$$
Isaksson and Mossel guessed the best partition for this problem and proved some applications of their conjecture.  For example, the Standard Simplex Conjecture implies the Plurality is Stablest Conjecture.  For $k=3,n\geq2$ and $0<\rho<\rho_{0}(k,n)$, we prove the Standard Simplex Conjecture.  The full conjecture has applications to theoretical computer science \cite{isaksson11,khot07,mossel10} and to geometric multi-bubble problems.
\end{abstract}
\maketitle

\section{Introduction}\label{secintro}

The Standard Simplex Conjecture \cite{isaksson11} asks for the partition $\{A_{i}\}_{i=1}^{k}$ of $\R^{n}$ into $k\leq n+1$ sets of equal Gaussian measure of optimal noise stability.  This Conjecture generalizes a seminal result of Borell, \cite{borell85,mossel10}, which corresponds to the $k=2$ case of the Standard Simplex Conjecture.  Borell's result says that the two disjoint regions of fixed Gaussian measures $0<a<1$ and $1-a$ and of optimal noise stability must be separated by a hyperplane.  Since two disjoint sets of total Gaussian measure $1$ can be described by a single set and its complement, Borell's result can be stated as follows.  Let $A\subset\R^{n}$ have Gaussian measure $0<a<1$ and let $\rho\in(0,1)$.  Then the following quantity, which is referred to as the noise stability of $A$, is maximized when $A$ is a half-space.
\begin{equation}\label{zero0}
\int_{\R^{n}}\int_{\R^{n}}1_{A}(x)1_{A}(x\rho+y\sqrt{1-\rho^{2}})
e^{-(x_{1}^{2}+\cdots+x_{n}^{2})/2}e^{-(y_{1}^{2}+\cdots+y_{n}^{2})/2}dxdy.
\end{equation}
When we say that $A$ is a half-space, we mean that $A$ is the set of points lying on one side of a hyperplane.  If $\rho\in(-1,0)$, then the noise stability \eqref{zero0} of $A$ is minimized among all sets of Gaussian measure $a$, when $A$ is a half-space.  We can rewrite \eqref{zero0} probabilistically as follows.  Let $X=(X_{1},\ldots,X_{n}),Y=(Y_{1},\ldots,Y_{n})\in\R^{n}$ be two standard Gaussian random vectors such that $\E(X_{i}Y_{j})=\rho\cdot 1_{(i=j)}$.  Then the noise stability \eqref{zero0} of $A$ is equal to $\mathbb{P}((X,Y)\in A\times A)$.

For modern proofs of Borell's theorem with additional stability statements, see \cite{mossel12,eldan13}.  In the present work, we prove a specific case of the Standard Simplex Conjecture for $k=3$, when $0<\rho<\rho_{0}(n)$.  Already for the case $k=3$, the methods used in the case $k=2$ do not seem to apply, so new techniques are required to treat the case $k=3$.  We first discuss consequences of the full conjecture and we then state the conjecture precisely.  The Standard Simplex Conjecture appears to be first stated explicitly in \cite{isaksson11}.  If true, this conjecture implies:
\begin{itemize}
\item Optimal hardness results for approximating the MAX-k-CUT problem \cite[Theorem 1.13]{isaksson11}, a generalization of the MAX-CUT problem. (These hardness results are optimal, assuming the Unique Games Conjecture).
\item The Plurality is Stablest Conjecture \cite{khot07},\cite{isaksson11}[Theorem 1.10], an extension of the Majority is Stablest Conjecture \cite{mossel10} asserting that: the most noise-stable way to determine the winner of an election between $k$ candidates is to take the plurality. (This result assumes that no one person has too much influence over the election's outcome, and each candidate has an equal probability of winning).
\item The solution of a multi-bubble problem in Gaussian space \cite{corneli08,isaksson11,ledoux96}: in $\R^{n}$, minimize the total Gaussian perimeter of $k\leq n+1$ sets of Gaussian measure $1/k$.
\end{itemize}

The MAX-k-CUT problem asks for the partition of the vertices of any graph into $k$ sets of maximum total edge perimeter.  For the precise statement, see Definition \ref{maxkcutdef} below.  For a graph on $n$ vertices, the MAX-k-CUT problem cannot be solved time polynomial in $n$, unless P=NP \cite{frieze95}.  Yet, we can always find an approximate solution of the MAX-k-CUT problem in time polynomial in $n$ \cite{frieze95}.  To create this approximate solution, we label the vertices of the graph by vectors in $\R^{n}$, solve an appropriate semidefinite program for these vectors, and we then ``round'' these vectors into $k$ bins.  In particular, two vectors are rounded into the same bin if they lie in the same subset of a given partition $\{A_{i}\}_{i=1}^{k}$ of $\R^{n}$.  The best way to perform this rounding procedure is then provided by the partition $\{A_{i}\}_{i=1}^{k}$ of optimal noise stability.  That is, the Standard Simplex Conjecture exactly describes the best way to solve the MAX-k-CUT problem \cite[Theorem A.6]{isaksson11}.  This connection between combinatorial optimization and geometry has been well-studied; see e.g. \cite{rag09,austrin10,khot09,khot11,braverman11,heilman11}.  For a survey of the complexity theoretic motivation for problems related to the Standard Simplex Conjecture, see \cite{khot12}, where Grothendieck inequalities are emphasized.

The Plurality is Stablest Conjecture for $k=2$ was proven in \cite{mossel10}, where it was found to be a consequence of Borell's theorem, after applying a nonlinear central limit theorem, which is referred to as an invariance principle.  For $k=2$, this problem is known as the Majority is Stablest Theorem.  For more on the invariance principle, see also \cite{chat06}.  The invariance principle of \cite{mossel10} is proven by combining the Lindeberg replacement method with the hypercontractive inequality \cite{gross75}.  The Plurality is Stablest Conjecture says that the Plurality function nearly maximizes discrete noise stability over all functions $f\colon\{1,\ldots,k\}^{n}\to\{1,\ldots,k\}$.  In this context, we think of the domain of $f$ as $n$ voters who vote for any one of $k$ candidates.  Given $n$ votes $(a_{1},\ldots,a_{n})\in\{1,\ldots,k\}^{n}$, the value $f(a_{1},\ldots,a_{n})\in\{1,\ldots,k\}$ is the winner of the election.  The Plurality is Stablest conjecture also assumes that each candidate has an equal probability of winning the election, and no one person has too much influence over the outcome of the election.  It turns out that the latter assumption means that the function $f$ can be well approximated by a function $g\colon\R^{n}\to\{1,\ldots,k\}$.  That is, the noise stability of $f$ is close to the sum of noise stabilities of the sets $g^{-1}(1),\ldots,g^{-1}(k)$.  This approximation procedure, which uses an invariance principle, shows the equivalence of the Plurality is Stablest Conjecture and Standard Simplex Conjecture \cite[Theorems 1.10 and 1.11]{isaksson11}.  We are therefore partially motivated to solve the Standard Simplex Conjecture to attempt to complete the picture set out by the sequence of works \cite{borell85,khot07,mossel10,isaksson11}.

The problem of minimizing Gaussian perimeter arises as an endpoint case of the Standard Simplex Conjecture.  The Standard Simplex Conjecture is a statement involving a sum of terms of the form \eqref{zero0}, and the Gaussian perimeter is recovered by letting $\rho\to 1^{-}$.

We now precisely state the Standard Simplex Conjecture.  Let $\rho\in(-1,1)$, $n\geq1$, $n\in\Z$, let $f\colon\R^{n}\to\R$ be bounded and measurable, and define $d\gamma_{n}(y)\colonequals e^{-(y_{1}^{2}+\cdots+y_{n}^{2})/2}dy/(2\pi)^{n/2}$, $y=(y_{1},\ldots,y_{n})\in\R^{n}$.  For $x\in\R^{n}$, define
\begin{equation}\label{six0}
T_{\rho}f(x)\colonequals\int_{\R^{n}}f(x\rho+y\sqrt{1-\rho^{2}})d\gamma_{n}(y).
\end{equation}
The operator defined by \eqref{six0} is known as the noise operator, or Bonami-Beckner operator, or Ornstein-Uhlenbeck operator.  In particular, the Ornstein-Uhlenbeck operator is often written with $\rho=e^{-t}$, $t>0$, so that $T_{e^{-t}}$ becomes a semigroup.

\begin{definition}\label{partdef}
Let $A_{1},\ldots,A_{k}\subset\R^{n}$ be measurable, $k\leq n+1$.  We say that $\{A_{i}\}_{i=1}^{k}$ is a \textbf{partition} of $\R^{n}$ if $\cup_{i=1}^{k}A_{i}=\R^{n}$, and $\gamma_{n}(A_{i}\cap A_{j})=0$ for $i\neq j$, $i,j\in\{1,\ldots,k\}$.  Let $\{z_{i}\}_{i=1}^{k}$ be the vertices of a regular simplex centered at the origin of $\R^{n}$.  For each $i\in\{1,\ldots,k\}$, define $A_{i}\colonequals\{x\in\R^{n}\colon \langle x,z_{i}\rangle=\max_{j\in\{1,\ldots,k\}}\langle x,z_{j}\rangle\}$, the Voronoi region of $z_{i}$.  We call $\{A_{i}\}_{i=1}^{k}$ a \textbf{regular simplicial conical partition}.
\end{definition}

\begin{conj}[\textbf{Standard Simplex Conjecture}, \cite{isaksson11}]\label{SSC}
Let $n\geq2$, let $\rho\in[-1,1]$, and let $3\leq k\leq n+1$.  Let $\{A_{i}\}_{i=1}^{k}$ be a partition of $\R^{n}$.
\begin{itemize}
\item[(a)]  If $\rho\in(0,1]$, and if $\gamma_{n}(A_{i})=1/k$, $\forall$ $i\in\{1,\ldots,k\}$, then among all such partitions of $\R^{n}$, the quantity
\begin{equation}\label{six1.5}
J\colonequals\sum_{i=1}^{k}\int_{\R^{n}}1_{A_{i}}(x)T_{\rho}(1_{A_{i}})(x)d\gamma_{n}(x)
\end{equation}
is maximized by a regular simplicial conical partition.
\item[(b)] If $\rho\in[-1,0)$ (with no restriction on the measures of the sets $A_{i}$, $i\in\{1,\ldots,k\}$), then among all partitions of $\R^{n}$, the quantity $J$ is minimized by a regular simplicial conical partition.
\end{itemize}
\end{conj}

The following theorem is our main result.

\begin{theorem}[\textbf{Main Theorem}]\label{thm0}
Fix $n\geq2$, $k=3$.  There exists $\rho_{0}=\rho_{0}(n,k)>0$ such that Conjecture \ref{SSC} holds for $\rho\in(0,\rho_{0})$.
\end{theorem}

Theorem \ref{thm0} seems to have no direct relation to Gaussian isoperimetric problems \cite{corneli08}, since these problems are implied by letting $\rho\to1^{-}$ in Conjecture \ref{SSC}.   Also, \cite[Lemma A.4,Theorem A.6]{isaksson11} shows that Theorem \ref{thm0} seems to give no new information about the MAX-k-CUT problem, since in this problem, $\rho<0$ is most relevant.  Surprisingly, our proof strategy does not work for $\rho<0$, as we show in Theorem \ref{thm3}.

Let $X=(X_{1},\ldots,X_{n}),Y=(Y_{1},\ldots,Y_{n})$ be jointly standard normal $n$-dimensional Gaussian random variables such that the covariances satisfy $\mathbb{E}(X_{i}Y_{j})=\rho\cdot1_{\{i=j\}}$, $i,j\in\{1,\ldots,n\}$.  In \cite{isaksson11}, the quantity \eqref{six1.5} is written as $\sum_{i=1}^{k}\mathbb{P}((X,Y)\in A_{i}\times A_{i})$.  To see that our formulation of Conjecture \ref{SSC} is equivalent to that of \cite{isaksson11}, let $A\subset\R^{n}$ and note that
\begin{flalign*}
\int_{\R^{n}} 1_{A}T_{\rho}1_{A}d\gamma_{n}
&=\int_{\R^{n}} 1_{A}(x)\int_{\R^{n}} 1_{A}(x\rho+y\sqrt{1-\rho^{2}})d\gamma_{n}(y)d\gamma_{n}(x)\\
&=\int_{\R^{n}}\int_{\R^{n}} 1_{A}(x)1_{A}(x\rho+y\sqrt{1-\rho^{2}})d\gamma_{n}(y)d\gamma_{n}(x)
=\mathbb{P}((X,Y)\in A\times A).
\end{flalign*}


%
%
%
%
%

\subsection{MAX-k-CUT and the Unique Games Conjecture}
We now rigorously describe the complexity theoretic notions referenced above.
\begin{definition}[\embolden{MAX-k-CUT}]\label{maxkcutdef}
Let $k,n\in\N$, $k\geq2$.  We define the weighted MAX-k-CUT problem.  We are given a symmetric matrix $\{a_{ij}\}_{i,j=1}^{n}$ with $a_{ij}\geq0$ for all $i,j\in\{1,\ldots,n\}$.  The goal of the MAX-k-CUT problem is to find the following quantity:
$$
\max_{c\colon \{1,\ldots,n\}\to\{1,\ldots,k\}}\sum_{\substack{i,j\in\{1,\ldots,n\}\colon\\ c(i)\neq c(j)}}a_{ij}.
$$
\end{definition}

\begin{definition}[\embolden{$\Gamma$-MAX-2LIN(k)}]\label{max2lindef}
Let $k\in\N$, $k\geq2$.  We define the $\Gamma$-MAX-2LIN(k) problem.  In this problem, we are given $m\in\N$ and $2m$ variables $x_{i}\in\Z/k\Z$, $i\in\{1,\ldots,2m\}$.  We are also given a matrix $\{a_{ij}\}_{i,j=1}^{2m}$ with $a_{ij}\geq0$ for all $i,j\in\{1,\ldots,2m\}$.  An element $(i,j)$ corresponds to one of $m$ linear equations of the form $x_{i}-x_{j}=c_{ij}(\mathrm{mod}\,k)$, $i,j\in\{1,\ldots,2m\}$, $c_{ij}\in\Z/k\Z$.  The goal of the $\Gamma$-MAX-2LIN(k) problem is to find the following quantity:
\begin{equation}\label{zero1}
\max_{(x_{1},\ldots,x_{2m})\in(\Z/k\Z)^{2m}}
\sum_{\substack{(i,j)\in E\colon\\ x_{i}-x_{j}=c_{ij}(\mathrm{mod}\,k)}}a_{ij}.
\end{equation}
\end{definition}

\begin{definition}[\textbf{Unique Games Conjecture}, \cite{khot07}]
For every $\epsilon\in(0,1)$, there exists a prime number $p(\epsilon)$ such that no polynomial time algorithm can distinguish between the following two cases, for instances of $\Gamma$-MAX-2LIN($p(\epsilon)$) with $w=1$:
\begin{itemize}
\item[(i)] \eqref{zero1} is larger than $(1-\epsilon)m$, or
\item[(ii)] \eqref{zero1} is smaller than $\epsilon m$.
\end{itemize}
\end{definition}

If \eqref{zero1} were equal to $m$, then we could find $(x_{1},\ldots,x_{2m})$ achieving the maximum in \eqref{zero1} by linear algebra.  One can therefore interpret the Unique Games Conjecture as an assertion that approximate linear algebra is hard.

\begin{theorem}\label{thm9}
\textnormal{(\textbf{Optimal Approximation for MAX-k-CUT}, \cite{isaksson11}[Theorem 1.13],\cite{frieze95}).}
Let $k\in\N$, $k\geq2$.  Let $\{A_{i}\}_{i=1}^{k}\subset\R^{k-1}$ be a regular simplicial conical partition.  Define
$$
\alpha_{k}\colonequals
\inf_{-\frac{1}{k-1}\leq\rho\leq1}\,
\frac{k-k^{2}\sum_{i=1}^{k}\int_{\R^{n}}1_{A_{i}}T_{\rho}1_{A_{i}}d\gamma_{n}}{(k-1)(1-\rho)}
=\inf_{-\frac{1}{k-1}\leq\rho\leq0}\,
\frac{k-k^{2}\sum_{i=1}^{k}\int_{\R^{n}}1_{A_{i}}T_{\rho}1_{A_{i}}d\gamma_{n}}{(k-1)(1-\rho)}.
$$
Assume Conjecture \ref{SSC} and the Unique Games Conjecture.  Then, for any $\epsilon>0$, there exists a polynomial time algorithm that approximates MAX-k-CUT within a multiplicative factor $\alpha_{k}-\epsilon$, and it is NP-hard to approximate MAX-k-CUT within a multiplicative factor of $\alpha_{k}+\epsilon$.
\end{theorem}

\subsection{Plurality is Stablest}

We now briefly describe the Plurality is Stablest Conjecture.  This Conjecture seems to first appear in \cite{khot07}.  The work \cite{khot07} emphasizes the applications of this conjecture to MAX-k-CUT and to MAX-2LIN(k).

Let $n\geq2,k\geq3$
Let $(W_{1},\ldots,W_{k})$ be an orthonormal basis for the space of functions $\{g\colon\{1,\ldots,k\}\to[0,1]\}$ equipped with the inner product $\langle g,h\rangle_{k}\colonequals\frac{1}{k}\sum_{\sigma\in\{1,\ldots,k\}}g(\sigma)h(\sigma)$.  Assume that $W_{1}=1$.  By orthonormality, for every $\sigma\in\{1,\ldots,k\}$, there exists $\widehat{g}(\sigma)\in\R$ such that the following expression holds: $g=\sum_{\sigma\in\{1,\ldots,k\}}\widehat{g}(\sigma)W_{\sigma}$.  Define
$$\textstyle\Delta_{k}\colonequals\{(x_{1},\ldots,x_{k})\in\R^{k}\colon\forall\,1\leq i\leq k, 0\leq x_{i}\leq1,\sum_{i=1}^{k}x_{i}=1\}.$$
Let $f\colon\{1,\ldots,k\}^{n}\to\Delta_{k}$, $f=(f_{1},\ldots,f_{k})$, $f_{i}\colon\{1,\ldots,k\}^{n}\to[0,1]$, $i\in\{1,\ldots,k\}$.  Let $\sigma=(\sigma_{1},\ldots,\sigma_{n})\in\{1,\ldots,k\}^{n}$.  Define $W_{\sigma}\colonequals\prod_{i=1}^{n}W_{\sigma_{i}}$, and let $\abs{\sigma}\colonequals\abs{\{i\in\{1,\ldots,n\}\colon\sigma_{i}\neq1\}}$.  Then for every $\sigma\in\{1,\ldots,k\}^{n}$ there exists $\widehat{f_{i}}(\sigma)\in\R$ such that $f_{i}=\sum_{\sigma\in\{1,\ldots,k\}^{n}}\widehat{f_{i}}(\sigma)W_{\sigma}$, $i\in\{1,\ldots,k\}$.  For $\rho\in[-1,1]$ and $i\in\{1,\ldots,k\}$, define
$$
T_{\rho}f_{i}\colonequals\sum_{\sigma\in\{1,\ldots,k\}^{n}}\rho^{\abs{\sigma}}\widehat{f_{i}}(\sigma)W_{\sigma},\quad
T_{\rho}f\colonequals(T_{\rho}f_{1},\ldots,T_{\rho}f_{k})\in\R^{k}.
$$

Let $m\geq2$, $k\geq3$.  For each $j\in\{1,\ldots,k\}$, let $e_{j}=(0,\ldots,0,1,0,\ldots,0)\in\R^{k}$ be the $j^{th}$ unit coordinate vector.  Let $\sigma\in\{1,\ldots,k\}^{n}$.  Define $\mathrm{PLUR}_{m,k}\colon\{1,\ldots,k\}^{m}\to\Delta_{k}$ such that
$$\mathrm{PLUR}_{m,k}(\sigma)
\colonequals\begin{cases}
e_{j}&,\mbox{if }\abs{\{i\in\{1,\ldots,m\}\colon\sigma_{i}=j\}}>
\abs{\{i\in\{1,\ldots,m\}\colon\sigma_{i}=r\}},\\
&\qquad\qquad\forall\,r\in\{1,\ldots,k\}\setminus\{j\}\\
\frac{1}{k}\sum_{i=1}^{k}e_{i}&,\mbox{otherwise}
\end{cases}
$$

\begin{conj}[\textbf{Plurality is Stablest Conjecture}, \cite{isaksson11}]\label{PS}
Let $n\geq2$, $k\geq3$, $\rho\in[-\frac{1}{k-1},1]$, $\epsilon>0$.  Let $\langle\cdot,\cdot\rangle$ denote the standard inner product on $\R^{n}$.  Then there exists $\tau>0$ such that, if $f\colon\{1,\ldots,k\}^{n}\to\Delta_{k}$ satisfies $\sum_{\sigma\in\{1,\ldots,k\}^{n}\colon\sigma_{j}\neq1}(\widehat{f_{i}}(\sigma))^{2}\leq\tau$ for all $i\in\{1,\ldots,k\}$, $j\in\{1,\ldots,n\}$, then
\begin{itemize}
\item[(a)]  If $\rho\in(0,1]$, and if $\frac{1}{k^{n}}\sum_{\sigma\in\{1,\ldots,k\}^{n}}f(\sigma)=\frac{1}{k}\sum_{i=1}^{k}e_{i}$, then
$$\frac{1}{k^{n}}\sum_{\sigma\in\{1,\ldots,k\}^{n}}\langle f(\sigma),T_{\rho}f(\sigma)\rangle\leq
\lim_{m\to\infty}\frac{1}{k^{m}}\sum_{\sigma\in\{1,\ldots,k\}^{m}}\langle\mathrm{PLUR}_{m,k}(\sigma), T_{\rho}(\mathrm{PLUR}_{m,k})(\sigma)\rangle
+\epsilon.
$$
\item[(b)] If $\rho\in[-1/(k-1),0)$, then
$$\frac{1}{k^{n}}\sum_{\sigma\in\{1,\ldots,k\}^{n}}\langle f(\sigma),T_{\rho}f(\sigma)\rangle\geq
\lim_{m\to\infty}\frac{1}{k^{m}}\sum_{\sigma\in\{1,\ldots,k\}^{m}}\langle\mathrm{PLUR}_{m,k}(\sigma), T_{\rho}(\mathrm{PLUR}_{m,k})(\sigma)\rangle
-\epsilon.
$$
\end{itemize}
\end{conj}

\subsection{A Synopsis of the Main Theorem}

We now describe the proof of Theorem \ref{thm0}.  We first take the derivative $d/d\rho$ of the quantity $J$ defined by \eqref{six1.5}.  This procedure is common, and it dates back at least to the the proof of the Log-Sobolev Inequality by Gross \cite{gross75}.  Taking this derivative allows us to relate $J$ to the works \cite{khot09,khot11}.  In Section \ref{secvar}, we modify the results of \cite{khot09,khot11} to prove the existence of a partition that maximizes $(d/d\rho)J$.  Then, in Section \ref{secpert}, we further modify results of \cite{khot09,khot11} to show that, if $\rho>0$ is small, then a partition maximizing $(d/d\rho)J$ is close to a partition maximizing $(d/d\rho)|_{\rho=0}J$.  And by \cite{khot09}, we know that the partition maximizing $(d/d\rho)|_{\rho=0}J$ is a regular simplicial conical partition, for dimension $n\geq2$ and $k=3$ partition elements.

So, for small $\rho>0$, a partition maximizing $(d/d\rho)J$ is close to a regular simplicial conical partition.  The structure of the operator $T_{\rho}$ then permits the exploitation of a feedback loop.  This feedback loop says: if our partition maximizes $(d/d\rho)J$ for small $\rho>0$, and if this partition is close to a regular simplicial conical partition, then this partition is even closer to a regular simplicial conical partition.  This feedback loop is investigated in Section \ref{seciter}, especially in the crucial Lemma \ref{lemma8}.  A similar feedback loop was already apparent in \cite{khot09}[Lemma 3.3].  The full argument of Theorem \ref{thm0} is then assembled in Section \ref{secmain}.  By using this feedback loop, we show in Theorem \ref{thm1} that a regular simplicial conical partition maximizes $(d/d\rho)J$ for small $\rho>0$, $k=3$, $n\geq2$.  Then, the Fundamental Theorem of Calculus allows us to relate $(d/d\rho)J$ to $J$, therefore completing the proof of the main theorem, Theorem \ref{thm0}.

Since Lemma \ref{lemma8} is rather lengthy and crucial to this investigation, we will further describe the idea behind it.  If we know that our partition maximizes $(d/d\rho)J$, and if we also know that this partition is close to a regular simplicial conical partition, then the first variation should immediately tell us that our partition is actually a regular simplicial conical partition.  Unfortunately, this intuition does not translate into a proof.  The main technical problem is that the sets we are dealing with are unbounded, and we need to know precise information about the Ornstein-Uhlenbeck operator applied to these sets, for points that are very far from the origin.  Since the Gaussian measure decays exponentially away from the origin, this means that it becomes hard to say something precise about the points in these sets that are very far from the origin.  So, we require very precise estimates of the Ornstein-Uhlenbeck operator, and the errors that it accrues when we evaluate it far from the origin.  These estimates are performed in Lemmas \ref{lemma6.1} and \ref{lemma7}.  Unfortunately, to use these estimates effectively, we need to slowly enlarge the regions where we use these estimates.  The details of enlarging these regions becomes surprisingly complicated, occupying the seven steps of Lemma \ref{lemma8}.

In Section \ref{secmain}, we also show the surprising fact that our strategy fails for small negative correlation.  That is, for small $\rho<0$, $(d/d\rho)J$ is not maximized by the regular simplicial conical partition.  This result does not confirm or deny Conjecture \ref{SSC} for $\rho<0$.  However, one may interpret from this result that the case of Conjecture \ref{SSC} for $\rho<0$ could be more difficult than the case $\rho>0$.

We should also emphasize the lack of symmetrization in the proof of Theorem \ref{thm0}.  Symmetrization is one of a few general strategies that solves many optimization problems.  In our context, symmetrization would appear as follows.  Recall the definition of $J$ from \eqref{six1.5}.  Suppose we have a partition $\{A_{i}\}_{i=1}^{k}\subset\R^{n}$.  Change this partition into a ``more symmetric'' partition $\{\widetilde{A_{i}}\}_{i=1}^{k}$ such that $J$ or $(d/d\rho)J$ is larger for $\{\widetilde{A_{i}}\}_{i=1}^{k}$.  In the proof of the main theorem, it is tempting to use this symmetrization paradigm.  The works \cite{borell85},\cite{mossel10} and \cite{isaksson11} use Gaussian symmetrization in a crucial way.  However, we find this approach to be less natural for Conjecture \ref{SSC}, so we do not explicitly use symmetrization.  Nevertheless, symmetry does play a crucial role in our proof, especially in the estimates of Section \ref{secpert}.  It should also be noted that the works \cite{khot09,khot11} do not explicitly use symmetrization, and this lack of symmetrization is one of their novel aspects.

\subsection{Preliminaries}

We follow the exposition of \cite{ledoux96}.  Let $n\geq1$, $n\in\Z$.  Let $\N=\{0,1,2,3,\ldots\}$.  For $f\colon\R^{n}\to\R$ measurable, let $\vnorm{f}_{L_{2}(\gamma_{n})}\colonequals(\int_{\R^{n}}\abs{f}^{2}d\gamma_{n})^{1/2}$.  Let $L_{2}(\gamma_{n})\colonequals\{f\colon\R^{n}\to\R\colon\vnorm{f}_{L_{2}(\gamma_{n})}<\infty\}$.  Let $\ell_{2}^{n}$ denote the $\ell_{2}$ norm on $\R^{n}$.  For $x\in\R^{n}$ and $r>0$, define $B(x,r)\colonequals\{y\in\R^{n}\colon\vnorm{x-y}_{\ell_{2}^{n}}<r\}$.

For $f\in L_{2}(\gamma_{n})$, define $T_{\rho}$ as in \eqref{six0}.  The operator $T_{\rho}$ is a parametrization of the Ornstein-Uhlenbeck operator.  The operator $T_{\rho}$ is not a semigroup, but it satisfies $T_{\rho_{1}}T_{\rho_{2}}=T_{\rho_{1}\rho_{2}}$, $\rho_{1},\rho_{2}\in[-1,1]$, by \eqref{six1.8} below.  We use definition \eqref{six0} since the usual Ornstein-Uhlenbeck operator is only defined for $\rho\in[0,1]$.  Let $\lambda>0$, $x\in\R$.  Recall that the Hermite polynomials of one variable are defined by the generating function
\begin{equation}\label{six1.9}
e^{\lambda x-\lambda^{2}/2}\equalscolon\sum_{\ell\in\N}\lambda^{\ell}h_{\ell}(x).
\end{equation}

Alternatively, one defines the polynomials $H_{\ell}(x)$ such that $h_{\ell}(x)=2^{-\ell/2}(\ell!)^{-1}H_{\ell}(x/\sqrt{2})$.  This convention is used in \cite{andrews99}, where the orthogonality properties of the Hermite polynomials are derived.

Note that $\int_{\R}h_{\ell}^{2}d\gamma_{1}=1/\ell!$, and $\{\sqrt{\ell!}\,h_{\ell}\}_{\ell\in\N}$ is an orthonormal basis of $L_{2}(\gamma_{1})$.  Recall also that $h_{0}(x)=1$ and $h_{1}(x)=x$.  Set $f(x)\colonequals e^{\lambda x-\lambda^{2}/2}$.  A routine computation shows that $T_{\rho}(f)(x)=e^{(\lambda\rho)x-(\lambda\rho)^{2}/2}$.  Indeed
\begin{flalign*}  
T_{\rho}(f)(x)
&= \int_{\R^{n}} e^{\lambda(x\rho+y\sqrt{1-\rho^{2}})-\lambda^{2}/2}d\gamma_{1}(y)
=\int_{\R^{n}} e^{(\lambda\rho)x+(\lambda\sqrt{1-\rho^{2}})y-\lambda^{2}/2-y^{2}/2}\frac{dy}{\sqrt{2\pi}}\\
&=e^{(\lambda\rho)x-\lambda^{2}/2}\int_{\R^{n}} e^{-\frac{1}{2}(y-\lambda\sqrt{1-\rho^{2}})^{2}+\lambda^{2}(1-\rho^{2})/2}\frac{dy}{\sqrt{2\pi}}
=e^{(\lambda\rho)x-\lambda^{2}/2+\lambda^{2}(1-\rho^{2})/2}\\
&=e^{(\lambda\rho)x-(\lambda\rho)^{2}/2}.
\end{flalign*}
Therefore, by \eqref{six1.9},
\begin{equation}\label{six1}
T_{\rho}f(x)=\sum_{\ell\in\N}\lambda^{\ell}\rho^{\ell}h_{\ell}(x).
\end{equation}
So, by linearity, $T_{\rho}h_{\ell}(x)=\rho^{\ell}h_{\ell}(x)$.

We now extend the above observations to higher dimensions.  Let $f\in L_{2}(\gamma_{n})$, so that $f=\sum_{\ell\in\N^{n}}a_{\ell}h_{\ell}\sqrt{\ell!}$, $a_{\ell}\in\R$, where $\ell=(\ell_{1},\ldots,\ell_{n})\in\N^{n}$ and $h_{\ell}(x)=\prod_{i=1}^{n}h_{\ell_{i}}(x_{i})$.  Write $\abs{\ell}\colonequals \ell_{1}+\cdots+\ell_{n}$ and $\ell!\colonequals (\ell_{1}!)\cdots(\ell_{n}!)$.  Then $T_{\rho}$ satisfies $T_{\rho}h_{\ell}=\rho^{\abs{\ell}}h_{\ell}$ and for $x\in\R^{n}$,
\begin{equation}\label{six1.8}
T_{\rho}f(x)=\sum_{\ell\in\N^{n}}\rho^{\abs{\ell}}\sqrt{\ell!}\,h_{\ell}(x)\left(\int_{\R^{n}}\sqrt{\ell!}\,h_{\ell}fd\gamma_{n}\right)
\end{equation}

Let $\Delta\colonequals\sum_{i=1}^{n}\partial^{2}/\partial x_{i}^{2}$, and define
\begin{equation}\label{six1.88}
L\colonequals-\Delta+\sum_{i=1}^{n}x_{i}\cdot\frac{\partial}{\partial x_{i}}.
\end{equation}
A well-known calculation shows the following equality, which we prove in the Appendix, Section \ref{secapp}.
\begin{flalign}
\label{six1.1}\frac{d}{d\rho}T_{\rho}f(x)
&=\rho^{-1}LT_{\rho}f(x)=\frac{1}{\rho}\left(\langle x,\nabla T_{\rho}f(x)\rangle-\Delta T_{\rho}f(x)\right)\\
&=\frac{1}{\sqrt{1-\rho^{2}}}\bigg[\left\langle x,\int_{\R^{n}}yf(x\rho+y\sqrt{1-\rho^{2}})d\gamma_{n}(y)\right\rangle\nonumber\\
&\label{six1.2}\qquad\qquad+\frac{\rho}{\sqrt{1-\rho^{2}}}
\int_{\R^{n}}\left(\sum_{i=1}^{n}(1-y_{i}^{2})\right)f(x\rho+y\sqrt{1-\rho^{2}})d\gamma_{n}(y)\bigg].
\end{flalign}

We say that $A\subset\R^{n}$ is a \textbf{cone} if $A$ is measurable and $\forall$ $t>0$, $t A=A$.

\begin{definition}\label{regulardef2}
A \textbf{simplicial conical partition} $\{A_{i}\}_{i=1}^{k}$ is a partition of $\R^{n}$ together with a simplex $S\subset\R^{k-1}$ with $0\leq k-1\leq n$ and a rotation $\sigma$ of $\R^{n}$ such that $0\in S$ and such that each facet $F_{i}$ of $\sigma(S\times\R^{n-k+1})$ generates a partition element, i.e. $A_{i}=\{t\sigma(F_{i}\times\R^{n-j})\colon t\in[0,\infty)\}$, $i\in\{1,\ldots,j+1\}$.  Let $k-1\leq n$ and let $\{z_{i}\}_{i=1}^{k}\subset\R^{n}$ be nonzero vectors that do not all lie in a $(k-1)$-dimensional hyperplane.  Define a partition such that, for $i\in\{1,\ldots,k\}$, $A_{i}\colonequals\{x\in \R^{n}\colon \langle x,z_{i}\rangle=\max_{j=1,\ldots,k}\langle x,z_{j}\rangle\}$.  Such a partition is called the simplicial conical partition induced by $\{z_{i}\}_{i=1}^{k}$.

If $\{A_{i}\}_{i=1}^{k}$ is a simplicial conical partition induced by the vectors $\{\int_{A_{i}}xd\gamma_{n}(x)\}_{i=1}^{k}$, then we say the partition is a \textbf{balanced conical partition}.  If $\{z_{i}\}_{i=1}^{k}\subset\R^{n}$ are the vertices of a $(k-1)$-dimensional regular simplex in $\R^{n}$ centered at the origin, then the partition induced by $\{z_{i}\}_{i=1}^{k}$ is called a \textbf{regular simplicial conical partition}.
\end{definition}

Let $f\in L_{2}(\gamma_{n})$.  By Plancherel and \eqref{six1.8}
\begin{equation}\label{six1.6}
\int_{\R^{n}} f T_{\rho}fd\gamma_{n}=\sum_{\ell\in\N^{n}}\rho^{\abs{\ell}}\abs{\int_{\R^{n}} f\sqrt{\ell!}h_{\ell}d\gamma_{n}}^{2}.
\end{equation}
Substituting \eqref{six1.6} into \eqref{six1.5} gives
\begin{equation}\label{six2}
\sum_{i=1}^{k}\int_{\R^{n}}1_{A_{i}}T_{\rho}(1_{A_{i}})d\gamma_{n}
=\sum_{i=1}^{k}\bigg[\gamma_{n}(A_{i})^{2}
+\rho\vnorm{\int_{A_{i}}xd\gamma_{n}(x)}_{\ell_{2}^{n}}^{2}
+\sum_{\substack{\,\ell\in\N^{n}\\\abs{\ell}\geq2}}\rho^{\abs{\ell}}\abs{\int_{A_{i}}\sqrt{\ell!}h_{\ell}d\gamma_{n}}^{2}\bigg].
\end{equation}
Taking the derivative $d/d\rho$ of \eqref{six2} at $\rho=0$, we get a quantity studied in \cite{khot09,khot11}.
\begin{equation}\label{six3.0}
\frac{d}{d\rho}\sum_{i=1}^{k}\int_{\R^{n}}1_{A_{i}}T_{\rho}(1_{A_{i}})d\gamma_{n}
=\sum_{i=1}^{k}\bigg[
\vnorm{\int_{A_{i}}xd\gamma_{n}(x)}_{\ell_{2}^{n}}^{2}
+\sum_{\substack{\,\ell\in\N^{n}\\\abs{\ell}\geq2}}\abs{\ell}\rho^{\abs{\ell}-1}\abs{\int_{A_{i}}\sqrt{\ell!}h_{\ell}d\gamma_{n}}^{2}\bigg].
\end{equation}
\begin{equation}\label{six3}
\left.\frac{d}{d\rho}\right|_{\rho=0}\sum_{i=1}^{k}\int_{\R^{n}}1_{A_{i}}T_{\rho}(1_{A_{i}})d\gamma_{n}
=\sum_{i=1}^{k}\vnorm{\int_{A_{i}}xd\gamma_{n}(x)}_{\ell_{2}^{n}}^{2}.
\end{equation}

\section{Noise Stability for Zero Correlation}

This section concerns noise stability at the endpoint $\rho=0$.  Specifically, we will investigate the quantity \eqref{six3}, which has already been studied in \cite{khot09,khot11}.  Using our understanding of \eqref{six3}, we will then be able to analyze the left side of \eqref{six3.0} when $\rho$ is small, using the equality \eqref{six3.0}.  Before beginning our discussion, we first need to consider partitions of $\R^{n}$ within the convex set defined in \eqref{blop}.  Definition \ref{d2def} provides the metric allowing an assertion that two partitions are close to each other, and Definition \ref{measuredef} allows us to discuss the Gaussian measure restricted to hypersurfaces.

\begin{definition}\label{Deltadef}
Let $H\colonequals\oplus_{i=1}^{k}L_{2}(\gamma_{n})$ and define
\begin{equation}\label{blop}
\Delta_{k}(\gamma_{n})\colonequals\{(f_{1},\ldots,f_{k})\in H\colon \forall\,1\leq i\leq k,0\leq f_{i}\leq 1,\textstyle\sum_{i=1}^{k}f_{i}=1\}.
\end{equation}
\end{definition}

\begin{definition}\label{dkedef}
Let $\epsilon\geq0$.  Define
$$
\Delta_{k}^{\epsilon}(\gamma_{n})\colonequals\{(f_{1},\ldots,f_{k})\in \Delta_{k}(\gamma_{n})\colon\frac{1}{k}-\epsilon\leq\int_{\R^{n}} f_{i}d\gamma_{n}\leq\frac{1}{k}+\epsilon\}.
$$
\end{definition}

\begin{definition}\label{d2def}
Define a metric $d_{2}$ on partitions $\{A_{i}\}_{i=1}^{k},\{C_{i}\}_{i=1}^{k}$ of $\R^{n}$ by the formula
$$
d_{2}(\{A_{i}\}_{i=1}^{k},\{C_{i}\}_{i=1}^{k})\colonequals\inf_{\substack{\sigma\in SO(n)\\\pi\,\mathrm{a}\,\mathrm{permutation}}}
\left(\sum_{i=1}^{k}\vnorm{1_{A_{i}}-1_{(\sigma C_{\pi(i)})}}_{L_{2}(\gamma_{n})}^{2}\right)^{1/2}.
$$
\end{definition}

\begin{definition}\label{measuredef}
Let $A\subset\R^{n}$, let $\mathcal{L}$ denote Lebesgue measure on $\R^{n}$, and define the distance $d(x,y)\colonequals\vnorm{x-y}_{\ell_{2}^{n}}$, $x,y\in\R^{n}$.  Denote $\delta_{A}$ as the measure $\mathcal{L}$ on $\R^{n}$ restricted to $A$.  That is, $\delta_{A}(B)\colonequals\liminf_{\delta\to0}\frac{1}{2\delta}\mathcal{L}\{y\in\R^{n}\colon\exists\,x\in A\cap B\,\,\mbox{with}\,\,d(x,y)<\delta\}$, $B\subset\R^{n}$.  Also, we denote $\gamma_{n}(\delta_{A})\colonequals\liminf_{\delta\to0}\frac{1}{2\delta}\gamma_{n}\{y\in\R^{n}\colon \exists\, x\in A\,\,\mbox{with}\,\, d(x,y)<\delta\}$.
\end{definition}

The next two lemmas are derived from \cite{khot09}.  Lemma \ref{lemma0.5} is a quantitative variant of Lemma \ref{lemma0}, and it will be further improved in Lemma \ref{lemma5} below.  In particular, Lemma \ref{lemma0.5} says that, if the first variation condition for achieving the optimum value of \eqref{flop} is nearly satisfied, then the partition is close to being simplicial.
\begin{lemma}\label{lemma0}
\cite[Lemma 3.3, Corollary 3.4]{khot09}  Let $n\geq2$ and let $\{B_{i}\}_{i=1}^{3}$ be a regular simplicial conical partition of $\R^{n}$.  Then $(1_{B_{1}},1_{B_{2}},1_{B_{3}})$ uniquely achieves the following supremum, up to orthogonal transformation.
\begin{equation}\label{flop}
\sup_{(f_{1},f_{2},f_{3})\in\Delta_{3}(\gamma_{n})}\sum_{i=1}^{3}\vnorm{\int_{\R^{n}}xf_{i}(x)d\gamma_{n}(x)}_{\ell_{2}^{n}}^{2}.
\end{equation}
\end{lemma}
\begin{lemma}\label{lemma0.5}
Let $n\geq2$ and let $\{B_{i}\}_{i=1}^{3},\{C_{i}\}_{i=1}^{2}\subset\R^{n}$ be regular simplicial conical partition.  Let $\{A_{i}\}_{i=1}^{3}\subset\R^{n}$ be a simplicial conical partition.  Let $z_{i}\colonequals\int_{A_{i}}xd\gamma_{n}(x)$, and let $v_{ij}\in S^{n-1}\cap A_{i}\cap A_{j}\cap\mathrm{span}\{z_{i},z_{j}\}$.  If $\abs{\langle z_{i}-z_{j},v_{ij}\rangle}\leq\epsilon<10^{-16}$ $\forall$ $i,j\in\{1,2,3\}$, $i\neq j$, and if $d_{2}(\{A_{i}\}_{i=1}^{3},\{C_{1},C_{2},\emptyset\})>1/100$, then $d_{2}(\{A_{i}\}_{i=1}^{3},\{B_{i}\}_{i=1}^{3})\leq\sqrt{6\epsilon}$.
\end{lemma}
\begin{proof}
For $i,j\in\{1,2,3\}$, let $0\leq\alpha_{i}\leq\pi$ such that $A_{i}$ is a cone with angle $\alpha_{i}$.  Let $\sigma\colon\R^{n}\to\R^{n}$ be a reflection that fixes $A_{i}\cap A_{j}$.  Without loss of generality, $\sigma(A_{j})\subset A_{i}$.  Then $z_{i}-z_{j}=\int_{A_{i}\setminus\sigma(A_{j})}xd\gamma_{n}(x)$ and $\vnorm{z_{i}-z_{j}}_{2}=\sin((\alpha_{i}-\alpha_{j})/2)/\sqrt{2\pi}$.  Let $0\leq\theta\leq\pi$ such that $\vnorm{z_{i}-z_{j}}_{2}\cos(\theta)=\langle z_{i}-z_{j},v_{ij}\rangle$.  Then either $\vnorm{z_{i}-z_{j}}_{2}\leq\sqrt{\epsilon/18\pi}$, or $\abs{\cos\theta}\leq\sqrt{18\pi\epsilon}$.  In the first case, $\alpha_{i}-\alpha_{j}\leq\sqrt{\epsilon}$.  So, to complete the proof, it suffices to show that the second case does not occur.  We find a contradiction by assuming that the second case occurs.

If $\abs{\cos\theta}\leq\sqrt{18\pi\epsilon}$, then since $\theta=(\alpha_{i}-\alpha_{j})/2$, we must have $\abs{\alpha_{i}-\alpha_{j}-\pi}<18\sqrt{\epsilon}$, so $\pi-18\sqrt{\epsilon}<\alpha_{i}-\alpha_{j}<\pi+18\sqrt{\epsilon}$, i.e. $\pi-18\sqrt{\epsilon}<\alpha_{i}\leq\pi$ and $\alpha_{j}\leq18\sqrt{\epsilon}$.  Then for $r\neq i,j$, $r\in\{1,2,3\}$, we have $\alpha_{r}=2\pi-\alpha_{i}-\alpha_{j}>2\pi-\pi-18\sqrt{\epsilon}>\pi-18\sqrt{\epsilon}$.  Since $\alpha_{i},\alpha_{j}>\pi-18\sqrt{\epsilon}$, we conclude that $d_{2}(\{A_{i}\}_{i=1}^{3},\{C_{1},C_{2},\emptyset\})<18\epsilon^{1/4}<1/100$, a contradiction.
\end{proof}

We require the ensuing explicit calculation from \cite{khot09} in Lemma \ref{lemma5} below.  This calculation is reduced to a computation of Lagrange Multipliers in \cite[Corollary 3.4]{khot09}.  For any $(f_{1},\ldots,f_{k})\in\Delta_{k}(\gamma_{n})$, define $\psi_{0}(f_{1},\ldots,f_{k})\colonequals\sum_{i=1}^{k}\vnorm{\int_{\R^{n}}xf_{i}(x)d\gamma_{n}(x)}_{\ell_{2}^{n}}^{2}$.
\begin{lemma}\label{lemma5.3}\cite[Corollary 3.4]{khot09}
$$
\sup_{(f_{1},f_{2})\in\Delta_{2}(\gamma_{n})}\psi_{0}(f_{1},f_{2})=\frac{1}{\pi},\quad
\sup_{(f_{1},f_{2},f_{3})\in\Delta_{3}(\gamma_{n})}\psi_{0}(f_{1},f_{2},f_{3})=\frac{9}{8\pi}.
$$
\end{lemma}

The following Lemma is a quantitative improvement of Lemmas \ref{lemma0} and \ref{lemma0.5}.  Combining Lemma \ref{lemma5} with \eqref{six3} will show that an optimizer of $(d/d\rho)\sum_{i=1}^{k}\int_{\R^{n}}1_{A_{i}}T_{\rho}1_{A_{i}}d\gamma_{n}$ is almost simplicial conical for small $\rho>0$.

\begin{lemma}\label{lemma5}
Let $\epsilon>0$, $n\geq2$.  Let $\{A_{i}\}_{i=1}^{3}$ be a partition of $\R^{n}$, and let $\{B_{i}\}_{i=1}^{3}$ be a regular simplicial conical partition of $\R^{n}$.  Assume that $\epsilon<1/100$ and
\begin{equation}\label{three20}
\psi_{0}(1_{A_{1}},1_{A_{2}},1_{A_{3}})>\sup_{(f_{1},f_{2},f_{3})\in\Delta_{3}(\gamma_{n})}\psi_{0}(f_{1},f_{2},f_{3})-\epsilon.
\end{equation}
Then
\begin{equation}\label{three24}
d_{2}(\{A_{i}\}_{i=1}^{3},\{B_{i}\}_{i=1}^{3})\leq6\epsilon^{1/8}.
\end{equation}
\end{lemma}
\begin{proof}
Assume that \eqref{three20} holds.  For $i\in\{1,2,3\}$, let $z_{i}\colonequals\int_{A_{i}}xd\gamma_{n}(x)$, $w_{i}\colonequals\int_{B_{i}}xd\gamma_{n}(x)$.  We may assume that, for all $i,j\in\{1,2,3\}$ with $i\neq j$, $\langle z_{i},z_{j}\rangle<0$.  To see this, we argue by contradiction.  Suppose there exist $i,j\in\{1,2,3\}$, $i\neq j$ with $\langle z_{i},z_{j}\rangle\geq0$.  For $p\in\{1,2,3\}$, $p\neq i,j$, let $A_{p}''\colonequals A_{p}$, let $A_{i}''\colonequals A_{i}\cup A_{j}$, and let $A_{j}''\colonequals\emptyset$.  For $p\in\{1,2,3\}$, let $z_{p}''\colonequals\int_{A_{p}''}xd\gamma_{n}(x)$.  Then
$$
\sum_{p=1}^{3}\vnorm{z_{p}''}_{\ell_{2}^{n}}^{2}-\sum_{p=1}^{3}\vnorm{z_{p}}_{\ell_{2}^{n}}^{2}
=\vnorm{z_{i}+z_{j}}_{\ell_{2}^{n}}^{2}-\vnorm{z_{i}}_{\ell_{2}^{n}}^{2}-\vnorm{z_{j}}_{\ell_{2}^{n}}^{2}
\geq0.
$$
Rewriting this inequality using the definition of $\psi_{0}$,
\begin{equation}\label{three25}
\psi_{0}(1_{A_{1}},1_{A_{2}},1_{A_{3}})\leq\psi_{0}(1_{A_{1}''},1_{A_{2}''},1_{A_{3}''}).
\end{equation}
Since $\{A_{p}''\}_{p=1}^{3}$ is a partition of $\R^{n}$ with at most two nonempty elements, Lemma \ref{lemma5.3} says
\begin{equation}\label{three26}
\bigg(\sup_{(f_{1},f_{2},f_{3})\in\Delta_{3}(\gamma_{n})}\psi_{0}(f_{1},f_{2},f_{3})\bigg)-\psi_{0}(1_{A_{1}''},1_{A_{2}''},1_{A_{3}''})
\geq\frac{1}{8\pi}>10^{-2}.
\end{equation}
Combining \eqref{three25} and \eqref{three26} contradicts \eqref{three20}.  Therefore, $\langle z_{i},z_{j}\rangle<0$ for all $i,j\in\{1,2,3\}$.

We now claim that, for each pair $i,j\in\{1,2,3\}$ with $i\neq j$, we have
\begin{equation}\label{three27}
\max_{p\in\{i,j\}}\vnorm{z_{p}}_{\ell_{2}^{n}}^{2}\geq1/16.
\end{equation}
We again argue by contradiction.  Suppose there exist $i,j\in\{1,2,3\}$ with $i\neq j$ and $\max_{p\in\{i,j\}}\vnorm{z_{p}}_{\ell_{2}^{n}}^{2}<1/16$.  Let $p\in\{1,2,3\}$, $p\neq i,j$.  Then $\vnorm{z_{p}}_{\ell_{2}^{n}}^{2}\leq1/(2\pi)$ with equality if and only if $1_{A_{p}}$ is a half-space whose boundary contains the origin of $\R^{n}$.  This follows immediately from rearrangement.  Observe, if $z_{p}\neq0$,
\begin{flalign*}
\vnorm{z_{p}}_{\ell_{2}^{n}}^{2}
&=\bigg\langle z_{p},\int_{A_{p}}xd\gamma_{n}(x)\bigg\rangle
\leq\bigg\langle z_{p},\int_{A_{p}\cap\{x\colon\langle x,z_{p}\rangle\geq0\}}xd\gamma_{n}(x)\bigg\rangle
\leq\bigg\langle z_{p},\int_{\{x\colon\langle x,z_{p}\rangle\geq0\}}xd\gamma_{n}(x)\bigg\rangle\\
&=\bigg\langle\int_{A_{p}}xd\gamma_{n}(x),\int_{\{x\colon\langle x,z_{p}\rangle\geq0\}}xd\gamma_{n}(x)\bigg\rangle
\leq\bigg|\bigg|\int_{\{x\colon\langle x,z_{p}\rangle\geq0\}}xd\gamma_{n}(x)\bigg|\bigg|_{\ell_{2}^{n}}^{2}=\frac{1}{2\pi}.
\end{flalign*}

Therefore,
$$\psi_{0}(1_{A_{1}},1_{A_{2}},1_{A_{3}})\leq1/8+1/(2\pi)\leq1/\pi.$$
This inequality contradicts \eqref{three20} as in \eqref{three26}, since $\sup_{(f_{1},f_{2},f_{3})\in\Delta_{k}(\gamma_{n})}\psi_{0}(f_{1},\ldots,f_{k})=9/(8\pi)$, using Lemma \ref{lemma5.3}.  We conclude that \eqref{three27} holds.

Define $\delta$ such that
\begin{equation}\label{three50}
\delta\colonequals\max_{i,j\in\{1,2,3\},i\neq j}\gamma_{n}(\{x\in\R^{n}\colon\langle z_{i}-z_{j},x\rangle\leq0\}\cap A_{i}).
\end{equation}
Fix $i,j\in\{1,2,3\}$ such that $\delta=\gamma_{n}(\{x\in\R^{n}\colon\langle z_{i}-z_{j},x\rangle\leq0\}\cap A_{i})$.  We want to find a bound on $\delta$.  Let $0<h$ such that $\int_{0}^{h}d\gamma_{1}=\delta$.  Now, define $\{A_{r}'\}_{r=1}^{3}$ such that $A_{p}'=A_{p}$ for $p\neq i,j$, $A_{i}'=A_{i}\setminus\left(A_{i}\cap\{x\in\R^{n}\colon\langle z_{i}-z_{j},x\rangle\leq0\}\right)$ and $A_{j}'=A_{j}\cup\left(A_{i}\cap\{x\in\R^{n}\colon\langle z_{i}-z_{j},x\rangle\leq0\}\right)$.  Let $z_{p}'\colonequals\int_{A_{p}'}xd\gamma_{n}(x)$, $p=1,2,3$.  Then
\begin{equation}\label{three51}
\begin{aligned}
&\sum_{p=1}^{3}\vnorm{z_{p}'}_{\ell_{2}^{n}}^{2}-\sum_{p=1}^{3}\vnorm{z_{p}}_{\ell_{2}^{n}}^{2}\\
&\qquad=2\bigg\langle\int_{\{y\colon\langle z_{i}-z_{j},y\rangle\leq0\}\cap A_{i}}yd\gamma_{n}(y),z_{j}-z_{i}\bigg\rangle
+2\bigg\|\int_{\{y\colon\langle z_{i}-z_{j},y\rangle\leq0\}\cap A_{i}}yd\gamma_{n}(y)\bigg\|_{\ell_{2}^{n}}^{2}\\
&\qquad\geq2\bigg\langle\int_{\{y\colon-h\leq\langle z_{i}-z_{j},y\rangle\leq0\}}yd\gamma_{n}(y),z_{j}-z_{i}\bigg\rangle
=2\vnorm{z_{i}-z_{j}}_{\ell_{2}^{n}}\int_{0}^{h}yd\gamma_{1}(y)
\stackrel{\eqref{three27}}{>}\delta^{2}/3.
\end{aligned}
\end{equation}
Here we used rearrangement and also the inequality $\vnorm{z_{i}-z_{j}}_{\ell_{2}^{n}}>(\max_{p\in\{i,j\}}\vnorm{z_{p}}_{\ell_{2}^{n}}^{2})^{1/2}$, which itself uses $\langle z_{i},z_{j}\rangle<0$.

By \eqref{three20} and \eqref{three51}, $\delta^{2}<3\epsilon$, i.e.
\begin{equation}\label{three21}
\delta<\sqrt{3\epsilon}.
\end{equation}

Now, for $p\in\{1,2,3\}$, let $\widetilde{A_{p}}\colonequals\{x\in\R^{n}\colon\langle x,z_{p}\rangle=\max_{j=1,2,3}\langle x,z_{j}\rangle\}$ and let $\widetilde{z_{p}}\colonequals\int_{\widetilde{A_{p}}}xd\gamma_{n}(x)$.  By \eqref{three21} and \eqref{three50},
\begin{equation}\label{three22}
d_{2}(\{A_{i}\}_{i=1}^{3},\{\widetilde{A_{i}}\}_{i=1}^{3})\leq3\sqrt{2}\,\epsilon^{1/4}.
\end{equation}
For $p\in\{1,2,3\}$, let $y_{p}\colonequals\widetilde{z_{p}}-z_{p}\in\R^{n}$, so that $\vnorm{y_{p}}_{2}\leq3\sqrt{2}\,\epsilon^{1/4}$ by \eqref{three22} and Hilbert space duality.  Let $x\in\R^{n}$.  Then for $i,j\in\{1,2,3\}$, $i\neq j$,
\begin{equation}\label{three23}
\langle \widetilde{z_{i}}-\widetilde{z_{j}},x\rangle=
\langle z_{i}-z_{j},x\rangle+\langle y_{i}-y_{j},x\rangle.
\end{equation}

For $i,j\in\{1,2,3\}$, $i\neq j$, let $v_{ij}=S^{n-1}\cap \widetilde{A_{i}}\cap\widetilde{A_{j}}\cap\mathrm{span}\{\widetilde{z_{i}}\}_{i=1}^{3}$.  By definition of $v_{ij}$ and $\{\widetilde{A_{i}}\}_{i=1}^{3}$, $\langle z_{i}-z_{j},v_{ij}\rangle=0$.  So, by \eqref{three23}, $\abs{\langle \widetilde{z_{i}}-\widetilde{z_{j}},v_{ij}\rangle}\leq3\sqrt{2}\,\epsilon^{1/4}$, implying that $d_{2}(\{\widetilde{A_{i}}\}_{i=1}^{3},\{B_{i}\}_{i=1}^{3})\leq3\cdot2^{3/4}\epsilon^{1/8}$, by Lemma \ref{lemma0.5}.  This inequality together with \eqref{three22} and the triangle inequality for $d_{2}$ prove \eqref{three24}.
\end{proof}

\section{The First Variation}\label{secvar}

Recall \eqref{six1.88}.  The following existence argument which uses convexity is a variant of \cite[Lemma 3.1]{khot09} and \cite[Lemma 2.1]{khot11}.
\begin{lemma}[\textbf{First Variation}]\label{lemma1}
Let $\rho\in(0,1)$.  Then $\exists$ a partition $\{A_{i}\}_{i=1}^{k}$ of $\R^{n}$ such that
\begin{equation}\label{three0}
\sum_{i=1}^{k}\int_{\R^{n}}1_{A_{i}}\frac{d}{d\rho}T_{\rho}1_{A_{i}}d\gamma_{n}
=\sup_{(f_{1},\ldots,f_{k})\in\Delta_{k}(\gamma_{n})}\sum_{i=1}^{k}\int_{\R^{n}} f_{i}\frac{d}{d\rho}T_{\rho}f_{i}d\gamma_{n}
\end{equation}
Also, for each $i\in\{1,\ldots,k\}$, the following containment holds, less sets of $\gamma_{n}$ measure zero:
\begin{equation}\label{three1}
A_{i}\supseteq\{x\in\R^{n}\colon LT_{\rho}1_{A_{i}}(x)>LT_{\rho}1_{A_{j}}(x),\forall\,j\neq i,j\in\{1,\ldots,k\}\}.
\end{equation}
\end{lemma}
\begin{proof}
We show that \eqref{six1.5} is maximized over $\Delta_{k}(\gamma_{n})$, which contains the set of partitions of $\R^{n}$.  Note that $\Delta_{k}(\gamma_{n})\subset H$ is norm closed, convex, and norm bounded.  Therefore, $\Delta_{k}(\gamma_{n})$ is weakly closed.  Also, $\Delta_{k}(\gamma_{n})$ is weakly compact by the Banach-Alaoglu Theorem.  Using \eqref{six1.1}, define $\psi_{\rho}\colon\Delta_{k}(\gamma_{n})\to\R$ by
\begin{equation}\label{three1.5}
\psi_{\rho}(g_{1},\ldots,g_{k})\colonequals\frac{d}{d\rho}\sum_{i=1}^{k}\int_{\R^{n}} g_{i}T_{\rho}g_{i}d\gamma_{n}
\colonequals\rho^{-1}\sum_{i=1}^{k}\int_{\R^{n}} g_{i}LT_{\rho}g_{i}d\gamma_{n}.
\end{equation}

By \eqref{six2}, $\psi_{\rho}$ is an exponentially decaying sum of uniformly bounded weakly continuous functions.  Therefore, $\psi_{\rho}$ is weakly continuous on the weakly compact set $\Delta_{k}(\gamma_{n})$.  So there exists $(f_{1},\ldots,f_{k})\in \Delta_{k}(\gamma_{n})$ that maximizes $\psi_{\rho}$.

Since $\rho\in(0,1]$, \eqref{six3.0} implies: $\forall$ $f\in L_{2}(\gamma_{n})$, $\int f LT_{\rho}fd\gamma_{n}\geq0$.  We now apply this fact to see that $\psi_{\rho}$ is convex.  Let $\lambda\in[0,1]$, $(g_{1},\ldots,g_{k}),(h_{1},\ldots,h_{k})\in\Delta_{k}(\gamma_{n})$.  Then
\begin{flalign*}
&\lambda\psi_{\rho}(g_{1},\ldots,g_{k})+(1-\lambda)\psi_{\rho}(h_{1},\ldots,h_{k})
-\psi_{\rho}(\lambda g_{1}+(1-\lambda)h_{1},\ldots,\lambda g_{k}+(1-\lambda)h_{k})\\
&=\frac{1}{\rho}\sum_{i=1}^{k}\left[\lambda\int_{\R^{n}} g_{i}LT_{\rho}g_{i}+(1-\lambda)\int_{\R^{n}} h_{i}LT_{\rho}h_{i}
-\int_{\R^{n}}(\lambda g_{i}-(1-\lambda)h_{i})LT_{\rho}(\lambda g_{i}-(1-\lambda)h_{i})\right]\\
&=\lambda(1-\lambda)\int_{\R^{n}}(g_{i}-h_{i})LT_{\rho}(g_{i}-h_{i})\geq0.
\end{flalign*}

Since $\psi_{\rho}$ is convex on $\Delta_{k}(\gamma_{n})$, $\psi_{\rho}$ achieves its maximum at an extreme point of $\Delta_{k}(\gamma_{n})$.  Therefore, there exists a partition $\{A_{i}\}_{i=1}^{k}$ of $\R^{n}$ such that $(1_{A_{1}},\ldots,1_{A_{k}})\in\Delta_{k}(\gamma_{n})$ maximizes $\psi_{\rho}$ on $\Delta_{k}(\gamma_{n})$ \cite[Lemma 2.1]{khot11}. Specifically, it is noted in \cite[Lemma 2.1]{khot11} that the extreme points of $\Delta_{k}(\gamma_{n})$ are exactly the partitions of $\R^{n}$ into $k$ pieces.

We now prove \eqref{three1} by contradiction.  By the Lebesgue density theorem \cite{stein70}[1.2.1, Proposition 1], we may assume that, for all $i\in\{1,\ldots,k\}$, if $y\in A_{i}$, then we have $\lim_{r\to0}\gamma_{n}(A_{i}\cap B(y,r))/\gamma_{n}(B(y,r))=1$.  Suppose there exist $j,m\in\{1,\ldots,k\}$ and there exists $y\in\R^{n}$, $r>0$ such that $y\in A_{j}$, $\gamma_{n}(B(y,r)\cap A_{j})>0$ and $LT_{\rho}1_{A_{j}}(y)<LT_{\rho}1_{A_{m}}(y)$.  By \eqref{six0},
$$
T_{\rho}1_{A_{j}}(x)=\int_{\R^{n}}1_{A_{j}}(y)e^{-\vnorm{y-x\rho}_{2}^{2}/[2(1-\rho^{2})]}\frac{dy}{(2\pi(1-\rho^{2}))^{n/2}}.
$$
So, $LT_{\rho}1_{A_{j}}=\rho(d/d\rho)T_{\rho}1_{A_{j}}(x)$ is a continuous function of $x$.

Therefore, there exists a ball $B(y,r)$, $r>0$ such that $\gamma_{n}(B(y,r)\cap A_{j})>0$ and such that
$$\sup_{x\in B(y,r)}LT_{\rho}1_{A_{j}}(x)<\inf_{x\in B(y,r)}LT_{\rho}1_{A_{m}}(x).$$
Let $\phi(x)\colonequals1_{B(y,r)\cap A_{j}}(x)$.  For $\lambda\in[0,1]$, note that
\begin{equation}\label{three3}
\left(1_{A_{1}},\ldots,1_{A_{j}}-\lambda\phi,\ldots,1_{A_{}}+\lambda\phi,\ldots,1_{A_{k}}\right)\in\Delta_{k}(\gamma_{n}).
\end{equation}

However,
\begin{equation}\label{three2.5}
\left.\frac{d}{d\lambda}\right|_{\lambda=0}\psi_{\rho}(1_{A_{1}},\ldots,1_{A_{j}}-\lambda\phi,\ldots,1_{A_{m}}+\lambda\phi,\ldots,1_{A_{k}})
=\frac{2}{\rho}\int\phi(x)LT_{\rho}(1_{A_{m}}-1_{A_{j}})(x)d\gamma_{n}(x)>0.
\end{equation}
But \eqref{three2.5} contradicts the maximality of $(1_{A_{1}},\ldots,1_{A_{k}})$ on $\Delta_{k}(\gamma_{n})$, so \eqref{three1} holds.

\end{proof}


\section{Perturbative Estimates}\label{secpert}

Recalling \eqref{three1.5}, the following estimates allow us to relate $\psi_{\rho}$ to $\psi_{0}$ for small $\rho>0$, for simplicial conical partitions.  In particular, we make a close examination of the two quantities of \eqref{six1.2}.  Since lemma \ref{lemma4} gives precise estimates of the two quantities of \eqref{six1.2}, combining Lemma \ref{lemma4} with \eqref{three1} gives precise geometric information about a partition $\{A_{i}\}_{i=1}^{k}\subset\R^{n}$ optimizing noise stability.  In particular, to see one way that we will apply Lemma \ref{lemma4}, see \eqref{two3} below.  However, note that \eqref{two3} below does not give sufficiently precise information to identify the sets optimizing noise stability.  So, the real need for Lemma \ref{lemma4} will occur in the proof of the Main Lemma \ref{lemma8}, where the precise estimate \eqref{five6} is used.

\begin{lemma}\label{lemma3}
Let $A\subset\R^{n}$ be a cone.  Then
$$
\int_{\R^{n}}\bigg(\sum_{i=1}^{n}(1-y_{i}^{2})\bigg)1_{A}(y)d\gamma_{n}(y)=0.
$$
\end{lemma}
\begin{proof}
The assertion follows by standard equalities for the moments of a Gaussian random variable.  Let $\alpha>0$.  Define $f(\alpha)$ by the formula
$$
f(\alpha)\colonequals\int_{\R^{n}}1_{A}(y)e^{-\alpha(y_{1}^{2}+\cdots+y_{n}^{2})/2}\frac{dy}{(2\pi)^{n/2}}.
$$
By changing variables, $f(\alpha)=\alpha^{-n/2}\int_{\R^{n}}1_{A}(y)d\gamma_{n}(y)$.  So,
$$-\frac{1}{2}\int_{\R^{n}}\bigg(\sum_{i=1}^{n}y_{i}^{2}\bigg)1_{A}(y)d\gamma_{n}(y)
=\left.\frac{df(\alpha)}{d\alpha}\right|_{\alpha=1}
=-\frac{n}{2}\int_{\R^{n}}1_{A}(y)d\gamma_{n}(y).$$
\end{proof}

\begin{lemma}\label{lemma4}
Fix $k=3$, $n\geq2$, $\rho\in(0,1)$.  Let $\{C_{i}\}_{i=1}^{k}\subset\R^{2}$ be a simplicial conical partition.  Let $\{B_{i}\}_{i=1}^{k}\colonequals\{C_{i}\times\R^{n-2}\}_{i=1}^{k}$.  Fix $i,j\in\{1,\ldots,k\}$.  Let $\sigma\colon\R^{n}\to\R^{n}$ denote reflection across $B_{i}\cap B_{j}$.  Assume that $B_{i}=\sigma B_{j}$ and that $B_{i}\subset\{x\in\R^{n}\colon x_{1}\geq0\}$. Let $e_{1}=(1,0,\ldots,0)$, $e_{2}=(0,1,0,\ldots,0)$, $e_{1},e_{2}\in\R^{n}$.  For $p\in\{1,\ldots,k\}$, let $z_{p}\colonequals\int_{B_{p}}xd\gamma_{n}(x)$.  Note that $\mathrm{span}\{z_{i},z_{j}\}=\mathrm{span}\{e_{1},e_{2}\}$.  Let $n_{j}\in\R^{n}$ be the interior unit normal of $B_{j}$ so that $n_{j}$ is normal to the face $(\partial B_{j})\setminus(\partial B_{i})$, and let $n_{i}\in\R^{n}$ be the interior unit normal of $B_{i}$ so that $n_{i}$ is normal to the face $(\partial B_{i})\setminus(\partial B_{j})$.

(i) If $x\in B_{i}\cap\{x\in\R^{n}\colon\langle x,n_{j}\rangle\leq0\}$, then
\begin{equation}\label{four5}
\frac{1}{\rho}\left\langle x,\nabla T_{\rho}(1_{B_{i}}-1_{B_{j}})(x)\right\rangle
\geq 2x_{1}\gamma_{n}\left(\delta_{\frac{(B_{i}\cap B_{j})-x\rho}{\sqrt{1-\rho^{2}}}}\right)
+\langle x,n_{i}\rangle\gamma_{n}\left(\delta_{\frac{((\partial B_{i})\setminus B_{j})-x\rho}{\sqrt{1-\rho^{2}}}}\right).
\end{equation}

(ii) If $x\in B_{i}\cap\{x\in\R^{n}\colon\langle x,n_{j}\rangle\geq0\}$, then
\begin{equation}\label{four5.5}
\frac{1}{\rho}\left\langle x,\nabla T_{\rho}(1_{B_{i}}-1_{B_{j}})(x)\right\rangle
\geq2x_{1}\gamma_{n}\left(\delta_{\frac{(B_{i}\cap B_{j})-x\rho}{\sqrt{1-\rho^{2}}}}\right).
\end{equation}

(iii) For $x\in B_{i}$,
\begin{equation}\label{four0.5}
\begin{aligned}
\bigg|\int_{\R^{n}}\bigg(\sum_{\ell=1}^{n}(1-y_{\ell}^{2})\bigg)(1_{B_{i}}-1_{B_{j}})(x\rho+y\sqrt{1-\rho^{2}})d\gamma_{n}(y)\bigg|
\leq\frac{\rho}{\sqrt{1-\rho^{2}}}(\sqrt{6}+(n-1)\sqrt{2})x_{1}.
\end{aligned}
\end{equation}

(iv) For $x\in B_{i}$ with $x_{1}>\sqrt{n}\sqrt{1-\rho^{2}}/\rho$,
\begin{equation}\label{four0.6}
\begin{aligned}
\int_{\R^{n}}\bigg(\sum_{\ell=1}^{n}(1-y_{\ell}^{2})\bigg)(1_{B_{i}}-1_{B_{j}})(x\rho+y\sqrt{1-\rho^{2}})d\gamma_{n}(y)\geq0.
\end{aligned}
\end{equation}
\end{lemma}
\begin{proof}[Proof of (i)]
Below, we use differentiation in the distributional sense.  Let $x\neq0$.  For $x\in(\partial B_{i})\cap B_{j}$, $\nabla1_{B_{i}}(x)=e_{1}$ since $B_{i}\subset\{x\in\R^{n}\colon x_{1}\geq0\}$, and for $x\in(\partial B_{i})\setminus B_{j}$, $\nabla 1_{B_{i}}(x)=n_{i}$.  Similarly, for $x\in(\partial B_{j})\cap B_{i}$, $-\nabla1_{B_{j}}(x)=e_{1}$, and for $x\in(\partial B_{j})\setminus B_{i}$, $-\nabla 1_{B_{j}}(x)=-n_{j}$.  Then
\begin{equation}\label{four2}
\begin{aligned}
&\frac{1}{\rho}\nabla T_{\rho}(1_{B_{i}}-1_{B_{j}})(x)
=T_{\rho}(\nabla (1_{B_{i}}-1_{B_{j}}))(x)\\
&\,=T_{\rho}[2(e_{1})\delta_{B_{i}\cap B_{j}}
+n_{i}\delta_{(\partial B_{i})\setminus B_{j}}
+(-n_{j})\delta_{(\partial B_{j})\setminus B_{i}}](x)\\
&\,\stackrel{\eqref{six0}}{=}2e_{1}\gamma_{n}\left(\delta_{\frac{(B_{i}\cap B_{j})-x\rho}{\sqrt{1-\rho^{2}}}}\right)
+n_{i}\gamma_{n}\left(\delta_{\frac{((\partial B_{i})\setminus B_{j})-x\rho}{\sqrt{1-\rho^{2}}}}\right)
+(-n_{j})\gamma_{n}\left(\delta_{\frac{((\partial B_{j})\setminus B_{i})-x\rho}{\sqrt{1-\rho^{2}}}}\right).
\end{aligned}
\end{equation}
Here we used
$$\int_{\R^{n}}1_{A}(x\rho+y\sqrt{1-\rho^{2}})d\gamma_{n}(y)=\int_{\R^{n}}1_{A-x\rho}(y\sqrt{1-\rho^{2}})d\gamma_{n}(y)
=\int_{\R^{n}} 1_{(A-x\rho)/\sqrt{1-\rho^{2}}}(y)d\gamma_{n}(y).$$

Let $x$ with $x\in B_{i}$ and $\langle x,(-n_{j})\rangle\geq0$.  Then \eqref{four2} immediately proves \eqref{four5}.
\end{proof}
\begin{proof}[Proof of (ii)]
Let $x\in B_{i}\cap\{x\in\R^{n}\colon\langle x,n_{j}\rangle\geq0\}$.  By reflecting across $B_{i}\cap B_{j}$, \begin{equation}\label{four5.4}
\gamma_{n}\left(\delta_{\frac{((\partial B_{i})\setminus B_{j})-x\rho}{\sqrt{1-\rho^{2}}}}\right)
\geq\gamma_{n}\left(\delta_{\frac{((\partial B_{j})\setminus B_{i})-x\rho}{\sqrt{1-\rho^{2}}}}\right).
\end{equation}
Define
$$
w\colonequals n_{i}\gamma_{n}\left(\delta_{\frac{((\partial B_{i})\setminus B_{j})-x\rho}{\sqrt{1-\rho^{2}}}}\right)
+(-n_{j})\gamma_{n}\left(\delta_{\frac{((\partial B_{j})\setminus B_{i})-x\rho}{\sqrt{1-\rho^{2}}}}\right).
$$
By \eqref{four5.4}, $w$ is in the convex hull of $e_{1}$ and $n_{i}$.  In particular, $\langle x,w\rangle\geq0$, since $x\in B_{i}$.  Combining $\langle x,w\rangle\geq0$ with \eqref{four2} proves \eqref{four5.5}.
\end{proof}
\begin{proof}[Proof of (iii)]
By reflecting across $B_{i}\cap B_{j}$,
$$x\in B_{i}\cap B_{j}\quad\Longrightarrow\quad \int_{\R^{n}}\bigg(\sum_{\ell=1}^{n}(1-y_{\ell}^{2})\bigg)(1_{B_{i}}-1_{B_{j}})(x\rho+y\sqrt{1-\rho^{2}})d\gamma_{n}(y)=0.$$
So, a derivative bound gives \eqref{four0.5}.  Specifically, we apply the Fundamental Theorem of Calculus to the following identity, along with $\vnorm{(1-y_{2}^{2})y_{1}}_{L_{2}(\gamma_{n})}=\sqrt{2}$ and $\vnorm{y_{1}^{3}+3y_{1}}_{L_{2}(\gamma_{n})}=\sqrt{6}$.
\begin{flalign*}
&\frac{\partial}{\partial x_{1}}\int_{\R^{n}}\bigg(\sum_{\ell=1}^{n}(1-y_{\ell}^{2})\bigg)(1_{B_{i}}-1_{B_{j}})(x\rho+y\sqrt{1-\rho^{2}})d\gamma_{n}(y)\\
&\qquad=\frac{\rho}{\sqrt{1-\rho^{2}}}\int_{\R^{n}}
\bigg(-3y_{1}-y_{1}^{3}+\sum_{\ell\neq1}(1-y_{\ell}^{2})y_{1}\bigg)(1_{B_{i}}-1_{B_{j}})(x\rho+y\sqrt{1-\rho^{2}})d\gamma_{n}(y).
\end{flalign*}
\end{proof}
\begin{proof}[Proof of (iv)]
Let $x$ with $x_{1}>\sqrt{n}\sqrt{1-\rho^{2}}/\rho$ and consider the following cone 
$$A\colonequals\{0\}\cup\left\{y\in\R^{n}\colon y\neq0\,\wedge\,\sqrt{n}\frac{y}{\vnorm{y}_{2}}\in\frac{B_{i}-x\rho}{\sqrt{1-\rho^{2}}}\right\}.$$
By Lemma \ref{lemma3}, $\int_{\R^{n}}1_{A}(y)\sum_{\ell=1}^{n}(1-y_{\ell}^{2})d\gamma_{n}(y)=0$.  If $d(x,\partial B_{i})\geq\sqrt{n}\sqrt{1-\rho^{2}}/\rho$, then $A=\R^{n}$ and $1_{A}(y)\sum_{\ell=1}^{n}(1-y_{\ell}^{2})1_{B_{i}^{c}}(x\rho+y\sqrt{1-\rho^{2}})\leq0$, so
\begin{flalign*}
&\int_{\R^{n}}\bigg(\sum_{\ell=1}^{n}(1-y_{\ell}^{2})\bigg)(1_{B_{i}}-1_{B_{j}})(x\rho+y\sqrt{1-\rho^{2}})d\gamma_{n}(y)\\
&\qquad\geq\int_{\R^{n}}\bigg(\sum_{\ell=1}^{n}(1-y_{\ell}^{2})\bigg)1_{B_{i}}(x\rho+y\sqrt{1-\rho^{2}})d\gamma_{n}(y)
\geq\int_{\R^{n}}\bigg(\sum_{\ell=1}^{n}(1-y_{\ell}^{2})\bigg)d\gamma_{n}(y)=0.
\end{flalign*}
So, it remains to consider the case $d(x,\partial B_{i})<\sqrt{n}\sqrt{1-\rho^{2}}/\rho$.  In this case $A\neq\R^{n}$.

Since $x_{1}>\sqrt{n}\sqrt{1-\rho^{2}}/\rho$, we have $\sum_{\ell=1}^{n}(1-y_{\ell}^{2})1_{B_{j}}(x\rho+y\sqrt{1-\rho^{2}})\leq0$.  Also, we have $\sum_{\ell=1}^{n}(1-y_{\ell}^{2})1_{A^{c}}(y)1_{B_{i}}(x\rho+y\sqrt{1-\rho^{2}})\geq0$,  $\sum_{\ell=1}^{n}(1-y_{\ell}^{2})1_{A}(y)1_{B_{i}^{c}}(x\rho+y\sqrt{1-\rho^{2}})\leq0$, so
\begin{flalign*}
&\int_{\R^{n}}\bigg(\sum_{\ell=1}^{n}(1-y_{\ell}^{2})\bigg)(1_{B_{i}}-1_{B_{j}})(x\rho+y\sqrt{1-\rho^{2}})d\gamma_{n}(y)\\
&\qquad\geq\int_{\R^{n}}\bigg(\sum_{\ell=1}^{n}(1-y_{\ell}^{2})\bigg)1_{B_{i}}(x\rho+y\sqrt{1-\rho^{2}})d\gamma_{n}(y)\\
&\qquad\geq\int_{\R^{n}}\bigg(\sum_{\ell=1}^{n}(1-y_{\ell}^{2})\bigg)1_{A}(y)1_{B_{i}}(x\rho+y\sqrt{1-\rho^{2}})d\gamma_{n}(y)\\
&\qquad\geq\int_{\R^{n}}\bigg(\sum_{\ell=1}^{n}(1-y_{\ell}^{2})\bigg)1_{A}(y)d\gamma_{n}(y)=0.
\end{flalign*}
\end{proof}

\section{Iterative Estimates}\label{seciter}

The following estimates control the errors that appear in the proof of Theorem \ref{thm0}.  Being rather technical in nature, this section could be skipped on a first reading.

\begin{lemma}\label{lemma6}
For $\ell=(\ell_{1},\ldots,\ell_{n})\in\N^{n}$ and $x=(x_{1},\ldots,x_{n})\in\R^{n}$,
$$
h_{\ell}(x)\sqrt{\ell!}
\leq\abs{\ell}^{n}3^{\abs{\ell}}\prod_{i=1}^{n}\max\{1,\abs{x_{i}}^{\ell_{i}}\}.
$$
\end{lemma}
\begin{proof}
Let $\ell\in\N$.
$$
\sum_{\ell=0}^{\infty}\lambda^{\ell}h_{\ell}(x)
\stackrel{\eqref{six1.9}}{=}e^{\lambda x-\lambda^{2}/2}
=\sum_{p=0}^{\infty}\frac{x^{p}}{p!}\lambda^{p}\sum_{q=0}^{\infty}\frac{(-1)^{q}(\lambda)^{2q}(1/2)^{q}}{q!}
=\sum_{\ell=0}^{\infty}\lambda^{\ell}\sum_{m=0}^{\lfloor \ell/2\rfloor}\frac{x^{\ell-2m}(-1)^{m}2^{-m}}{m!(\ell-2m)!}.
$$
Here we let $p+2q=\ell$, $m=q$.  In particular,
\begin{equation}\label{four0.8}
h_{\ell}(x)=\sum_{m=0}^{\lfloor \ell/2\rfloor}\frac{x^{\ell-2m}(-1)^{m}2^{-m}}{m!(\ell-2m)!}.
\end{equation}

Using Stirling's formula, $\sqrt{2\pi}\ell^{\ell+1/2}e^{-\ell}\leq \ell!\leq e\,\ell^{\ell+1/2}e^{-\ell}$.  Let $\ell,m$ with $m\in\{0,\ldots,\ell/2\}$, $\ell\geq1$.  Note that $\lim_{x\to0^{+}}x^{x}=1$ and $\min_{x\in[0,1]}x^{x}>2/3$.  Also, $m+\ell-2m=\ell-m\geq \ell/2$.  For $m\neq0$, write $m=\ell j$, $j\in[1/\ell,1/2]$.  Note that $\max\{m,\ell-2m\}\geq \ell/3$.  Then
\begin{flalign*}
&\frac{\sqrt{\ell!}}{m!(\ell-2m)!}\\
&\leq\frac{\sqrt{e}\,\ell^{(1/2)\ell+1/4}e^{-\ell/2}}{2\pi m^{m+1/2}e^{-m}(\ell-2m)^{\ell-2m+1/2}e^{-(\ell-2m)}}
=\frac{\sqrt{e}}{2\pi}\frac{\ell^{(1/2)\ell+1/4}}{m^{m+1/2}(\ell-2m)^{\ell-2m+1/2}}e^{\ell/2-m}\\
&=\frac{\sqrt{e}}{2\pi}\frac{\ell^{1/4}}{\sqrt{m}\sqrt{\ell-2m}}\frac{\ell^{\ell/2}}{m^{m}(\ell-2m)^{\ell-2m}}e^{\ell/2-m}
=\frac{\sqrt{e}}{2\pi}\frac{\ell^{1/4}}{\sqrt{m}\sqrt{\ell-2m}}\frac{\ell^{\ell/2}}{(\ell j)^{\ell j}(\ell(1-2j))^{\ell(1-2j)}}e^{\frac{\ell}{2}-m}\\
&=\frac{\sqrt{e}}{2\pi}\frac{\ell^{1/4}}{\sqrt{m}\sqrt{\ell-2m}}\frac{\ell^{\ell/2}e^{\ell/2-m}}{\ell^{\ell(1-j)}j^{\ell j}(1-2j)^{\ell(1-2j)}}
=\frac{\sqrt{e}}{2\pi}\frac{\ell^{1/4}}{\sqrt{m}\sqrt{\ell-2m}}\frac{\ell^{\ell/2}e^{\ell/2-m}}{\ell^{\ell/2}\ell^{\frac{\ell}{2}(1-2j)}j^{\ell j}(1-2j)^{\ell(1-2j)}}\\
&=\frac{\sqrt{e}}{2\pi}\frac{\ell^{1/4}}{\sqrt{m}\sqrt{\ell-2m}}\frac{1}{j^{\ell j}(1-2j)^{\ell(1-2j)}}\frac{e^{\frac{\ell}{2}(1-2j)}}{\ell^{\frac{\ell}{2}(1-2j)}}
=\frac{\sqrt{e}}{2\pi}\frac{\ell^{1/4}}{\sqrt{m}\sqrt{\ell-2m}}\frac{(e/\ell)^{\frac{\ell}{2}(1-2j)}}{j^{\ell j}(1-2j)^{\ell(1-2j)}}\\
&\leq\frac{e}{2\pi}\frac{\ell^{1/4}}{\sqrt{m}\sqrt{\ell-2m}}\frac{1}{j^{\ell j}(1-2j)^{\ell(1-2j)}}
\leq\frac{e}{2\pi}\frac{\ell^{1/4}}{\sqrt{m}\sqrt{\ell-2m}}\frac{1}{(2/3)^{2\ell}}\\
&\leq\frac{e\sqrt{3}}{\ell^{1/4}2\pi}\frac{1}{(2/3)^{2\ell}}
=\frac{e\sqrt{3}}{\ell^{1/4}2\pi}(9/4)^{\ell}\leq \ell^{-1/4}(9/4)^{\ell}.
\end{flalign*}
Here we used $(e/\ell)^{(\ell/2)(1-2j)}\leq\sqrt{e}$ for $\ell=1,2$.

Also, for $m=0$ we have $\frac{\sqrt{\ell!}}{m!(\ell-2m)!}=1$, and for $m=\ell/2$ we have
\begin{flalign*}
\frac{\sqrt{\ell!}}{m!(\ell-2m)!}
&=\frac{\sqrt{\ell!}}{(\ell/2)!}
\leq\frac{\sqrt{e}\ell^{\ell/2+1/4}e^{-\ell/2}}{\sqrt{2\pi}(\ell/2)^{\ell/2+1/2}e^{-\ell/2}}
=\sqrt{e}\frac{\ell^{1/4}}{\sqrt{2\pi}\ell^{1/2}2^{-\ell/2}2^{-1/2}}\\
&=\sqrt{\frac{e}{\pi}}\ell^{-1/4}2^{\ell/2}
\leq \ell^{-1/4}2^{\ell/2}.
\end{flalign*}

So, combining the above estimates with \eqref{four0.8},
\begin{flalign*}
\absf{h_{\ell}(x)\sqrt{\ell!}}
&\leq\sum_{m=0}^{\lfloor \ell/2\rfloor}\ell^{-1/4}(9/4)^{\ell}\abs{x}^{\ell-2m}
\leq\sum_{m=0}^{\lfloor \ell/2\rfloor}\ell^{-1/4}(9/4)^{\ell}\max\{1,\abs{x}^{\ell-2m}\}\\
&\leq\ell\,\ell^{-1/4}(9/4)^{\ell}\max\{1,\abs{x}^{\ell}\}
\leq\ell\,3^{\ell}\max\{1,\abs{x}^{\ell}\}.
\end{flalign*}

Therefore, for $\ell=(\ell_{1},\ldots,\ell_{n})\in\N^{n}$,
$$
h_{\ell}(x)\sqrt{\ell!}
\leq\ell_{1}\cdots\ell_{n}3^{\ell_{1}+\cdots+\ell_{n}}\prod_{i=1}^{k}\max\{1,\abs{x_{i}}^{\ell_{i}}\}
\leq\abs{\ell}^{n}3^{\abs{\ell}}\prod_{i=1}^{n}\max\{1,\abs{x_{i}}^{\ell_{i}}\}.
$$

\end{proof}

The following Lemma uses standard tail bounds for a Gaussian random variable.  We therefore omit the proof.

%
%
\begin{lemma}\label{lemma6.1}  Let $\eta>0,t>0$, and let $n\geq2$.  Then
$$
\bigg|\int_{[-\eta,\eta]\times[t,\infty]\times\R^{n-2}}
\sum_{\substack{\ell\in\N^{n}\colon\\0\leq\abs{\ell}\leq3}}
\prod_{i=1}^{n}\abs{y_{i}}^{\ell_{i}}d\gamma_{n}(y)\bigg|
\leq 3000n^{3}\eta(t^{2}+2)e^{-t^{2}/2},
$$
\begin{flalign*}
\bigg|\int_{B(0,t)^{c}}\sum_{\substack{\ell\in\N^{n}\colon\\0\leq\abs{\ell}\leq3}}
\prod_{i=1}^{n}\abs{y_{i}}^{\ell_{i}}d\gamma_{n}(y)\bigg|
&\leq4n^{3}2^{-n/2}(\Gamma(n/2))^{-1}(n+2)!(t^{n+1}+1)e^{-t^{2}/2}\\
&\leq 100(n+2)!(t^{n+1}+1)e^{-t^{2}/2}.
\end{flalign*}
\end{lemma}
%

The following Lemma says, if $\int_{\R^{n}} xf(x)d\gamma_{n}(x)$ is parallel to the $x_{1}$-axis, then $(d/d\rho)T_{\rho}f(x)$ should be bounded by a constant multiplied by $\abs{x_{1}}+O(\rho)$.  The precise error term \eqref{five0} will be needed in Lemma \ref{lemma8} to determine the size of $(d/d\rho)T_{\rho}(1_{A_{i}}-1_{A_{j}})$.  The error term \eqref{five0} will be estimated by Lemma \ref{lemma6.1}, and the resulting estimate will be introduced into \eqref{three1}.

\begin{lemma}\label{lemma7}
Let $\rho\in(-1,1)$, $n\geq2$.  Suppose $f\in L_{2}(\gamma_{n})$ with $\int_{\R^{n}} y_{2}f(y)d\gamma_{n}(y)=0$.  Let $x_{1}\geq0$ and $x_{2}\geq0$.  Then
\begin{equation}\label{five0}
\begin{aligned}
&\abs{\frac{d}{d\rho}T_{\rho}f(x_{1},x_{2},0,\ldots,0)}
\leq\bigg(\abs{x_{1}}+2\rho(\abs{x_{2}}^{2}+(n+1)\rho\abs{x_{2}}+\abs{x_{1}x_{2}}+2n)\bigg)\\
&\qquad\qquad\qquad\cdot\sup_{\substack{t_{1}\in[0,x_{1}],\,t_{2}\in[0,x_{2}]\\ \eta\in[0,\rho]}}
\bigg|\int_{\R^{n}}\sum_{\substack{\ell\in\N^{n}\colon\\0\leq\abs{\ell}\leq3}}
\prod_{i=1}^{n}\abs{y_{i}}^{\ell_{i}}f((t_{1},t_{2},0,\ldots,0)\eta+y\sqrt{1-\eta^{2}})d\gamma_{n}(y)\bigg|.
\end{aligned}
\end{equation}
\end{lemma}
\begin{proof}

By integrating by parts, note that
\begin{flalign*}
&\frac{d}{d\rho}\int_{\R^{n}} y_{2}f(y\sqrt{1-\rho^{2}})d\gamma_{n}(y)\\
&\,=\frac{\rho}{1-\rho^{2}}\int_{\R^{n}} y_{2}((n+1)-y_{2}^{2})f(y\sqrt{1-\rho^{2}})d\gamma_{n}(y)
-\sum_{i\neq2}\frac{\rho}{1-\rho^{2}}\int_{\R^{n}} y_{i}^{2}y_{2}f(y\sqrt{1-\rho^{2}})d\gamma_{n}(y).
\end{flalign*}
So, using $\int_{\R^{n}} y_{2}f(y)d\gamma_{n}(y)=0$ and the Fundamental Theorem of Calculus,
\begin{equation}\label{five0.1}
\begin{aligned}
&\int_{\R^{n}} y_{2}f(y\sqrt{1-\rho^{2}})d\gamma_{n}(y)\\
&\leq\frac{\rho^{2}}{1-\rho^{2}}\sup_{\eta\in[0,\rho]}\left(\int_{\R^{n}} y_{2}((n+1)-y_{2}^{2})f(y\sqrt{1-\eta^{2}})d\gamma_{n}(y)\right.\\
&\qquad\qquad\qquad\qquad\left.-\sum_{i\neq2}\int_{\R^{n}} y_{i}^{2}y_{2}f(y\sqrt{1-\eta^{2}})d\gamma_{n}(y)\right).
\end{aligned}
\end{equation}
%

By integrating by parts again, note that
\begin{equation}\label{five0.2}
\begin{aligned}
&\frac{\partial}{\partial x_{2}}\int_{\R^{n}} y_{2}f((0,x_{2},0,\ldots,0)\rho+y\sqrt{1-\rho^{2}})d\gamma_{n}(y)\\
&\qquad=\frac{\rho}{\sqrt{1-\rho^{2}}}\int_{\R^{n}} f((0,x_{2},0,\ldots,0)\rho+y\sqrt{1-\rho^{2}})(y_{2}^{2}-1)d\gamma_{n}(y).
\end{aligned}
\end{equation}
Applying the Fundamental Theorem of Calculus to \eqref{five0.2} and then using \eqref{five0.1},
\begin{equation}\label{five1}
\begin{aligned}
&\abs{\int_{\R^{n}} y_{2}f((0,x_{2},0,\ldots,0)\rho+y\sqrt{1-\rho^{2}})d\gamma_{n}(y)}\\
&\,\leq\abs{x_{2}}\sup_{t\in[0,x_{2}]}\abs{\frac{\rho}{\sqrt{1-\rho^{2}}}\int_{\R^{n}} f((0,t,0,\ldots,0)\rho+y\sqrt{1-\rho^{2}})(y_{2}^{2}-1)d\gamma_{n}(y)}\\
&\,+\frac{\rho^{2}}{1-\rho^{2}}\sup_{\eta\in[0,\rho]}\left(\int_{\R^{n}} y_{2}((n+1)-y_{2}^{2})f(y\sqrt{1-\eta^{2}})d\gamma_{n}(y)\right.\\
&\qquad\qquad\qquad\qquad\left.-\sum_{i\neq2}\int_{\R^{n}} y_{i}^{2}y_{2}f(y\sqrt{1-\eta^{2}})d\gamma_{n}(y)\right).
\end{aligned}
\end{equation}

By integrating by parts as before,
\begin{equation}\label{five3}
\begin{aligned}
&\frac{\partial}{\partial x_{1}}[x_{2}\int_{\R^{n}} y_{2}f((x_{1},x_{2},0,\ldots,0)\rho+y\sqrt{1-\rho^{2}})d\gamma_{n}(y)]\\
&\qquad=x_{2}\frac{\rho}{\sqrt{1-\rho^{2}}}\int_{\R^{n}} y_{2}y_{1}f((x_{1},x_{2},0,\ldots,0)\rho+y\sqrt{1-\rho^{2}})d\gamma_{n}(y).
\end{aligned}
\end{equation}
Combining \eqref{six1.2}, \eqref{five1} and \eqref{five3},
\begin{equation}\label{five4} 
\begin{aligned}
&\abs{(d/d\rho)T_{\rho}f(x)}
\leq\abs{x_{1}}\abs{\int_{\R^{n}} y_{1}f((x_{1},x_{2},0,\ldots,0)\rho+y\sqrt{1-\rho^{2}})d\gamma_{n}(y)}\\
&\,\,+\abs{x_{2}}^{2}\sup_{t\in[0,x_{2}]}\abs{\frac{\rho}{\sqrt{1-\rho^{2}}}\int_{\R^{n}} f((0,t,0,\ldots,0)\rho+y\sqrt{1-\rho^{2}})(y_{2}^{2}-1)d\gamma_{n}(y)}\\
&\,\,+\frac{\rho^{2}\abs{x_{2}}}{1-\rho^{2}}\sup_{\eta\in[0,\rho]}\left(\int_{\R^{n}} y_{2}((n+1)-y_{2}^{2})f(y\sqrt{1-\eta^{2}})d\gamma_{n}(y)\right.\\
&\qquad\qquad\qquad\qquad\qquad\left.-\sum_{i\neq2}\int_{\R^{n}} y_{i}^{2}y_{2}f(y\sqrt{1-\eta^{2}})d\gamma_{n}(y)\right)\\
&\,\,+\abs{x_{1}x_{2}}\sup_{t\in[0,x_{1}]}\frac{\rho}{\sqrt{1-\rho^{2}}}\abs{\int_{\R^{n}} y_{2}y_{1}f((t,x_{2},0,\ldots,0)\rho+y\sqrt{1-\rho^{2}})d\gamma_{n}(y)}\\
&\,\,+\frac{\rho}{\sqrt{1-\rho^{2}}}\abs{\int_{\R^{n}}(\sum_{i=1}^{n}(y_{i}^{2}-1))f((x_{1},x_{2},0,\ldots,0)\rho+y\sqrt{1-\rho^{2}})d\gamma_{n}(y)}.
\end{aligned}
\end{equation}
We then deduce \eqref{five0} from \eqref{five4}.
\end{proof}

\section{The Main Lemma}

Lemma \ref{lemma8} below represents the main tool in the proof of the main theorem.  As depicted in Figure \ref{fig0}, Lemma \ref{lemma8} says that, if an optimal partition is close to being simplicial conical, then it is actually much closer to being simplicial conical.  So, this Lemma can be understood as a feedback loop, or as a contractive mapping type of argument.  We first give an intuitive sketch of the proof of the Lemma.  Let $\rho>0$.  We begin with a partition $\{A_{p}\}_{p=1}^{3}\subset\R^{n}$ maximizing noise stability \eqref{six1.5}.  We assume that there are disjoint sets $\{D_{p}\}_{p=1}^{3}$ that resemble a simplicial conical partition, as in the left side of Figure \ref{fig0}.  We also assume that $A_{p}\supset D_{p}$ for all $p=1,2,3$.  We then find a sequence of sets $\{D_{p,1}\}_{p=1}^{3}$, $\{D_{p,2}\}_{p=1}^{3}$, $\ldots$ $\{D_{p,R}\}_{p=1}^{3}$ such that $D_{p,r}\subset D_{p,r+1}$ for all $1\leq p\leq 3$, for all $r\geq1$.  This sequence of sets is chosen so that the following implication can be proven:
\begin{equation}\label{ten}
A_{p}\supset D_{p,r}\quad\Longrightarrow\quad A_{p}\supset D_{p,r+1}.
\end{equation}
In order to prove \eqref{ten}, we need to show: if $A_{p}\supset D_{p,r}$, then we can get sufficiently strong estimates on $LT_{\rho}1_{D_{p,r+1}}$ such that \eqref{three1} can be verified on $A_{p}$ for each $p=1,2,3$.  For example, in Step 1 of the proof of Lemma \ref{lemma8}, the estimate \eqref{five8} eventually implies \eqref{five8.3}.  And \eqref{five8.3} says that $A_{p}$ must contain more points than the initial information that we assumed in \eqref{five0.3}.

Finally, we need to choose our sets $\{D_{p,r}\}_{p=1}^{3}$ appropriately so that, after finitely many implications of the form \eqref{ten}, we eventually get the conclusion \eqref{five0.4}.  That is, the three sets $\{D_{p,R}\}_{p=1}^{3}$ resemble the right side of Figure \ref{fig0}, and $A_{p}\supset D_{p,R}$ for each $p=1,2,3$.  Thus concludes our description of the main strategy of the proof.  Within the proof itself, the sets $\{D_{p,1}\}_{p=1}^{3}$, $\{D_{p,2}\}_{p=1}^{3}$, $\ldots$ will not be explicitly defined.  However, portions of these sets will be defined at the end of every Step of the proof.  In particular, examine the sets defined by the following sequence of assertions: \eqref{five0.3}, \eqref{five8.3}, \eqref{five7.3}, \eqref{five11}, \eqref{eight10}, \eqref{eight21}, \eqref{eight32}, \eqref{eight40}, and finally \eqref{five0.4}.

Unfortunately, there are many technical obstacles that stand in the way of bringing this strategy to fruition.  The first minor issue is that we cannot control small rotations of our sets.  At every step of the proof, we therefore need to redefine our simplicial sets $B_{i},B_{j}$ to account for these small rotations.  However, the main technical issue is that it is not at all obvious how to choose the sets $\{D_{p,r}\}_{p=1}^{3}$ for $r=1,2,3,\ldots$ such that \eqref{ten} can be proven for each $r=1,2,3,\ldots$.  Moreover, the simplest choice of these sets, namely dilations of the sets depicted in Figure \ref{fig0}, do not produce satisfactory estimates.

Ultimately, the sequence of sets defined by \eqref{five8.3}, \eqref{five7.3}, \eqref{five11}, \eqref{eight10}, \eqref{eight21}, \eqref{eight32}, \eqref{eight40} succeeds in proving the sequence of implications \eqref{ten} for $r=1,2,3,\ldots$.  Lemma \ref{lemma7} allows us to control the errors from our estimates, and we then make around seven modifications of the same error estimate within Lemma \ref{lemma8}.  This error estimate allows us to apply Lemma \ref{lemma4}, so that we can improve our knowledge of the optimal partition $\{A_{i}\}_{i=1}^{k}$ via \eqref{three1}.

It would be preferable to write Lemma \ref{lemma8} as seven applications of a single Lemma, however the statement of such a Lemma would perhaps be so long and convoluted that its application would become opaque.  We therefore use the longer presentation below in the hope of providing greater clarity.  Finally, in the statement of Lemma \ref{lemma8} below, note that the plane $\Pi$ exists independently of $i,j\in\{1,\ldots,k\}$.

\begin{lemma}\label{lemma8}
Fix $n\geq2$, $k=3$.  Let $0<\eta<\rho<e^{-20(n+1)^{10^{12}n^{3}(n+2)!}}$.  Let $\{A_{i}\}_{i=1}^{k}$ be a partition of $\R^{n}$ such that \eqref{three1} holds.  Let $\Pi\subset\R^{n}$ be a fixed $2$-dimensional plane such that $0\in\Pi$.  Assume that, for each pair $i,j\in\{1,2,\ldots,k\}$ with $i\neq j$, there exists $\lambda'>0$ and there exists a regular simplicial conical partition $\{B_{p}'\}_{p=1}^{k}\subset\R^{n}$ such that
\begin{equation}\label{five0.5}
\int_{\R^{n}}y(1_{A_{i}}(y)-1_{A_{j}}(y))d\gamma_{n}(y)=\lambda'\int_{\R^{n}}y(1_{B_{i}'}(y)-1_{B_{j}'}(y))d\gamma_{n}(y),
\end{equation}
such that
\begin{equation}\label{five0.9}
\int_{B_{p}'}xd\gamma_{n}(x)\in\Pi,\,\forall\,p\in\{i,j\},
\end{equation}
and such that
\begin{equation}\label{five0.3}
\begin{aligned}
&\{x\in B_{i}'\cup B_{j}'\colon 1_{A_{i}}(x)-1_{A_{j}}(x)\neq 1_{B_{i}'}(x)-1_{B_{j}'}(x)\}\\
&\quad\subset\{x\in B_{i}'\cup B_{j}'\colon\absf{d(x,(\partial B_{i}')\cup(\partial B_{j}'))}<\eta\vee\vnorm{x}_{2}\geq\sqrt{-2\log\eta}+(\rho+\eta)\sqrt{-2\log\rho}\}.
\end{aligned}
\end{equation}
Then, for each pair $i,j\in\{1,2,\ldots,k\}$ with $i\neq j$, there exists $\lambda''>0$ and there exists a regular simplicial conical partition $\{B_{p}''\}_{p=1}^{k}\subset\R^{n}$ such that $\int_{\R^{n}}y(1_{A_{i}}(y)-1_{A_{j}}(y))d\gamma_{n}(y)=\lambda''\int_{\R^{n}}y(1_{B_{i}''}(y)-1_{B_{j}''}(y))d\gamma_{n}(y)$, such that
$\int_{B_{p}''}xd\gamma_{n}(x)\in\Pi,\,\forall\,p\in\{i,j\}$, and such that
\begin{equation}\label{five0.4}
\begin{aligned}
&\{x\in B_{i}''\cup B_{j}''\colon 1_{A_{i}}(x)-1_{A_{j}}(x)\neq 1_{B_{i}''}(x)-1_{B_{j}''}(x)\}\\
&\quad\subset\{x\in B_{i}''\cap B_{j}''\colon\absf{d(x,(\partial B_{i}'')\cup(\partial B_{j}''))}<\rho\eta\vee\vnorm{x}_{2}\geq\sqrt{-2\log(\rho\eta)}+1\}.
\end{aligned}
\end{equation}
\end{lemma}
\begin{figure}[h!]
\centering
\def\svgwidth{\figwidtha} 
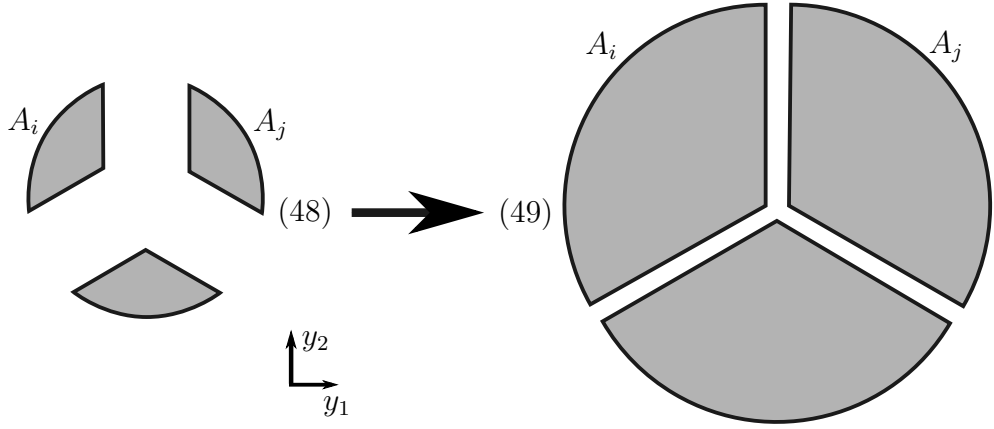
\caption{Depiction of Lemma \ref{lemma8}}
\label{fig0}
\end{figure}

\begin{proof}
Fix $i,j\in\{1,2,\ldots,k\}$ with $i\neq j$.  By applying a rotation to $\R^{n}$, we assume that $B_{i}'\cap B_{j}'\subset\{x\in\R^{n}\colon x_{1}=0\}$ and $B_{i}'\subset\{x\in\R^{n}\colon x_{1}\geq0\}$.  Assume that \eqref{five0.3} and \eqref{five0.5} hold. Let $n_{i}'\in\R^{n}$ denote the interior unit normal of $B_{i}'$ such that $n_{i}'$ is normal to $(\partial B_{i}')\setminus B_{j}'$, and let $n_{j}'\in\R^{n}$ denote the interior unit normal of $B_{j}'$ such that $n_{j}'$ is normal to $(\partial B_{j}')\setminus B_{i}'$.  Define $B_{i},B_{j}$ such that
\begin{equation}\label{five7.5}
\begin{aligned}
&B_{i}=B_{i,\frac{2\eta}{\sqrt{-2\log\eta}}}\colonequals B_{i}'\cup\{x\in\R^{n}\colon x_{1}\geq0\wedge\langle n_{i}',x/\vnorm{x}_{2}\rangle\geq-2\eta/\sqrt{-2\log\eta}\},\\
&B_{j}=B_{j,\frac{2\eta}{\sqrt{-2\log\eta}}}\colonequals B_{j}'\cup\{x\in\R^{n}\colon x_{1}\leq0\wedge\langle n_{j}',x/\vnorm{x}_{2}\rangle\geq-2\eta/\sqrt{-2\log\eta}\}.
\end{aligned}
\end{equation}
Let $x=(x_{1},\ldots,x_{n})\in B_{i}\cup B_{j}$.  If $x_{1}<\sqrt{n}\sqrt{1-\rho^{2}}/\rho$, then
$$\gamma_{n}\left(\delta_{\frac{(B_{i}\cap B_{j})-x\rho}{\sqrt{1-\rho^{2}}}}\right)
\geq\frac{e^{-n/2}}{\sqrt{2\pi}}\int_{\sqrt{n}}^{\infty}e^{-t^{2}/2}dt/\sqrt{2\pi}
\geq\frac{e^{-n/2}}{2\pi}\frac{1}{2\sqrt{n}}e^{-n/2}
\geq \frac{1}{100\sqrt{n}}e^{-n}.$$
So, using Lemma \ref{lemma4}, \eqref{six1.1}, and $\rho<10^{-5}n^{-3/2}e^{-n}$, if $x\in B_{i}\cup B_{j}$ then
\begin{equation}\label{five6}
\rho^{-1}\sign(x_{1})\cdot LT_{\rho}(1_{B_{i}}-1_{B_{j}})(x)\geq
\begin{cases}
\frac{1}{9}\abs{x_{1}}e^{-x_{1}^{2}\rho^{2}/(2(1-\rho^{2}))}&,\abs{x_{1}}\leq1\,\vee\, x_{2}\geq0\\
\frac{1}{9}\frac{\abs{x_{1}}}{\rho\abs{x_{2}}}e^{-\vnorm{x}_{2}^{2}\rho^{2}/(2(1-\rho^{2}))}
&,\rho x_{2}\leq-1/\sqrt{3}.
\end{cases}
\end{equation}

Let $\sigma\colon\R^{n}\to\R^{n}$ be a rotation such that the $x_{1}$-axis is fixed.  For any such rotation, let
\begin{equation}\label{five6.0}
g(x)=g_{\sigma}(x)\colonequals1_{A_{i}}(\sigma x)-1_{A_{j}}(\sigma x)-(1_{B_{i}}(\sigma x)-1_{B_{j}}(\sigma x)).
\end{equation}
By \eqref{five0.5}, and since $B_{i}\cup B_{j}$ is symmetric with respect to reflection across $B_{i}\cap B_{j}\subset\{x\in\R^{n}\colon x_{1}=0\}$, $\exists$ $\lambda>0$ such that $\int_{\R^{n}} y(1_{A_{i}}(y)-1_{A_{j}}(y))d\gamma_{n}(y)=\lambda\int_{\R^{n}} y(1_{B_{i}}(y)-1_{B_{j}}(y))d\gamma_{n}(y)$.  So $\int_{\R^{n}}y_{2}g(y)d\gamma_{n}(y)=0$, for all such rotations $\sigma$.  For all $x\in\R^{n}$, and for all rotations $\sigma\colon\R^{n}\to\R^{n}$ fixing the $x_{1}$-axis,
\begin{equation}\label{five7.0}
\abs{LT_{\rho}(1_{A_{i}}-1_{A_{j}})(\sigma x)-LT_{\rho}(1_{B_{i}}-1_{B_{j}})(\sigma x)}\leq\abs{LT_{\rho}g(x)}.
\end{equation}

\noindent\embolden{Step 1.  An estimate for large $x$.}

Let $\sigma\colon\R^{n}\to\R^{n}$ be any rotation fixing the $x_{1}$-axis.  By \eqref{five6.0}, $\abs{g}\leq2$.  Applying \eqref{five0.3} and \eqref{five7.5} and the inclusion-exclusion principle, $g=0$ on the set
\begin{flalign*}
&\{y\in\R^{n}\colon d(\sigma y-\rho x,(\partial B_{i}')\cup(\partial B_{j}'))>\eta+3\eta\\
&\qquad\wedge\vnorm{\sigma y-\rho x}_{2}\leq\sqrt{-2\log\eta}+(\rho+\eta)\sqrt{-2\log\rho}\}.
\end{flalign*}

Let $x\in\R^{n}$ with $\vnorm{x}_{2}^{2}\leq-4\log(\eta\rho)$.  Since $0<\eta<\rho$, we have $\rho\vnorm{x}_{2}\leq-4\rho\log(\eta\rho)\leq-8\rho\log\eta$.
By \eqref{five0.3}, \eqref{five0.9} and the inclusion-exclusion principle, $g\neq0$ only on the following sets: $\{y\in\R^{n}\colon\absf{d(\sigma y-\rho x,(\partial B_{i}')\cup(\partial B_{j}'))}\leq4\eta/\sqrt{1-\rho^{2}}\}$ and $\{y\in\R^{n}\colon\vnorm{\sigma y-\rho x}_{2}\geq\sqrt{-2\log\eta}/\sqrt{1-\rho^{2}}\}$.  Then Lemma \ref{lemma6.1} says
\begin{flalign*}
&\sup_{\substack{t_{1}\in[\min(x_{1},0),\max(x_{1},0)]\\\,t_{2}\in[\min(x_{2},0),\max(x_{2},0)]\\ \alpha\in[0,\rho]}}
\bigg|\int_{\R^{n}}\sum_{\substack{\ell\in\N^{n}\colon\\0\leq\abs{\ell}\leq3}}
\prod_{i=1}^{n}\abs{y_{i}}^{\ell_{i}}g((t_{1},t_{2},0,\ldots,0)\alpha+y\sqrt{1-\alpha^{2}})d\gamma_{n}(y)\bigg|\\
&\qquad\leq500000n^{3}4\eta+200(n+2)!((-2(1-2\rho)\log\eta)^{(n+1)/2}+1)\eta^{1-2\rho}.
\end{flalign*}
Using Lemma \ref{lemma7} and \eqref{five7.0},
\begin{equation}\label{five8}
\begin{aligned}
&\vnorm{x}_{2}^{2}\leq-4\log(\rho\eta)\,\wedge\,\abs{x_{1}}\geq(\rho\eta)^{1/3}\\
&\Longrightarrow\rho^{-1}\abs{LT_{\rho}(1_{A_{i}}-1_{A_{j}})(x_{1},x_{2},0,\ldots,0)-LT_{\rho}(1_{B_{i}}-1_{B_{j}})(x_{1},x_{2},0,\ldots,0)}\\
&\qquad\leq[\abs{x_{1}}+2\rho(\abs{x_{2}}^{2}+(n+1)\rho\abs{x_{2}}+\abs{x_{1}x_{2}}+2n)]\\
&\qquad\qquad\cdot[500000n^{3}4\eta+200(n+2)!((-2(1-2\rho)\log\eta)^{(n+1)/2}+1)\eta^{1-2\rho}]\\
&\qquad<10^{7}(n+2)!(-2\log\eta)^{(n+5)/2}\eta^{1-2\rho}.
\end{aligned}
\end{equation}

Also, by \eqref{five6}, and using that $0<\eta<\rho<10^{-5}n^{-3/2}e^{-n}$,
\begin{equation}\label{five8.1}
\begin{aligned}
&\vnorm{x}_{2}^{2}\leq-4\log(\rho\eta)\,\wedge\, x\in B_{i}\cup B_{j}\\
&\qquad\Longrightarrow\rho^{-1}\sign(x_{1})\cdot(LT_{\rho}(1_{B_{i}}-1_{B_{j}})(x))>
\begin{cases}
\frac{1}{9}\abs{x_{1}}(\rho\eta)^{2\rho^{2}/(1-\rho^{2})}&,\abs{x_{1}}\leq1\,\vee\,x_{2}\geq0\\
\frac{1}{9}\frac{\abs{x_{1}}}{\rho\abs{x_{2}}}(\rho\eta)^{2\rho^{2}/(1-\rho^{2})}&, \rho x_{2}\leq-1/\sqrt{3}
\end{cases}.
\end{aligned}
\end{equation}

Combining \eqref{five8} and \eqref{five8.1}, using \eqref{five6.0} and $0<\eta<\rho<e^{-20(n+1)^{10^{12}(n+2)!}}$,
\begin{equation}\label{five8.11}
\begin{aligned}
&\abs{x_{1}}\geq(\rho\eta)^{1/3}
\wedge\vnorm{x}_{2}^{2}\leq-4\log(\rho\eta)
\wedge x\in B_{i}\cup B_{j}\\
&\qquad\Longrightarrow\rho^{-1}\abs{LT_{\rho}(1_{A_{i}}-1_{A_{j}})(x)-LT_{\rho}(1_{B_{i}}-1_{B_{j}})(x)}
\leq\eta^{4/5}\\
&\qquad\quad\wedge\,\rho^{-1}\sign(x_{1})\cdot LT_{\rho}(1_{B_{i}}-1_{B_{j}})(x)
\geq(\rho\eta)^{2\rho^{2}/(1-\rho^{2})}(\rho\eta)^{1/3}\min\left(1,\frac{1}{\rho\sqrt{-4\log(\eta\rho)}}\right).
\end{aligned}
\end{equation}

By \eqref{five8.11},
\begin{equation}\label{five8.2}
\abs{x_{1}}\geq(\rho\eta)^{1/3}
\,\wedge\,\vnorm{x}_{2}^{2}\leq-4\log(\rho\eta)
\,\wedge\,x\in B_{i}\cup B_{j}
\Longrightarrow\rho^{-1}\sign(x_{1})\cdot LT_{\rho}(1_{A_{i}}-1_{A_{j}})(x)>0.
\end{equation}
Finally, applying \eqref{five8.2} to Lemma \ref{lemma1} for all $i',j'\in\{1,\ldots,k\}$, $i'\neq j'$, and using \eqref{five7.5} together with the inclusion-exclusion principle,
\begin{equation}\label{five8.3}
\abs{x_{1}}\geq(\rho\eta)^{1/3}
\,\wedge\,\vnorm{x}_{2}^{2}\leq-4\log(\rho\eta)
\,\wedge\,x\in B_{i}'\cup B_{j}'
\Longrightarrow\sign(x_{1})\cdot(1_{A_{i}}(x)-1_{A_{j}}(x))>0.
\end{equation}

\noindent\embolden{Step 2.  An estimate for small $x$.}

Let $\sigma\colon\R^{n}\to\R^{n}$ be any rotation fixing the $x_{1}$-axis.  For $x=(x_{1},\ldots,x_{n})\in\R^{n}$ with $\vnorm{x}_{2}^{2}\leq-2\log\rho$, we have $\rho\vnorm{x}_{2}\leq\rho\sqrt{-2\log\rho}$.  Suppose also that $\abs{x_{1}}\leq\eta$ and $x\in B_{i}\cup B_{j}$.  By \eqref{five7.5}, \eqref{five0.3}, \eqref{five8.3}, \eqref{five0.9} and the inclusion-exclusion principle, $g\neq0$ only on the following sets:
\begin{flalign*}
&\{y\in\R^{n}\colon\absf{d(\sigma y-\rho x,(B_{i}\cap B_{j})\cup[(B_{i}\cup B_{j})\setminus(B_{i}'\cup B_{j}')])}\leq\eta/\sqrt{1-\rho^{2}}\},\\
&\{y\in\R^{n}\colon\absf{d(\sigma y-\rho x,(B_{i}\cap B_{j})\cup[(B_{i}\cup B_{j})\setminus(B_{i}'\cup B_{j}')])}\leq(\rho\eta)^{1/4}\\
&\qquad\qquad\qquad\qquad\qquad\qquad\wedge\vnorm{\sigma y-\rho x}_{2}\geq\sqrt{-2\log\eta}\},\\
&\{y\in\R^{n}\colon\vnorm{\sigma y-\rho x}_{2}\geq(1+1/10)\sqrt{-3\log(\eta\rho)}/\sqrt{1-\rho^{2}}\}.
\end{flalign*}
We then apply Lemma \ref{lemma6.1} to get
\begin{flalign*}
&\sup_{\substack{t_{1}\in[\min(x_{1},0),\max(x_{1},0)]\\\,t_{2}\in[\min(x_{2},0),\max(x_{2},0)]\\ \alpha\in[0,\rho]}}
\bigg|\int_{\R^{n}}\sum_{\substack{\ell\in\N^{n}\colon\\0\leq\abs{\ell}\leq3}}
\prod_{i=1}^{n}\abs{y_{i}}^{\ell_{i}}g((t_{1},t_{2},0,\ldots,0)\alpha+y\sqrt{1-\alpha^{2}})d\gamma_{n}(y)\bigg|\\
&\qquad\leq500000n^{3}\eta
+500000n^{3}(\rho\eta)^{1/4}(-2(1-\rho)^{2}\log\eta+1)\eta^{(1-\rho)^{2}}\\
&\qquad\quad+200(n+2)!((-3\log(\rho\eta))^{(n+1)/2}+1)(\rho\eta)^{3/2}+1600(n+2)!2\eta/\sqrt{-2\log\eta}.
\end{flalign*}
%

So, using Lemma \ref{lemma7}, \eqref{five7.0} and $0<\eta<\rho<e^{-20(n+1)^{10^{12}n^{3}(n+2)!}}$,
\begin{equation}\label{five7}
\begin{aligned}
&\rho^{3/4}\eta\leq\abs{x_{1}}\leq\eta\,\wedge\,\vnorm{x}_{2}^{2}\leq-2\log\rho\,\wedge\,x\in B_{i}\cup B_{j}\\
&\Longrightarrow\rho^{-1}\abs{LT_{\rho}(1_{A_{i}}-1_{A_{j}})(x_{1},x_{2},0,\ldots,0)-LT_{\rho}(1_{B_{i}}-1_{B_{j}})(x_{1},x_{2},0,\ldots,0)}\\
&\qquad\leq[\abs{x_{1}}+2\rho(\abs{x_{2}}^{2}+(n+1)\rho\abs{x_{2}}+\abs{x_{1}x_{2}}+2n)]\\
&\qquad\qquad\cdot\left[500000n^{3}3\eta+500000n^{3}(\rho\eta)^{1/4}(-2(1-\rho)^{2}\log\eta+1)\eta^{(1-\rho)^{2}}\right.\\
&\qquad\qquad\,\,\,\,\,\left.+200(n+2)!((-3\log(\rho\eta))^{(n+1)/2}+1)(\rho\eta)^{3/2}+1600(n+2)!2\eta/\sqrt{-2\log\eta}\right]\\
&\qquad<\frac{1}{10}\eta\rho^{3/4}.
\end{aligned}
\end{equation}
Also, by \eqref{five6},
\begin{equation}\label{five7.1}
\begin{aligned}
&\eta\rho^{3/4}\leq\abs{x_{1}}\leq\eta
\,\wedge\,\vnorm{x}_{2}^{2}\leq-2\log\rho
\,\wedge\,x\in B_{i}\cup B_{j}\\
&\qquad\qquad\qquad\Longrightarrow\rho^{-1}\sign(x_{1})\cdot LT_{\rho}(1_{B_{i}}-1_{B_{j}})(x)>\frac{1}{10}\abs{x_{1}}.
\end{aligned}
\end{equation}

Combining \eqref{five7} and \eqref{five7.1}, and using \eqref{five6.0},
\begin{equation}\label{five7.2}
\eta\rho^{3/4}\leq\abs{x_{1}}\leq\eta
\,\wedge\,\vnorm{x}_{2}^{2}\leq-2\log\rho
\,\wedge\,x\in B_{i}\cup B_{j}
\Longrightarrow\rho^{-1}\sign(x_{1})\cdot LT_{\rho}(1_{A_{i}}-1_{A_{j}})(x)>0.
\end{equation}
Similarly, by \eqref{five7} and \eqref{five6.0}, we have the following estimate.
\begin{equation}\label{five8.33}
\eta\leq\abs{x_{1}}\leq(\rho\eta)^{1/3}
\,\wedge\,\vnorm{x}_{2}^{2}\leq1
\,\wedge\,x\in B_{i}\cup B_{j}
\Longrightarrow\rho^{-1}\sign(x_{1})\cdot LT_{\rho}(1_{A_{i}}-1_{A_{j}})(x)>0.
\end{equation}
Finally, applying \eqref{five7.2} to Lemma \ref{lemma1} for all $i',j'\in\{1,\ldots,k\}$, and using \eqref{five7.5} together with the inclusion-exclusion principle, \eqref{five8.2} and \eqref{five8.33},
\begin{equation}\label{five7.3}
\eta\rho^{3/4}\leq\abs{x_{1}}\leq\eta
\,\wedge\,\vnorm{x}_{2}^{2}\leq-2\log\rho
\,\wedge\,x\in B_{i}'\cup B_{j}'
\Longrightarrow\sign(x_{1})\cdot(1_{A_{i}}(x)-1_{A_{j}}(x))>0.
\end{equation}

\noindent\embolden{Step 3.  An estimate for intermediate values of $x$.}

In summary, \eqref{five7.3} and \eqref{five8.3} improve our initial assumption \eqref{five0.3}.  We now repeat the above procedure with the improved assumptions.  Before continuing, we need to redefine $B_{i},B_{j}$.  Via \eqref{five7.5}, let
\begin{equation}\label{five60}
B_{i}\colonequals B_{i,\min\left(\frac{2\rho^{3/4}\eta}{\sqrt{-2\log\rho}},\frac{2\eta}{\sqrt{-2\log\eta}}\right)},\,\,
B_{j}\colonequals B_{j,\min\left(\frac{2\rho^{3/4}\eta}{\sqrt{-2\log\rho}},\frac{2\eta}{\sqrt{-2\log\eta}}\right)}
\end{equation}

\begin{figure}[h!]
\centering
\def\svgwidth{\figwidth} 
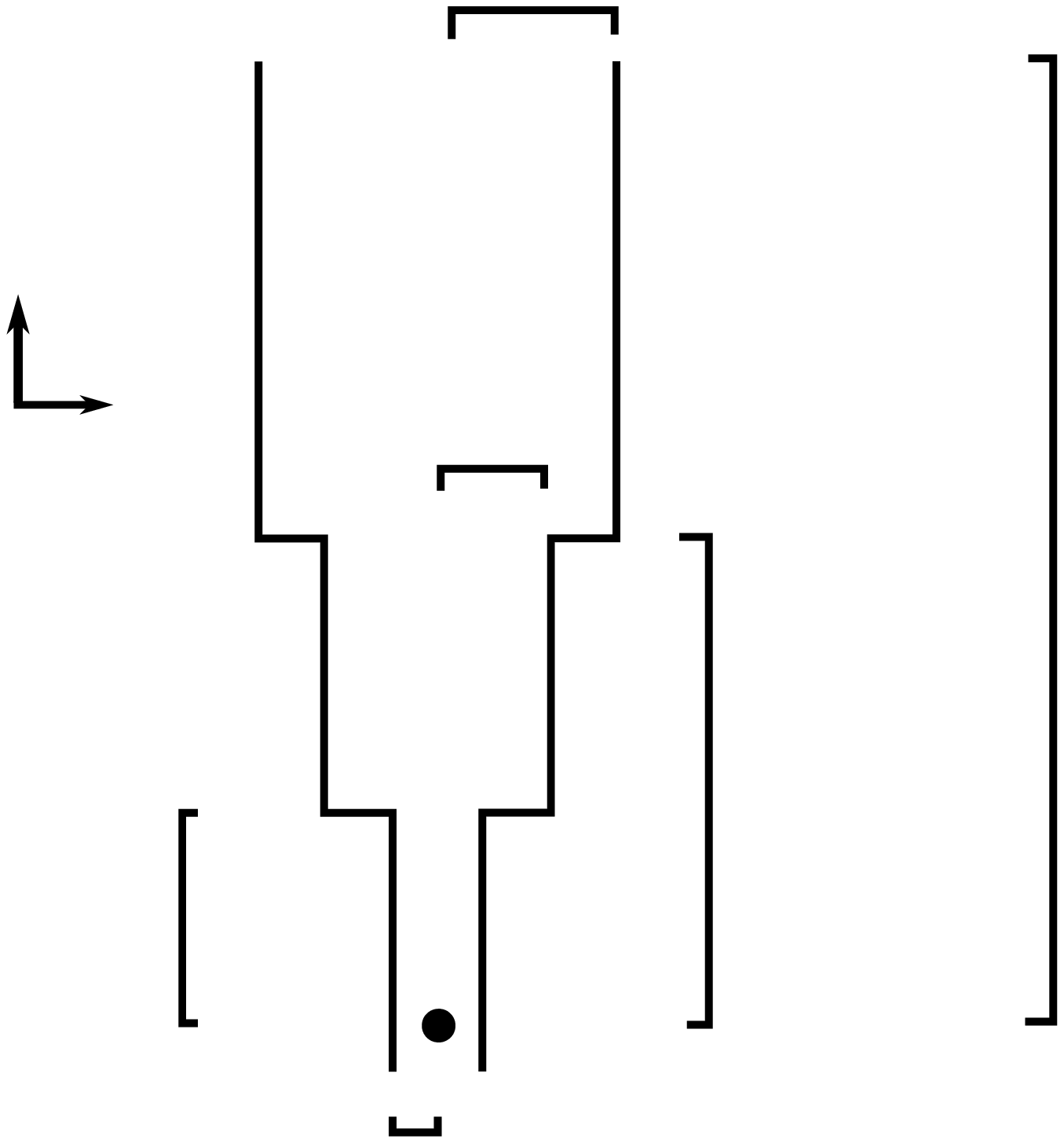
\caption{Integration regions where $g\neq0$, near $B_{i}'\cap B_{j}'$ for \eqref{five8.4}.}
\label{fig1}
\end{figure}

Let $x$ with $\vnorm{x}_{2}^{2}\leq-4\log\rho$, $\eta\rho\leq\abs{x_{1}}\leq\eta$.  Let $B\colonequals(B_{i}\cap B_{j})\cup[(B_{i}\cup B_{j})\setminus(B_{i}'\cup B_{j}')]$.  Suppose $x\in B_{i}\cup B_{j}$ also.  By \eqref{five7.3}, \eqref{five0.3}, \eqref{five8.3}, \eqref{five0.9} and the inclusion-exclusion principle, $g\neq0$ only on the following sets
\begin{flalign*}
&\{y\in\R^{n}\colon\absf{d(\sigma y-\rho x,B)}\leq\eta\rho^{3/4}/\sqrt{1-\rho^{2}}\},\\
&\{y\in\R^{n}\colon\absf{d(\sigma y-\rho x,B)}\leq\eta/\sqrt{1-\rho^{2}},\vnorm{\sigma y-\rho x}_{2}>\sqrt{-2\log\rho}/\sqrt{1-\rho^{2}}\},\\
&\{y\in\R^{n}\colon\absf{d(\sigma y-\rho x,B)}\leq(\rho\eta)^{1/4}/\sqrt{1-\rho^{2}},\vnorm{\sigma y-\rho x}_{2}\geq\sqrt{-2\log\eta}/\sqrt{1-\rho^{2}}\},\\
&\{y\in\R^{n}\colon\vnorm{\sigma y-\rho x}_{2}>(1+1/10)\sqrt{-3\log(\rho\eta)}/\sqrt{1-\rho^{2}}\}.
\end{flalign*}
We then apply Lemma \ref{lemma6.1} to get
\begin{equation}\label{five8.4}
\begin{aligned}
&\sup_{\substack{t_{1}\in[\min(x_{1},0),\max(x_{1},0)]\\\,t_{2}\in[\min(x_{2},0),\max(x_{2},0)]\\ \alpha\in[0,\rho]}}
\bigg|\int_{\R^{n}}\sum_{\substack{\ell\in\N^{n}\colon\\0\leq\abs{\ell}\leq3}}
\prod_{i=1}^{n}\abs{y_{i}}^{\ell_{i}}g((t_{1},t_{2},0,\ldots,0)\alpha+y\sqrt{1-\alpha^{2}})d\gamma_{n}(y)\bigg|\\
&\qquad\leq\eta\rho^{3/4}500000n^{3}
+500000n^{3}\eta(-2(1-\sqrt{2}\rho)^{2}\log\rho+1)\rho^{(1-\sqrt{2}\rho)^{2}}\\
&\qquad\quad+500000n^{3}(\eta\rho)^{1/4}(-2(1-\sqrt{2}\rho)^{2}\log\eta+1)\eta^{(1-\sqrt{2}\rho)^{2}}\\
&\qquad\quad+200(n+2)!((-3\log(\eta\rho))^{(n+1)/2}+1)(\rho\eta)^{3/2}\\
&\qquad\quad+1600(n+2)!\min(2\rho^{3/4}\eta/\sqrt{-2\log\rho},2\eta/\sqrt{-2\log\eta}\,).
\end{aligned}
\end{equation}

Applying \eqref{five8.4} to Lemma \ref{lemma7}, using \eqref{five6.0} and $\eta<\rho<e^{-20(n+1)^{10^{12}n^{3}(n+2)!}}$,
\begin{equation}\label{five9}
\begin{aligned}
&\vnorm{x}_{2}^{2}\leq-4\log(\rho)\,\wedge\,\eta\rho\leq\abs{x_{1}}\leq\eta\,\wedge x\in B_{i}\cup B_{j}\\
&\qquad\Longrightarrow\rho^{-1}\abs{LT_{\rho}(1_{A_{i}}-1_{A_{j}})(x)-LT_{\rho}(1_{B_{i}}-1_{B_{j}})(x)}<\frac{1}{10}\rho\eta.
\end{aligned}
\end{equation}

Also, by \eqref{five6},
\begin{equation}\label{five10}
\rho\eta\leq\abs{x_{1}}\leq\eta
\,\wedge\,\vnorm{x}_{2}^{2}\leq-4\log\rho
\,\wedge\,x\in B_{i}\cup B_{j}
\Longrightarrow\rho^{-1}\sign(x_{1})\cdot LT_{\rho}(1_{B_{i}}-1_{B_{j}})(x)>\frac{1}{10}\rho\eta.
\end{equation}

So, combining \eqref{five9}, \eqref{five10} for all $i',j'\in\{1,\ldots,k\}$, $i'\neq j'$, Lemma \ref{lemma1}, \eqref{five60}, and by applying the inclusion-exclusion principle, \eqref{five8.2} and \eqref{five8.33},
\begin{equation}\label{five11}
\rho\eta\leq\abs{x_{1}}\leq\eta
\,\wedge\,\vnorm{x}_{2}^{2}\leq-4\log\rho
\,\wedge\,x\in B_{i}'\cup B_{j}'
\Longrightarrow \sign(x_{1})\cdot(1_{A_{i}}(x)-1_{A_{j}}(x))>0.
\end{equation}

\noindent\embolden{Step 4.  Iterating the estimate for intermediate values of $x$.}

The estimate \eqref{five11} now has a cascading effect on the estimates below.  From \eqref{five11},
$$
\rho^{.9}\eta\leq\abs{x_{1}}\leq\eta
\,\wedge\,\vnorm{x}_{2}^{2}\leq-4\log\rho
\,\wedge\,x\in B_{i}'\cup B_{j}'
\Longrightarrow\sign(x_{1})\cdot(1_{A_{i}}(x)-1_{A_{j}}(x))>0.
$$

This estimate can be iterated on itself.  Let $K\in\N$, $K\geq1$, and let $M\in\N$ with $0\leq M\leq\sqrt{K}$. Suppose $\rho^{.9K}>\eta^{1/5}$.  We prove by induction on $K$ and $M$ that
\begin{equation}\label{eight1}
\begin{aligned}
&2^{M^{2}}\eta\rho^{.9K}\leq\abs{x_{1}}\leq\eta
\,\wedge\,\vnorm{x}_{2}^{2}\leq-2^{M+2}\log\rho
\,\wedge\,x\in B_{i}'\cup B_{j}'\\
&\qquad\qquad\Longrightarrow\sign(x_{1})\cdot(1_{A_{i}}(x)-1_{A_{j}}(x))>0.
\end{aligned}
\end{equation}
We already verified the case $M=0,K=1$.  We assume that, for $0\leq m<M$, 
\begin{equation}\label{eight2}
\begin{aligned}
&2^{m^{2}}\eta\rho^{.9K}\leq\abs{x_{1}}\leq\eta
\,\wedge\,\vnorm{x}_{2}^{2}\leq-2^{m+2}\log\rho
\,\wedge\,x\in B_{i}'\cup B_{j}'\\
&\qquad\qquad\Longrightarrow\sign(x_{1})\cdot(1_{A_{i}}(x)-1_{A_{j}}(x))>0.
\end{aligned}
\end{equation}
Assume also that, for $M\leq m\leq \sqrt{K-1}$ and $K\geq1$,
\begin{equation}\label{eight3}
\begin{aligned}
&2^{m^{2}}\eta\rho^{.9(K-1)}\leq\abs{x_{1}}\leq\eta
\,\wedge\,\vnorm{x}_{2}^{2}\leq-2^{m+2}\log\rho
\,\wedge\,x\in B_{i}'\cup B_{j}'\\
&\qquad\qquad\Longrightarrow\sign(x_{1})\cdot(1_{A_{i}}(x)-1_{A_{j}}(x))>0.
\end{aligned}
\end{equation}
We will conclude that \eqref{eight2} holds for $m=M$, i.e.
\begin{equation}\label{eight4}
\begin{aligned}
&2^{M^{2}}\eta\rho^{.9K}\leq\abs{x_{1}}\leq\eta
\,\wedge\,\vnorm{x}_{2}^{2}\leq-2^{M+2}\log\rho
\,\wedge\,x\in B_{i}'\cup B_{j}'\\
&\qquad\qquad\Longrightarrow\sign(x_{1})\cdot(1_{A_{i}}(x)-1_{A_{j}}(x))>0.
\end{aligned}
\end{equation}

We repeat the calculations \eqref{five60} through \eqref{five11}.  Redefine $B_{i},B_{j}$ so that

\begin{equation}\label{five61}
B_{i}\colonequals B_{i,\min\left(\frac{2\eta\rho^{.9K}}{\sqrt{-4\log\rho}},\frac{2\eta}{\sqrt{-2\log\eta}}\right)}\,,\quad
B_{j}\colonequals B_{j,\min\left(\frac{2\eta\rho^{.9K}}{\sqrt{-4\log\rho}},\frac{2\eta}{\sqrt{-2\log\eta}}\right)}.
\end{equation}

If $M>0$, we use \eqref{eight2} for $m=M-1$.  For any $M\geq0$, we use \eqref{eight3} for $M\leq m\leq\sqrt{K}$.  Let $x,M$ with $\vnorm{x}_{2}^{2}\leq-2^{M+2}\log\rho\leq-2^{\sp{\lfloor\sqrt{K}\rfloor+2}}\log\rho\leq-4\log\eta$, $2^{M^{2}}\eta\rho^{.9K}\leq\abs{x_{1}}\leq\eta$, $x\in B_{i}\cup B_{j}$.  Let $B\colonequals(B_{i}\cap B_{j})\cup[(B_{i}\cup B_{j})\setminus(B_{i}'\cup B_{j}')]$. Combining \eqref{eight2}, \eqref{eight3}, \eqref{five0.3}, \eqref{five8.3}, \eqref{five0.9} and the inclusion-exclusion principle, $g\neq0$ only on the following sets
\begin{flalign*}
&\{y\in\R^{n}\colon\absf{d(\sigma y-\rho x,B)}\leq\min(M,1)\cdot2^{(M-1)^{2}}\eta\rho^{.9K}/\sqrt{1-\rho^{2}}\},\\
&\cup_{M\leq m\leq\lfloor\sqrt{K-1}\rfloor}\{y\in\R^{n}\colon\absf{d(\sigma y-\rho x,B)}\leq2^{m^{2}}\eta\rho^{.9(K-1)}/\sqrt{1-\rho^{2}},\\
&\qquad\qquad\qquad\qquad\vnorm{\sigma y-\rho x}_{2}>\min(m,1)\cdot\sqrt{-2^{m+1}\log\rho}/\sqrt{1-\rho^{2}}\},\\
&\{y\in\R^{n}\colon\absf{d(\sigma y-\rho x,B)}\leq\eta,\vnorm{\sigma y-\rho x}_{2}>\sqrt{-2^{\sp{\lfloor\sqrt{K-1}\rfloor+2}}\log\rho}/\sqrt{1-\rho^{2}}\},\\
&\{y\in\R^{n}\colon\absf{d(\sigma y-\rho x,B)}\leq(\rho\eta)^{1/4}/\sqrt{1-\rho^{2}},\vnorm{\sigma y-\rho x}_{2}\geq\sqrt{-2\log\eta}/\sqrt{1-\rho^{2}}\},\\
&\{y\in\R^{n}\colon\vnorm{\sigma y-\rho x}_{2}>(1+1/10)\sqrt{-3\log(\rho\eta)}/\sqrt{1-\rho^{2}}\}.
\end{flalign*}
We then apply Lemma \ref{lemma6.1} to get
\begin{equation}\label{eight6}
\begin{aligned}
&\sup_{\substack{t_{1}\in[\min(x_{1},0),\max(x_{1},0)]\\\,t_{2}\in[\min(x_{2},0),\max(x_{2},0)]\\ \alpha\in[0,\rho]}}
\bigg|\int_{\R^{n}}\sum_{\substack{\ell\in\N^{n}\colon\\0\leq\abs{\ell}\leq3}}
\prod_{i=1}^{n}\abs{y_{i}}^{\ell_{i}}g((t_{1},t_{2},0,\ldots,0)\alpha+y\sqrt{1-\alpha^{2}})d\gamma_{n}(y)\bigg|\\
&\qquad\leq\min(M,1)\cdot2^{(M-1)^{2}}\eta\rho^{.9K}500000n^{3}\\
&\qquad\quad+\eta(-(1-\rho)^{2}2^{\sp{\lfloor\sqrt{K-1}\rfloor+2}}\log\rho+1)\rho^{(1-\rho)^{2}2^{\sp{\lfloor\sqrt{K-1}\rfloor+1}}}500000n^{3}\\
&\qquad\quad+500000n^{3}\sum_{m=M}^{\lfloor\sqrt{K-1}\rfloor}\eta\rho^{.9(K-1)}
(-(1-\rho)^{2}2^{m+1}\log\rho+1)2^{m^{2}}\rho^{(1-\rho)^{2}2^{m}\cdot\min(m,1)}\\
&\qquad\quad+500000n^{3}(\eta\rho)^{1/4}(-2(1-\sqrt{2}\rho)^{2}\log\eta+1)\eta^{(1-\sqrt{2}\rho)^{2}}\\
&\qquad\quad+200(n+2)!((-3\log(\eta\rho))^{(n+1)/2}+1)(\rho\eta)^{3/2}\\
&\qquad\quad+1600(n+2)!\min(2\eta\rho^{.9K}/\sqrt{-4\log\rho},2\eta/\sqrt{-2\log\eta}\,).
\end{aligned}
\end{equation}

Applying \eqref{eight6} to Lemma \ref{lemma7}, using \eqref{five6.0}, $\eta<\rho<e^{-20(n+1)^{10^{12}n^{3}(n+2)!}}$, and $\rho^{.9K}>\eta^{1/5}$,
\begin{equation}\label{eight7}
\begin{aligned}
&\vnorm{x}_{2}^{2}\leq-2^{M+2}\log(\rho)\,\wedge\,
2^{M^{2}}\eta\rho^{.9K}\leq\abs{x_{1}}\leq\eta\,\wedge\,x\in B_{i}\cup B_{j}\\
&\qquad\Longrightarrow\rho^{-1}\abs{LT_{\rho}(1_{A_{i}}-1_{A_{j}})(x)-LT_{\rho}(1_{B_{i}}-1_{B_{j}})(x)}<\frac{1}{10}2^{M^{2}}\eta\rho^{.9K}.
\end{aligned}
\end{equation}

Also, by \eqref{five6},
\begin{equation}\label{eight8}
\begin{aligned}
&2^{M^{2}}\eta\rho^{.9K}\leq\abs{x_{1}}\leq\eta
\,\wedge\,\vnorm{x}_{2}^{2}\leq-2^{M+2}\log\rho
\,\wedge\,x\in B_{i}\cup B_{j}\\
&\qquad\qquad\Longrightarrow\rho^{-1}\sign(x_{1})\cdot LT_{\rho}(1_{B_{i}}-1_{B_{j}})(x)>\frac{1}{10}2^{M^{2}}\eta\rho^{.9K}.
\end{aligned}
\end{equation}

So, combining \eqref{eight7}, \eqref{eight8} for all $i',j'\in\{1,\ldots,k\}$, $i'\neq j'$, Lemma \ref{lemma1}, \eqref{five61}, and by applying the inclusion-exclusion principle, \eqref{five8.2} and \eqref{five8.33},
\begin{equation}\label{eight9}
2^{M^{2}}\eta\rho^{.9K}\leq\abs{x_{1}}\leq\eta
\,\wedge\,\vnorm{x}_{2}^{2}\leq-2^{M+2}\log\rho
\,\wedge\,x\in B_{i}'\cup B_{j}'
\Longrightarrow \sign(x_{1})\cdot(1_{A_{i}}(x)-1_{A_{j}}(x))>0.
\end{equation}
Thus, the inductive step is completed.

Let $K\in\N$ with $-2\log\eta\leq-2^{\sp{\lfloor\sqrt{K}\rfloor+2}}\log\rho\leq-4\log\eta$.  Then \eqref{eight9} and \eqref{five0.3} say that

\begin{equation}\label{eight10}
2^{K}\eta\rho^{.9K}\leq\abs{x_{1}}\leq1
\,\wedge\,\vnorm{x}_{2}^{2}\leq-2\log\eta
\,\wedge\,x\in B_{i}'\cup B_{j}'
\Longrightarrow \sign(x_{1})\cdot(1_{A_{i}}(x)-1_{A_{j}}(x))>0.
\end{equation}

\noindent\embolden{Step 5.  Another iterative estimate, now for larger values of $x$.}

We perform another induction, though this time we hold $K$ fixed and use the additional ingredient \eqref{eight10}.  Let $M,R\in\N$ with $0\leq M\leq\sqrt{K}$, $R\geq0$ such that $\rho^{.9(K+R)}>\eta^{1/5}$.  We will induct on $M$ and $R$.  We assume that, for $0\leq m<M$, 
\begin{equation}\label{eight12}
\begin{aligned}
&2^{m^{2}}\eta\rho^{.9(K+R)}\leq\abs{x_{1}}\leq\eta
\,\wedge\,\vnorm{x}_{2}^{2}\leq-2^{m+2}\log\rho
\,\wedge\,x\in B_{i}'\cup B_{j}'\\
&\qquad\qquad\Longrightarrow\sign(x_{1})\cdot(1_{A_{i}}(x)-1_{A_{j}}(x))>0.
\end{aligned}
\end{equation}
We know that the case $R=0,0\leq M\leq\sqrt{K}$ of \eqref{eight12} holds by \eqref{eight4}.  We therefore assume that $R\geq1$.  Assume also that, for $M\leq m\leq \sqrt{K}$, 
\begin{equation}\label{eight13}
\begin{aligned}
&2^{m^{2}}\eta\rho^{.9(K+R-1)}\leq\abs{x_{1}}\leq\eta
\,\wedge\,\vnorm{x}_{2}^{2}\leq-2^{m+2}\log\rho
\,\wedge\,x\in B_{i}'\cup B_{j}'\\
&\qquad\qquad\Longrightarrow\sign(x_{1})\cdot(1_{A_{i}}(x)-1_{A_{j}}(x))>0.
\end{aligned}
\end{equation}
We will conclude that \eqref{eight12} holds for $m=M$, i.e.
\begin{equation}\label{eight14}
\begin{aligned}
&2^{M^{2}}\eta\rho^{.9(K+R)}\leq\abs{x_{1}}\leq\eta
\,\wedge\,\vnorm{x}_{2}^{2}\leq-2^{M+2}\log\rho
\,\wedge\,x\in B_{i}'\cup B_{j}'\\
&\qquad\qquad\Longrightarrow\sign(x_{1})\cdot(1_{A_{i}}(x)-1_{A_{j}}(x))>0.
\end{aligned}
\end{equation}
Redefine $B_{i},B_{j}$ so that
\begin{equation}\label{five67}
B_{i}\colonequals B_{i,\min\left(\frac{2\eta\rho^{.9(K+R)}}{\sqrt{-4\log\rho}},\frac{2\eta}{\sqrt{-2\log\eta}}\right)}\,,\quad
B_{j}\colonequals B_{j,\min\left(\frac{2\eta\rho^{.9(K+R)}}{\sqrt{-4\log\rho}},\frac{2\eta}{\sqrt{-2\log\eta}}\right)}.
\end{equation}

If $M>0$, we use \eqref{eight12} for $m=M-1$.  For any $M\geq0$, we also use \eqref{eight13} for $M\leq m\leq\sqrt{K}$.  Let $x,M$ with $\vnorm{x}_{2}^{2}\leq-2^{M+2}\log\rho\leq-2^{\sp{\lfloor\sqrt{K}\rfloor+2}}\log\rho\leq-4\log\eta$, $2^{M^{2}}\eta\rho^{.9K}\leq\abs{x_{1}}\leq\eta$, $x\in B_{i}\cup B_{j}$.  Let $B\colonequals(B_{i}\cap B_{j})\cup[(B_{i}\cup B_{j})\setminus(B_{i}'\cup B_{j}')]$. Combining \eqref{eight12}, \eqref{eight13}, \eqref{five0.3}, \eqref{five8.3}, \eqref{five0.9}, and the fact that $-2\log\eta\leq-2^{\sp{\lfloor\sqrt{K}\rfloor+2}}\log\rho\leq-4\log\eta$, we conclude that $g\neq0$ only on the following sets:
\begin{flalign*}
&\{y\in\R^{n}\colon\absf{d(\sigma y-\rho x,B)}\leq
\min(M,1)\cdot2^{(M-1)^{2}}\eta\rho^{.9(K+R)}/\sqrt{1-\rho^{2}}\},\\
&\cup_{M\leq m\leq\lfloor\sqrt{K}\rfloor}\{y\in\R^{n}\colon\absf{d(\sigma y-\rho x,B)}\leq
2^{m^{2}}\eta\rho^{.9(K+R-1)}/\sqrt{1-\rho^{2}},\\
&\qquad\qquad\qquad\qquad\qquad\qquad\quad\vnorm{\sigma y-\rho x}_{2}>\min(m,1)\cdot\sqrt{-2^{m+1}\log\rho}/\sqrt{1-\rho^{2}}\},\\
&\{y\in\R^{n}\colon\absf{d(\sigma y-\rho x,B)}\leq(\rho\eta)^{1/4}/\sqrt{1-\rho^{2}},\vnorm{\sigma y-\rho x}_{2}\geq\sqrt{-2\log\eta}/\sqrt{1-\rho^{2}}\},\\
&\{y\in\R^{n}\colon\vnorm{\sigma y-\rho x}_{2}>(1+1/10)\sqrt{-3\log(\rho\eta)}/\sqrt{1-\rho^{2}}\}.
\end{flalign*}
We then apply Lemma \ref{lemma6.1} to get
\begin{equation}\label{eight16}
\begin{aligned}
&\sup_{\substack{t_{1}\in[\min(x_{1},0),\max(x_{1},0)]\\\,t_{2}\in[\min(x_{2},0),\max(x_{2},0)]\\ \alpha\in[0,\rho]}}
\bigg|\int_{\R^{n}}\sum_{\substack{\ell\in\N^{n}\colon\\0\leq\abs{\ell}\leq3}}
\prod_{i=1}^{n}\abs{y_{i}}^{\ell_{i}}g((t_{1},t_{2},0,\ldots,0)\alpha+y\sqrt{1-\alpha^{2}})d\gamma_{n}(y)\bigg|\\
&\qquad\leq\min(M,1)\cdot2^{(M-1)^{2}}\eta\rho^{.9(K+R)}500000n^{3}\\
&\qquad\quad+500000n^{3}\sum_{m=M}^{\lfloor\sqrt{K}\rfloor}\eta\rho^{.9(K+R-1)}
(-(1-\rho)^{2}2^{m+1}\log\rho+1)2^{m^{2}}\rho^{(1-\rho)^{2}2^{m}}\\
&\qquad\quad+500000n^{3}(\eta\rho)^{1/4}(-2(1-\sqrt{2}\rho)^{2}\log\eta+1)\eta^{(1-\sqrt{2}\rho)^{2}}\\
&\qquad\quad+200(n+2)!((-3\log(\eta\rho))^{(n+1)/2}+1)(\rho\eta)^{3/2}\\
&\qquad\quad+1600(n+2)!\min(2\eta\rho^{.9(K+R)}/\sqrt{-4\log\rho},2\eta/\sqrt{-2\log\eta}).
\end{aligned}
\end{equation}

Applying \eqref{eight16} to Lemma \ref{lemma7}, using \eqref{five6.0}, $\eta<\rho<e^{-20(n+1)^{10^{12}n^{3}(n+2)!}}$, and $\rho^{.9(K+R)}>\eta^{1/5}$,
\begin{equation}\label{eight17}
\begin{aligned}
&\vnorm{x}_{2}^{2}\leq-2^{M+2}\log(\rho)\,\wedge\,
2^{M^{2}}\eta\rho^{.9(K+R)}\leq\abs{x_{1}}\leq\eta\,\wedge\,x\in B_{i}\cup B_{j}\\
&\qquad\Longrightarrow\rho^{-1}\abs{LT_{\rho}(1_{A_{i}}-1_{A_{j}})(x)-LT_{\rho}(1_{B_{i}}-1_{B_{j}})(x)}<\frac{1}{10}2^{M^{2}}\eta\rho^{.9(K+R)}.
\end{aligned}
\end{equation}

Also, by \eqref{five6},
\begin{equation}\label{eight18}
\begin{aligned}
&2^{M^{2}}\eta\rho^{.9(K+R)}\leq\abs{x_{1}}\leq\eta
\,\wedge\,\vnorm{x}_{2}^{2}\leq-2^{M+2}\log\rho
\,\wedge\,x\in B_{i}\cup B_{j}\\
&\qquad\qquad\Longrightarrow\rho^{-1}\sign(x_{1})\cdot LT_{\rho}(1_{B_{i}}-1_{B_{j}})(x)>\frac{1}{10}2^{M^{2}}\eta\rho^{.9(K+R)}.
\end{aligned}
\end{equation}

So, combining \eqref{eight17}, \eqref{eight18} for all $i',j'\in\{1,\ldots,k\}$, $i'\neq j'$, Lemma \ref{lemma1}, \eqref{five67}, and by applying the inclusion-exclusion principle, \eqref{five8.2} and \eqref{five8.33},
\begin{equation}\label{eight19}
\begin{aligned}
&2^{M^{2}}\eta\rho^{.9(K+R)}\leq\abs{x_{1}}\leq\eta
\,\wedge\,\vnorm{x}_{2}^{2}\leq-2^{M+2}\log\rho
\,\wedge\,x\in B_{i}'\cup B_{j}'\\
&\qquad\qquad\Longrightarrow \sign(x_{1})\cdot(1_{A_{i}}(x)-1_{A_{j}}(x))>0.
\end{aligned}
\end{equation}
Thus, the inductive step is completed.  Let $M=\lfloor\sqrt{K}\rfloor$.  Let $R\in\N$ such that $\eta^{1/5}\leq\rho^{.9(K+R)}\leq\eta^{1/5}\rho^{-.9}$.  If no such $R$ exists, then $\rho^{.9K}<\eta^{1/5}$, so $\rho^{.45K}<\eta^{1/10}$, and \eqref{eight21} below holds by combining \eqref{eight10} and \eqref{five0.3}.  Otherwise, $R\geq0$, so \eqref{eight19} and \eqref{five0.3} say that
\begin{equation}\label{eight20}
2^{K}\eta^{6/5}\rho^{-.9}\leq\abs{x_{1}}\leq1
\,\wedge\,\vnorm{x}_{2}^{2}\leq-2\log\eta
\,\wedge\,x\in B_{i}'\cup B_{j}'
\Longrightarrow \sign(x_{1})\cdot(1_{A_{i}}(x)-1_{A_{j}}(x))>0.
\end{equation}

Since $\eta^{1/5}\leq\rho^{.9(K+R)}\leq\eta^{1/5}\rho^{-.9}$, note that $\eta^{1/10}\leq\rho^{.45(K+R)}\leq\eta^{1/10}\rho^{-.45}$, so for $R\geq2$, we have $2^{K}\eta^{6/5}\rho^{-.9}\leq2^{K}\rho^{.45K}\rho^{.45R}\eta^{11/10}\rho^{-.9}<\eta^{11/10}$.  If $R=1$, and if $K\geq2$, note that $2^{K}\eta^{6/5}\rho^{-.9}\leq2^{K}\rho^{.2K}\rho^{.25K}\rho^{.45}\eta^{11/10}\rho^{-.9}<\eta^{11/10}$.  If $R=0$, $K\geq3$ then $2^{K}\eta^{6/5}\rho^{-.9}\leq2^{K}\rho^{.1K}\rho^{.35K}\eta^{11/10}\rho^{-.9}<\eta^{11/10}$.  If $1\leq R+K\leq3$, then $(1/5)\log\eta\leq3\log\rho$ and $2\log\rho\leq(1/5)\log\eta$, so \eqref{eight6} directly implies \eqref{eight21}.  More specifically, by \eqref{eight6}, Lemma \ref{lemma7},\eqref{five8.2} and \eqref{five8.33}, $\sign(x_{1})\cdot(1_{A_{i}}-1_{A_{j}})(x)>0$ for $x\in B_{i}'\cup B_{j}'$ with $\vnorm{x}_{2}^{2}\leq-2\log\eta$ and $\eta\rho^{.9K}\rho^{.9}\leq\abs{x_{1}}\leq\eta$.  Now, $\rho^{.45(K+R)}\leq\eta^{1/10}\rho^{-.45}$, so $\eta\rho^{.9K}\rho^{.9}=\eta\rho^{.45K}\rho^{.45K}\rho^{.9}\leq\eta\rho^{.45K}\rho^{.45(K+R)}\leq\eta^{11/10}$.

In the latter case, \eqref{eight21} follows, and in the former cases, \eqref{eight20} implies
\begin{equation}\label{eight21}
\eta^{11/10}\leq\abs{x_{1}}\leq1
\,\wedge\,\vnorm{x}_{2}^{2}\leq-2\log\eta
\,\wedge\,x\in B_{i}'\cup B_{j}'
\Longrightarrow \sign(x_{1})\cdot(1_{A_{i}}(x)-1_{A_{j}}(x))>0.
\end{equation}
In all cases, \eqref{eight21} holds.  We can finally use \eqref{eight21} to conclude the proof.

\noindent\embolden{Step 6.  Using Step 5 to get an estimate for large values of $x$.}

Redefine $B_{i},B_{j}$ so that
\begin{equation}\label{five68}
B_{i}\colonequals B_{i,2\eta^{11/10}/\sqrt{-2\log\eta}}\,,\quad
B_{j}\colonequals B_{j,2\eta^{11/10}/\sqrt{-2\log\eta}}.
\end{equation}

Let $\sigma\colon\R^{n}\to\R^{n}$ be any rotation fixing the $x_{1}$-axis.  Let $x$ with $\vnorm{x}_{2}^{2}\leq-4\log(\eta\rho)\leq-8\log\eta$ and $\eta^{21/20}\rho^{1/2}\leq\abs{x_{1}}\leq(\eta\rho)^{1/4}$.  Let $B\colonequals(B_{i}\cap B_{j})\cup[(B_{i}\cup B_{j})\setminus(B_{i}'\cup B_{j}')]$.  Suppose $x\in B_{i}\cup B_{j}$ also.  Combining \eqref{eight21}, \eqref{five8.3}, and \eqref{five0.9}, $g\neq0$ only on the following sets:
\begin{flalign*}
&\{y\in\R^{n}\colon\absf{d(\sigma y-\rho x,B)}\leq
\eta^{11/10}/\sqrt{1-\rho^{2}}\},\\
&\{y\in\R^{n}\colon\absf{d(\sigma y-\rho x,B)}\leq(\rho\eta)^{1/4}/\sqrt{1-\rho^{2}},\vnorm{\sigma y-\rho x}_{2}\geq\sqrt{-2\log\eta}/\sqrt{1-\rho^{2}}\},\\
&\{y\in\R^{n}\colon\vnorm{\sigma y-\rho x}_{2}>(1+1/10)\sqrt{-3\log(\rho\eta)}/\sqrt{1-\rho^{2}}\}.
\end{flalign*}
We then apply Lemma \ref{lemma6.1} to get
\begin{equation}\label{eight26}
\begin{aligned}
&\sup_{\substack{t_{1}\in[\min(x_{1},0),\max(x_{1},0)]\\\,t_{2}\in[\min(x_{2},0),\max(x_{2},0)]\\ \alpha\in[0,\rho]}}
\bigg|\int_{\R^{n}}\sum_{\substack{\ell\in\N^{n}\colon\\0\leq\abs{\ell}\leq3}}
\prod_{i=1}^{n}\abs{y_{i}}^{\ell_{i}}g((t_{1},t_{2},0,\ldots,0)\alpha+y\sqrt{1-\alpha^{2}})d\gamma_{n}(y)\bigg|\\
&\qquad\leq500000n^{3}\eta^{11/10}
+500000n^{3}(\rho\eta)^{1/4}(-2(1-2\rho)^{2}\log\eta+1)\eta^{(1-2\rho)^{2}}\\
&\qquad\quad+200(n+2)!((-3\log(\eta\rho))^{(n+1)/2}+1)(\rho\eta)^{3/2}
+1600(n+2)!2\eta^{11/10}/\sqrt{-2\log\eta}.
\end{aligned}
\end{equation}

Applying \eqref{eight26} to Lemma \ref{lemma7}, using \eqref{five6.0} and $\eta<\rho<e^{-20(n+1)^{10^{12}n^{3}(n+2)!}}$,
\begin{equation}\label{eight27}
\begin{aligned}
&\vnorm{x}_{2}^{2}\leq-4\log(\eta\rho)\,\wedge\,
\eta^{21/20}\rho^{1/2}\leq\abs{x_{1}}\leq(\eta\rho)^{1/4}\,\wedge\,x\in B_{i}\cup B_{j}\\
&\qquad\Longrightarrow\rho^{-1}\abs{LT_{\rho}(1_{A_{i}}-1_{A_{j}})(x)-LT_{\rho}(1_{B_{i}}-1_{B_{j}})(x)}<\frac{1}{10}\eta^{21/20}\rho^{1/2}.
\end{aligned}
\end{equation}

Also, by \eqref{five6},
\begin{equation}\label{eight28}
\begin{aligned}
&\eta^{21/20}\rho^{1/2}\leq\abs{x_{1}}\leq(\eta\rho)^{1/4}
\,\wedge\,\vnorm{x}_{2}^{2}-4\log(\eta\rho)
\,\wedge\,x\in B_{i}\cup B_{j}\\
&\qquad\qquad\Longrightarrow\rho^{-1}\sign(x_{1})\cdot LT_{\rho}(1_{B_{i}}-1_{B_{j}})(x)>\frac{1}{10}\eta^{21/20}\rho^{1/2}.
\end{aligned}
\end{equation}

So, combining \eqref{eight27}, \eqref{eight28} for all $i',j'\in\{1,\ldots,k\}$, $i'\neq j'$, Lemma \ref{lemma1}, \eqref{five68}, and by applying the inclusion-exclusion principle, \eqref{five8.2} and \eqref{five8.33},
\begin{equation}\label{eight29}
\begin{aligned}
&\eta^{21/20}\rho^{1/2}\leq\abs{x_{1}}\leq(\eta\rho)^{1/4}
\,\wedge\,\vnorm{x}_{2}^{2}\leq-4\log(\eta\rho)
\,\wedge\,x\in B_{i}'\cup B_{j}'\\
&\qquad\qquad\Longrightarrow \sign(x_{1})\cdot(1_{A_{i}}(x)-1_{A_{j}}(x))>0.
\end{aligned}
\end{equation}

So, \eqref{eight29} and \eqref{five8.3} say that
\begin{equation}\label{eight30}
\eta^{21/20}\rho^{1/2}\leq\abs{x_{1}}
\,\wedge\,\vnorm{x}_{2}^{2}\leq-4\log(\eta\rho)
\,\wedge\,x\in B_{i}'\cup B_{j}'
\Longrightarrow \sign(x_{1})\cdot(1_{A_{i}}(x)-1_{A_{j}}(x))>0.
\end{equation}

Finally, we use \eqref{eight30} in place of \eqref{eight21} and repeat the computations \eqref{eight26} through \eqref{eight29} to get
\begin{equation}\label{eight32}
\eta\rho\leq\abs{x_{1}}
\,\wedge\,\vnorm{x}_{2}^{2}\leq-4\log(\eta\rho)
\,\wedge\,x\in B_{i}'\cup B_{j}'
\Longrightarrow \sign(x_{1})\cdot(1_{A_{i}}(x)-1_{A_{j}}(x))>0.
\end{equation}

In conclusion, \eqref{five0.4} follows from \eqref{eight32} and \eqref{five0.5}, letting $B_{i}''\colonequals B_{i}'$ and $B_{j}''\colonequals B_{j}'$.
%

\noindent\embolden{Step 7.  Completing the proof.}

For completeness, we derive \eqref{eight32}.  Redefine $B_{i},B_{j}$ so that
\begin{equation}\label{five70}
B_{i}\colonequals B_{i,2\eta^{21/20}\rho^{1/2}/\sqrt{-4\log(\eta\rho)}}\,,\quad
B_{j}\colonequals B_{j,2\eta^{21/20}\rho^{1/2}/\sqrt{-4\log(\eta\rho)}}.
\end{equation}

Let $\sigma\colon\R^{n}\to\R^{n}$ be any rotation fixing the $x_{1}$-axis.  Let $x$ with $\vnorm{x}_{2}^{2}\leq-4\log(\eta\rho)\leq-8\log\eta$ and $\eta\rho\leq\abs{x_{1}}\leq(\eta\rho)^{1/3}$.  Let $B\colonequals(B_{i}\cap B_{j})\cup[(B_{i}\cup B_{j})\setminus(B_{i}'\cup B_{j}')]$.  Suppose $x\in B_{i}\cup B_{j}$ also.  Combining \eqref{eight30}, \eqref{five8.3}, and \eqref{five0.9}, $g\neq0$ only on the following sets:
\begin{flalign*}
&\{y\in\R^{n}\colon\absf{d(\sigma y-\rho x,B)}\leq
\eta^{21/20}\rho^{1/2}/\sqrt{1-\rho^{2}}\},\\
&\{y\in\R^{n}\colon\absf{d(\sigma y-\rho x,B)}\leq(\rho\eta)^{1/3}/\sqrt{1-\rho^{2}},\vnorm{\sigma y-\rho x}_{2}\geq\sqrt{-2\log\eta}/\sqrt{1-\rho^{2}}\},\\
&\{y\in\R^{n}\colon\vnorm{\sigma y-\rho x}_{2}>(1+1/10)\sqrt{-3\log(\rho\eta)}/\sqrt{1-\rho^{2}}\}.
\end{flalign*}
We then apply Lemma \ref{lemma6.1} to get
\begin{equation}\label{eight36}
\begin{aligned}
&\sup_{\substack{t_{1}\in[\min(x_{1},0),\max(x_{1},0)]\\\,t_{2}\in[\min(x_{2},0),\max(x_{2},0)]\\ \alpha\in[0,\rho]}}
\bigg|\int_{\R^{n}}\sum_{\substack{\ell\in\N^{n}\colon\\0\leq\abs{\ell}\leq3}}
\prod_{i=1}^{n}\abs{y_{i}}^{\ell_{i}}g((t_{1},t_{2},0,\ldots,0)\alpha+y\sqrt{1-\alpha^{2}})d\gamma_{n}(y)\bigg|\\
&\qquad\leq500000n^{3}\eta^{21/20}\rho^{1/2}
+500000n^{3}(\rho\eta)^{1/3}(-2(1-2\rho)^{2}\log\eta+1)\eta^{(1-2\rho)^{2}}\\
&\qquad\quad+200(n+2)!((-3\log(\eta\rho))^{(n+1)/2}+1)(\rho\eta)^{3/2}
+1600(n+2)!2\eta^{6/5}\rho^{1/2}/\sqrt{-4\log(\eta\rho)}.
\end{aligned}
\end{equation}

Applying \eqref{eight36} to Lemma \ref{lemma7}, using \eqref{five6.0} and $\eta<\rho<e^{-20(n+1)^{10^{12}n^{3}(n+2)!}}$,
\begin{equation}\label{eight37}
\begin{aligned}
&\vnorm{x}_{2}^{2}\leq-4\log(\eta\rho)\,\wedge\,
\eta\rho\leq\abs{x_{1}}\leq(\eta\rho)^{1/3}\,\wedge\,x\in B_{i}\cup B_{j}\\
&\qquad\Longrightarrow\rho^{-1}\abs{LT_{\rho}(1_{A_{i}}-1_{A_{j}})(x)-LT_{\rho}(1_{B_{i}}-1_{B_{j}})(x)}<\frac{1}{10}\eta\rho.
\end{aligned}
\end{equation}

Also, by \eqref{five6},
\begin{equation}\label{eight38}
\begin{aligned}
&\eta\rho\leq\abs{x_{1}}\leq(\eta\rho)^{1/3}
\,\wedge\,\vnorm{x}_{2}^{2}\leq-4\log(\eta\rho)
\,\wedge\,x\in B_{i}\cup B_{j}\\
&\qquad\qquad\Longrightarrow\rho^{-1}\sign(x_{1})\cdot LT_{\rho}(1_{B_{i}}-1_{B_{j}})(x)>\frac{1}{10}\eta\rho.
\end{aligned}
\end{equation}

So, combining \eqref{eight37}, \eqref{eight38} for all $i',j'\in\{1,\ldots,k\}$, $i'\neq j'$, Lemma \ref{lemma1}, \eqref{five70}, and by applying the inclusion-exclusion principle, \eqref{five8.2} and \eqref{five8.33},
\begin{equation}\label{eight39}
\eta\rho\leq\abs{x_{1}}\leq(\eta\rho)^{1/3}
\,\wedge\,\vnorm{x}_{2}^{2}\leq-4\log(\eta\rho)
\,\wedge\,x\in B_{i}'\cup B_{j}'
\Longrightarrow \sign(x_{1})\cdot(1_{A_{i}}(x)-1_{A_{j}}(x))>0.
\end{equation}

Then, \eqref{eight39} and \eqref{five8.3} say that
\begin{equation}\label{eight40}
\eta\rho\leq\abs{x_{1}}
\,\wedge\,\vnorm{x}_{2}^{2}\leq-4\log(\eta\rho)
\,\wedge\,x\in B_{i}'\cup B_{j}'
\Longrightarrow \sign(x_{1})\cdot(1_{A_{i}}(x)-1_{A_{j}}(x))>0.
\end{equation}

Finally, \eqref{eight32} follows from \eqref{eight40}, completing the proof.
\end{proof}

\section{Proof of the Main Theorem}\label{secmain}

We now combine the Lemmas of the previous sections, as described in Section \ref{secintro}.  The main effort involves verifying the assumption of the Main Lemma \ref{lemma8}. Once this is done, Lemma \ref{lemma8} can be iterated infinitely many times to complete the proof.

\begin{theorem}\label{thm1}
Fix $k=3$, $n\geq2$.  Define $\Delta_{k}(\gamma_{n})$ as in Definition \ref{Deltadef} and define $\psi_{\rho}$ as in \eqref{three1.5}.  Let $\{C_{i}\}_{i=1}^{k}\subset\R^{n}$ be a regular simplicial conical partition.  Then there exists $\rho_{0}=\rho_{0}(n,k)>0$ such that, for all $\rho\in(0,\rho_{0})$, $(1_{C_{1}},\ldots,1_{C_{k}})$ uniquely achieves the following supremum, up to rotation
$$
\sup_{(f_{1},\ldots,f_{k})\in\Delta_{k}(\gamma_{n})}
\rho^{-1}\sum_{i=1}^{k}\int_{\R^{n}} f_{i}LT_{\rho}f_{i}d\gamma_{n}
=\sup_{(f_{1},\ldots,f_{k})\in\Delta_{k}(\gamma_{n})}
\psi_{\rho}(f_{1},\ldots,f_{k}).
$$
\end{theorem}
\begin{proof}
Within the proof, we will assert that $\rho>0$ satisfies several upper bounds, and then at the end of the proof, we will define $\rho_{0}$ as the minimum of these upper bounds.  By Lemma \ref{lemma1}, let $\{A_{i}\}_{i=1}^{k}$ be a partition of $\R^{n}$ such that

\begin{equation}\label{three5.7}
\psi_{\rho}(1_{A_{1}},\ldots,1_{A_{k}})
=\sup_{(f_{1},\ldots,f_{k})\in\Delta_{k}(\gamma_{n})}
\psi_{\rho}(f_{1},\ldots,f_{k}).
\end{equation}

By \eqref{six2}, write
\begin{equation}\label{three5.8}
\rho^{-1}\sum_{i=1}^{k}\int_{\R^{n}} 1_{A_{i}}LT_{\rho}1_{A_{i}}d\gamma_{n}
=\sum_{i=1}^{k}\sum_{\ell\in\N^{n}}\abs{\ell}\abs{\int_{\R^{n}} 1_{A_{i}}\sqrt{\ell!}h_{\ell}d\gamma_{n}}^{2}\rho^{\abs{\ell}-1}.
\end{equation}

\noindent\embolden{Step 1.  The partition $\{A_{i}\}_{i=1}^{k}$ is close to being simplicial.}

For $i\in\{1,\ldots,k\}$, let $z_{i}\colonequals\int_{A_{i}}xd\gamma_{n}(x)\in\R^{n}$.  Subtracting the $\abs{\ell}=1$ term from both sides of \eqref{three5.8}, treating the remaining terms as error terms, and using that $\vnorm{1_{A_{i}}}_{L_{2}(\gamma_{n})}\leq1$ for all $i=1,\ldots,k$,
\begin{equation}\label{three5p2}
\bigg|\rho^{-1}\sum_{i=1}^{k}\int_{\R^{n}} 1_{A_{i}}LT_{\rho}1_{A_{i}}d\gamma_{n}-\sum_{i=1}^{k}\vnorm{z_{i}}_{\ell_{2}^{n}}^{2}\bigg|
\leq3k\rho.
\end{equation}
Therefore,
\begin{flalign*}
&\sum_{i=1}^{k}\vnorm{z_{i}}_{\ell_{2}^{n}}^{2}
\stackrel{\eqref{six3}}{=}\psi_{0}(1_{A_{1}},\ldots,1_{A_{k}})
\stackrel{\eqref{three5p2}}{\geq}\psi_{\rho}(1_{A_{1}},\ldots,1_{A_{k}})-3k\rho
\stackrel{\eqref{three5.7}}{\geq}\psi_{\rho}(1_{B_{1}},\ldots,1_{B_{k}})-3k\rho\\
&\quad\stackrel{\eqref{three5p2}}{\geq}\psi_{0}(1_{B_{1}},\ldots,1_{B_{k}})-6k\rho
\stackrel{(\mathrm{Lemma}\,\,\ref{lemma0})}{=}\sup_{(f_{1},\ldots,f_{k})\in\Delta_{k}(\gamma_{n})}\psi_{0}(f_{1},\ldots,f_{k})-6k\rho.
\end{flalign*}

\noindent\embolden{Step 2.  Applying a small rotation.}

For $i\in\{1,\ldots,k\}$, let $w_{i}\colonequals\int_{B_{i}}xd\gamma_{n}(x)$.  Let $\rho>0$ such that $6k\rho<10^{-2}$.  Then by Lemma \ref{lemma5},
\begin{equation}\label{three5p}
d_{2}(\{A_{i}\}_{i=1}^{k},\{B_{i}\}_{i=1}^{k})<6(6k\rho)^{1/8},
\end{equation}
\begin{equation}\label{three5p1}
\inf_{\sigma\in SO(n)}\left(\sum_{i=1}^{k}\vnorm{\sigma z_{i}-w_{i}}_{\ell_{2}^{n}}^{2}\right)^{1/2}<6(6k\rho)^{1/8}.
\end{equation}
Note that \eqref{three5p1} follows from \eqref{three5p} by Hilbert space duality and since the set of functions $\{x_{i}\}_{i=1}^{n}$ are contained in the orthonormal basis $\{h_{\ell}\sqrt{\ell!}\}_{\ell\in\N^{n}}$ of $L_{2}(\gamma_{n})$.

Let $x=(x_{1},\ldots,x_{n})\in\R^{n}$, let $i,j\in\{1,\ldots,k\}$, $i\neq j$, and write the following equality of $L_{2}$ functions
\begin{equation}\label{three6}
1_{A_{i}}(x)-1_{A_{j}}(x)\equalscolon\sum_{\ell\in \N^{n}}c_{\ell}h_{\ell}(x)\sqrt{\ell!}.
\end{equation}
Let $\ell=(\ell_{1},\ldots,\ell_{n})\in\N^{n}$.  By applying an orthogonal change of coordinates to $\{A_{p}\}_{p=1}^{k}$, we may assume that $c_{\ell}=0$ when $\abs{\ell}=1$, $\ell_{1}=0$.  By \eqref{six1.1} and \eqref{six1}, write
\begin{equation}\label{three7}
\rho^{-1}LT_{\rho}(1_{A_{i}}-1_{A_{j}})(x)=\sum_{\ell\in\N^{n}}c_{\ell}\abs{\ell}\rho^{\abs{\ell}-1}h_{\ell}(x)\sqrt{\ell!}.
\end{equation}

Since $d_{2}(\{A_{p}\}_{p=1}^{k},\{B_{p}\}_{p=1}^{k})<6(6k\rho)^{1/8}$, there exists $\{B_{p}''\}_{p=1}^{k}$ a regular simplicial conical partition, such that $(\sum_{p=1}^{k}\vnormf{1_{A_{p}}-1_{B_{p}''}}_{L_{2}(\gamma_{n})}^{2})^{1/2}<6(6k\rho)^{1/8}$.  In particular, by Hilbert space duality,
\begin{equation}\label{three7.8}
\bigg|\bigg|\int_{\R^{n}}x(1_{A_{i}}(x)-1_{A_{j}}(x)-(1_{B_{i}''}(x)-1_{B_{j}''}(x)))d\gamma_{n}(x)\bigg|\bigg|_{\ell_{2}^{n}}<6(6k\rho)^{1/8},
\end{equation}
\begin{equation}\label{three7.88}
\bigg|\bigg|\int_{\R^{n}}x(1_{(A_{i}\cup A_{j})^{c}}(x)-1_{(B_{i}''\cup B_{j}'')^{c}}(x))d\gamma_{n}(x)\bigg|\bigg|_{\ell_{2}^{n}}<6(6k\rho)^{1/8}.
\end{equation}

Since $k=3$, and since $\sum_{p=1}^{k}\int_{A_{p}}xd\gamma_{n}(x)=\int_{\R^{n}}xd\gamma_{n}(x)=0$, there exists a $2$-dimensional plane $\Pi\subset\R^{n}$ such that $0\in\Pi$ and such that, for all $p\in\{1,\ldots,k\}$, $\int_{A_{p}}xd\gamma_{n}(x)\in\Pi$.  Without loss of generality, $\Pi$ contains the $x_{1}$ and $x_{2}$ axes.

Let $\{B_{p}'\}_{p=1}^{k}$ be a regular simplicial conical partition such that
\begin{equation}\label{two0.1}
\left(\sum_{p=1}^{k}\vnormf{1_{B_{p}'}-1_{B_{p}''}}_{L_{2}(\gamma_{n})}^{2}\right)^{1/2}<10(6k\rho)^{1/16},
\end{equation}
such that for fixed $i\neq j$, $i,j\in\{1,\ldots,k\}$ and for some $\lambda'\in\R$,
\begin{equation}\label{two0}
\int_{\R^{n}}x(1_{A_{i}}(x)-1_{A_{j}}(x))d\gamma_{n}(x)=\lambda'\int_{\R^{n}}x(1_{B_{i}'}(x)-1_{B_{j}'}(x))d\gamma_{n}(x),
\end{equation}
and such that
\begin{equation}\label{two0.4}
\int_{\R^{n}}x(1_{(B_{i}'\cup B_{j}')^{c}})d\gamma_{n}(x)\in\Pi.
\end{equation}

Such $\{B_{p}'\}_{p=1}^{k}$ exists by \eqref{three7.8}, letting $\rho>0$ such that $\rho<(10000k)^{-8}$, so that
$$\vnorm{\int_{\R^{n}}x(1_{B_{i}''}(x)-1_{B_{j}''}(x))d\gamma_{n}(x)}_{\ell_{2}^{n}}=3\sqrt{2}/(4\sqrt{\pi}).$$
$$\vnorm{\int_{\R^{n}}x(1_{(B_{i}''\cup B_{j}'')^{c}}(x))d\gamma_{n}(x)}_{\ell_{2}^{n}}
=\vnorm{\int_{\R^{n}}x(1_{B_{i}''}(x))}_{\ell_{2}^{n}}
=\sqrt{6}/(4\sqrt{\pi}).$$
So, by the triangle inequality applied to \eqref{three7.8}, and \eqref{three7.88},
\begin{equation}\label{three7.90}
\vnorm{\int_{\R^{n}}x(1_{A_{i}}(x)-1_{A_{j}}(x))d\gamma_{n}(x)}_{\ell_{2}^{n}}
>3\sqrt{2}/(4\sqrt{\pi})-10^{-2}>1/3.
\end{equation}
\begin{equation}\label{three7.91}
\vnorm{\int_{\R^{n}}x(1_{(A_{i}\cup A_{j})^{c}}(x))d\gamma_{n}(x)}_{\ell_{2}^{n}}
>\sqrt{6}/(4\sqrt{\pi})-10^{-2}>1/3.
\end{equation}
Specifically, we first apply a rotation to $\{B_{p}''\}_{p=1}^{k}$ such that \eqref{two0} holds.  Then, by \eqref{three7.88}, we then apply another rotation that fixes the $x_{1}$ axis, so that \eqref{two0.4} holds.  By \eqref{three7.8}, \eqref{three7.88}, \eqref{three7.90} and \eqref{three7.91}, each of these two rotations can be chosen so that a given unit vector is moved in $\R^{n}$ a distance not more than $12(6k\rho)^{1/8}$.  And since we are rotating three polygonal cones with two facets each, \eqref{two0.1} holds.

Using \eqref{two0.1} and the triangle inequality,
\begin{equation}\label{three7.93}
(\sum_{p=1}^{k}\vnormf{1_{A_{p}}-1_{B_{p}'}}_{L_{2}(\gamma_{n})}^{2})^{1/2}<20(6k\rho)^{1/16}.
\end{equation}
Also, using that $c_{\ell}=0$ for $\abs{\ell}=1,\ell_{1}=0$, \eqref{two0} implies that $B_{i}'\cap B_{j}'\subset\{x\in\R^{n}\colon x_{1}=0\}$, and we may assume that $B_{i}'\subset\{x\in\R^{n}\colon x_{1}\geq0\}$.

Let $n_{i}'\in\R^{n}$ denote the interior unit normal of $B_{i}'$ such that $n_{i}'$ is normal to $(\partial B_{i}')\setminus B_{j}'$, and let $n_{j}'\in\R^{n}$ denote the interior unit normal of $B_{j}'$ such that $n_{j}'$ is normal to $(\partial B_{j}')\setminus B_{i}'$.  Then, define $B_{i},B_{j}$ such that
\begin{equation}\label{five7.55}
\begin{aligned}
&B_{i}\colonequals B_{i}'\cup\{x\in\R^{n}\colon x_{1}\geq0
\wedge\langle n_{i}',x/\vnorm{x}_{2}\rangle\geq-4\rho^{21/20}/\sqrt{-3\log\rho}\},\\
&B_{j}\colonequals B_{j}'\cup\{x\in\R^{n}\colon x_{1}\leq0\wedge
\langle n_{j}',x/\vnorm{x}_{2}\rangle\geq-4\rho^{21/20}/\sqrt{-3\log\rho}\}.
\end{aligned}
\end{equation}

Since $B_{i}\cup B_{j}$ is symmetric with respect to reflection across $B_{i}\cap B_{j}=B_{i}'\cap B_{j}'$, equation \eqref{two0} implies that there is a $\lambda>0$ such that
\begin{equation}\label{two00}
\int_{\R^{n}}x(1_{A_{i}}(x)-1_{A_{j}}(x))d\gamma_{n}(x)=\lambda\int_{\R^{n}}x(1_{B_{i}}(x)-1_{B_{j}}(x))d\gamma_{n}(x).
\end{equation}

\noindent\embolden{Step 3.  An estimate for small $x$.}

Combining \eqref{six1}, \eqref{six1.1}, and \eqref{two00}, there exists $\abs{b_{1}}<50(6k\rho)^{1/16}$ (by Hilbert space duality, eqref{three7.93} and \eqref{five7.55}) such that
\begin{equation}\label{two1}
\rho^{-1}LT_{\rho}(1_{A_{i}}-1_{A_{j}})(x)-\rho^{-1}LT_{\rho}(1_{B_{i}}-1_{B_{j}})(x)-x_{1}b_{1}
\equalscolon\sum_{\ell\in\N^{n}\colon\abs{\ell}\geq2}b_{\ell}\abs{\ell}\rho^{\abs{\ell}-1}h_{\ell}(x)\sqrt{\ell!}.
\end{equation}
Choose $\rho_{1}$ so that $0<\rho<\rho_{1}$ implies $100k^{1/16}\sum_{m=2}^{\infty}m(m+n-1)^{n}\rho^{m-2}m^{n}3^{m}(-\log\rho^{3})^{m/2}<\rho^{-1/80}/20$.  Recall that the number of $\ell\in\N^{n}$ such that $\abs{\ell}=m$ is equal to 
$\frac{m+n-1!}{m!(n-1)!}\leq (m+n-1)^{n}$.  Note that, $\abs{b_{\ell}}<100k^{1/16}\rho^{1/16}$, for all $\ell\in\N^{n}$, $\abs{\ell}\geq2$, by Hilbert space duality.  Let $x\in\R^{n}$ with $\vnorm{x}_{2}^{2}\leq-\log\rho^{3}$.  By \eqref{two1}, Lemma \ref{lemma6},
\begin{equation}\label{two2}
\begin{aligned}
&\abs{\rho^{-1}LT_{\rho}(1_{A_{i}}-1_{A_{j}})(x)-\rho^{-1}LT_{\rho}(1_{B_{i}}-1_{B_{j}})(x)-x_{1}b_{1}}\\
&\qquad\leq100k^{1/8}\rho^{17/16}\sum_{\ell\in\N^{n}\colon\abs{\ell}\geq2}\abs{\ell}\rho^{\abs{\ell}-2}\abs{h_{\ell}(x)}\sqrt{\ell!}\\
&\qquad\leq100k^{1/8}\rho^{17/16}\sum_{\ell\in\N^{n}\colon\abs{\ell}\geq2}
\abs{\ell}\rho^{\abs{\ell}-2}\abs{\ell}^{n}3^{\abs{\ell}}\prod_{i=1}^{n}\max\{1,\abs{x_{i}}^{\ell_{i}}\}\\
&\qquad\leq100k^{1/8}\rho^{17/16}\sum_{m=2}^{\infty}
m(m+n-1)^{n}\rho^{m-2}m^{n}3^{m}(-\log\rho^{3})^{m/2}
\leq\rho^{21/20}/20.
\end{aligned}
\end{equation}

From Lemma \ref{lemma4} and \eqref{six1.1}, for $x=(x_{1},\ldots,x_{n})$, with $B_{i}\cap B_{j}\subset\{x\in\R^{n}\colon x_{1}=0\}$,
\begin{equation}\label{two3}
x\in B_{i}\cup B_{j}\,\wedge\,\vnorm{x}_{2}^{2}\leq-\log\rho^{3}
\Longrightarrow
\rho^{-1}\sign(x_{1})\cdot LT_{\rho}(1_{B_{i}}-1_{B_{j}})(x)>(1/10)\abs{x_{1}}.
\end{equation}

Then \eqref{two2} and \eqref{two3} show that
\begin{equation}\label{three7.5}
\abs{x_{1}}>\rho^{21/20}
\,\wedge\,\vnorm{x}_{2}^{2}\leq-\log\rho^{3}
\,\wedge\, x\in B_{i}\cup B_{j}
\Longrightarrow\rho^{-1}\sign(x_{1})\cdot LT_{\rho}(1_{A_{i}}-1_{A_{j}})(x)>0.
\end{equation}

By \eqref{three5.7}, Lemma \ref{lemma1}, and by applying \eqref{three7.5} for all $i',j'\in\{1,\ldots,k\}$, $i'\neq j'$, along with the inclusion-exclusion principle,
\begin{equation}\label{two3.1}
\abs{x_{1}}>\rho^{21/20}
\,\wedge\,\vnorm{x}_{2}^{2}\leq-\log\rho^{3}
\,\wedge\, x\in B_{i}'\cup B_{j}'
\Longrightarrow\sign(x_{1})\cdot(1_{A_{i}}(x)-1_{A_{j}}(x))>0.
\end{equation}
From \eqref{two3.1},
\begin{equation}\label{two3.11}
\begin{aligned}
&x\in B_{i}'\cup B_{j}'
\,\wedge\,\absf{d(x,(\partial B_{i}')\cup(\partial B_{j}'))}>\rho^{21/20}
\,\wedge\,\vnorm{x}_{2}^{2}\leq-\log\rho^{3}\\
&\qquad\qquad\Longrightarrow\sign(x_{1})\cdot(1_{A_{i}}(x)-1_{A_{j}}(x))>0.
\end{aligned}
\end{equation}

\noindent\embolden{Step 4.  Applying the Main Lemma.}

Recall that there exists a $2$-dimensional plane $\Pi\subset\R^{n}$ such that $0\in\Pi$ and such that, for all $p\in\{1,\ldots,k\}$, $\int_{A_{p}}xd\gamma_{n}(x)\in\Pi$.  Define
\begin{flalign*}
S&\colonequals\mathrm{span}\left\{\int_{\R^{n}}(1_{B_{i}'}(x)-1_{B_{j}'}(x))xd\gamma_{n}(x)
,\int_{\R^{n}}(1_{B_{i}'}(x)-1_{(B_{i}'\cup B_{j}')^{c}}(x))xd\gamma_{n}(x)\right.\\
&\qquad\qquad\qquad\qquad\qquad\qquad\left.,\int_{\R^{n}}(1_{B_{j}'}(x)-1_{(B_{i}'\cup B_{j}')^{c}}(x))xd\gamma_{n}(x)\right\}.
\end{flalign*}
Note that $S$ is a $2$-dimensional plane and $0\in S$.  By \eqref{two0}, $\int_{\R^{n}}(1_{B_{i}'}(x)-1_{B_{j}'}(x))d\gamma_{n}(x)\in\Pi$.  Moreover, since $\{B_{i}',B_{j}',(B_{i}'\cup B_{j}')^{c}\}$ is a regular simplicial conical partition,
$$
S=\mathrm{span}\left\{\int_{B_{i}'}xd\gamma_{n}(x),\int_{B_{j}'}xd\gamma_{n}(x),\int_{(B_{i}'\cup B_{j}')^{c}}xd\gamma_{n}(x)\right\}.
$$
From \eqref{two0.4}, $S$ and $\Pi$ are $2$-dimensional planes that both contain the linearly independent vectors $\int_{(B_{i}'\cup B_{j}')^{c}}xd\gamma_{n}(x)$ and $\int_{\R^{n}}(1_{B_{i}'}(x)-1_{B_{j}'}(x))xd\gamma_{n}(x)$.  We therefore conclude that $S=\Pi$.  In particular,
\begin{equation}\label{two4}
\int_{B_{p}'}xd\gamma_{n}(x)\in\Pi,\,\forall\,p\in\{i,j\}.
\end{equation}

Let $\rho_{0}\colonequals\min(\rho_{1},10^{-2}/6k,e^{-20(n+1)^{10^{12}n^{3}(n+2)!}})$.  Using \eqref{two3.11}, \eqref{two0} and \eqref{two4}, we can iteratively apply Lemma \ref{lemma8} an infinite number of times.  In particular, any time we know the conclusion \eqref{five0.4}, we use \eqref{five0.4} in the assumption \eqref{five0.3}.  That is, we first apply Lemma \ref{lemma8} with $\eta=\rho^{21/20}$.  In this case, since $\eta=\rho^{21/20}$, \eqref{two3.11} implies \eqref{five0.3}, \eqref{two4} implies \eqref{five0.9}, and \eqref{two00} implies \eqref{five0.5}.  Now, using the conclusion \eqref{five0.4} of Lemma \ref{lemma8}, we can then apply Lemma \ref{lemma8} with $\eta=\rho^{1+21/20}$.  Once again, using the conclusion \eqref{five0.4} of Lemma \ref{lemma8}, we can apply Lemma \ref{lemma8} with $\eta=\rho^{2+21/20}$, and so on.  Repeating this process infinitely many times shows that there exists a regular simplicial conical partition $\{C_{i}\}_{i=1}^{k}$ which is equal to $\{A_{i}\}_{i=1}^{k}$.
\end{proof}

The Main Theorem now follows from Theorem \ref{thm1} and the Fundamental Theorem of Calculus.

\begin{theorem}[\textbf{Main Theorem}]\label{thm2}
Let $n\geq2,k=3$.  There exists $\rho_{0}=\rho_{0}(n,k)>0$ such that Conjecture \ref{SSC} holds for $\rho\in(0,\rho_{0})$.  Moreover, up to orthogonal transformation, the regular simplicial conical partition uniquely achieves the maximum of \eqref{six1.5} in Conjecture \ref{SSC}.
\end{theorem}
\begin{proof}
Choose $\rho_{0}$ via Theorem \ref{thm1} and let $0<\rho<\rho_{0}$.  Let $\{B_{i}\}_{i=1}^{k}\subset\R^{n}$ be a regular simplicial conical partition.  By Theorem \ref{thm1} and the fact that $\Delta_{k}^{0}(\gamma_{n})\subset\Delta_{k}(\gamma_{n})$,
\begin{equation}\label{two3.2}
\psi_{\rho}(1_{B_{1}},\ldots,1_{B_{k}})
=\sup_{(f_{1},\ldots,f_{k})\in\Delta_{k}^{0}(\gamma_{n})}\psi_{\rho}(f_{1},\ldots,f_{k}).
\end{equation}
Let $(f_{1},\ldots,f_{k})\in\Delta_{k}^{0}(\gamma_{n})$.  By \eqref{six2}, $\sum_{i=1}^{k}\int_{\R^{n}} f_{i}T_{0}f_{i}d\gamma_{n}=k(1/k^{2})=1/k$.  By the Fundamental Theorem of Calculus and \eqref{two3.2},
\begin{flalign*}
&\sum_{i=1}^{k}\int_{\R^{n}} f_{i}T_{\rho}f_{i}d\gamma_{n}
=\int_{0}^{\rho}\bigg[\frac{d}{d\alpha}\sum_{i=1}^{k}\int_{\R^{n}} f_{i}T_{\alpha}f_{i}d\gamma_{n}\bigg]d\alpha+\frac{1}{k}
=\int_{0}^{\rho}\psi_{\alpha}(f_{1},\ldots,f_{k})d\alpha+\frac{1}{k}\\
&\leq\int_{0}^{\rho}\psi_{\alpha}(1_{B_{1}},\ldots,1_{B_{k}})d\alpha+\frac{1}{k}
=\int_{0}^{\rho}\bigg[\frac{d}{d\alpha}\sum_{i=1}^{k}\int_{\R^{n}} 1_{B_{i}}T_{\alpha}1_{B_{i}}d\gamma_{n}\bigg]d\alpha+\frac{1}{k}\\
&\qquad\qquad\qquad=\sum_{i=1}^{k}\int_{\R^{n}} 1_{B_{i}}T_{\rho}1_{B_{i}}d\gamma_{n}.
\end{flalign*}
\end{proof}

By using the invariance principle of \cite[Theorem 1.10,Theorem 3.6,Theorem 7.1,Theorem 7.4]{isaksson11} which transfers results from partitions of Euclidean space to low-influence discrete functions, Theorem \ref{thm2} implies a weak form of the Plurality is Stablest Conjecture.  While the following result is quite far from Conjecture \ref{PS} and might not be of immediate use to complexity theory, it is included to indicate a possible application of Theorem \ref{thm2}.  Essentially, if we modify the exact application of the invariance principle that is used in \cite[Theorem 7.1]{isaksson11}, then Conjecture \ref{PS} follows.  However, by avoiding \cite[Theorem 7.1]{isaksson11}, we must make very restrictive assumptions on the function $f$ in Conjecture \ref{PS}.  Nevertheless, \cite[Theorem 7.4]{isaksson11} shows that the class of functions $f$ described in Corollary \ref{thm0.1} is nonempty.

Note that the most straightforward application of Theorem \ref{thm2} only gives vacuous cases of Conjecture \ref{PS}, in which $0<\rho<\rho_{0}(n,k)$.  In particular, since Theorem \ref{thm2} requires $0<\rho<\rho_{0}(n,k)$, by \eqref{six2} we must take $\epsilon<3k\rho$ to get a nontrivial statement in Conjecture \ref{PS}.  In this case, the invariance principle \cite[Theorem 3.6]{isaksson11} gives $\tau$ with $\log\tau=-C(\log(\epsilon))^{2}(1/\epsilon)$, so that $\tau$ becomes a function of $\rho$.  Since we provide a $\rho$ with inverse exponential dependence on $n$, then $\tau$ also has inverse exponential dependence on $n$.  Thus, no function $f$ can satisfy the assumptions of Conjecture \ref{PS} in this case.  To avoid this issue, we modify Conjecture \ref{PS} as follows.
\begin{cor}[Weak Form of Plurality is Stablest]\label{thm0.1}
Let $\rho_{0}(n,k)$ be given by Theorem \ref{thm2}.  Fix $n\geq2$, $k=3$, and Let $N\colonequals\log\log\log\log\log(n)\geq1$.  Let $0<\rho<\rho_{0}(N,k)<1/2$, $\epsilon>0$, $\tau=\tau(\epsilon,k)>0$.  Let $f\colon\{1,\ldots,k\}^{n}\to\Delta_{k}$ with $\sum_{\sigma\in\{1,\ldots,k\}^{n}\colon\sigma_{j}\neq0}(\widehat{f_{i}}(\sigma))^{2}\leq\tau$ for all $i\in\{1,\ldots,k\}$, $j\in\{1,\ldots,n\}$.  Assume that there exists $0<m<N$ and  $g\colon\R^{m}\to\Delta_{k}$ with $\int_{\R^{m}} gd\gamma_{m}=\frac{1}{k^{n}}\sum_{\sigma\in\{1,\ldots,k\}^{n}}f(\sigma)$, and such that
$$
\bigg|\int_{\R^{n}}\langle\,g,T_{\rho}g\rangle d\gamma_{n}
-\frac{1}{k^{n}}\sum_{\sigma\in\{1,\ldots,k\}^{n}}\langle f(\sigma),T_{\rho}f(\sigma)\rangle\bigg|<\epsilon.
$$
Then part (a) of Conjecture \ref{PS} holds.  From \cite{isaksson11}[Theorem 7.4], this class of $f$ is nontrivial.
\end{cor}
%

Unfortunately, the proof of Theorem \ref{thm2} fails for small negative $\rho$, as we now show.

\begin{theorem}\label{thm3}
Fix $k=3$, $n\geq2$.  Define $\Delta_{k}^{0}(\gamma_{n})$ as in Definition \ref{dkedef} and define $\psi_{\rho}$ as in \eqref{three1.5}.  Let $\{B_{i}\}_{i=1}^{k}\subset\R^{n}$ be a regular simplicial conical partition.  Then there exists $\rho_{2}=\rho_{2}(n,k)>0$ such that, for $\rho\in(-\rho_{2},0)$, $(1_{B_{1}},\ldots,1_{B_{k}})$ does not achieve the following supremum.
$$
\sup_{(f_{1},\ldots,f_{k})\in\Delta_{k}^{0}(\gamma_{n})}
\rho^{-1}\sum_{i=1}^{k}\int_{\R^{n}} f_{i}LT_{\rho}f_{i}d\gamma_{n}
=\sup_{(f_{1},\ldots,f_{k})\in\Delta_{k}^{0}(\gamma_{n})}
\psi_{\rho}(f_{1},\ldots,f_{k}).
$$
\end{theorem}
\begin{proof}
Let $e_{1}=(1,0,\ldots,0)$, $e_{2}=(0,1,0,\ldots,0)$.  Fix $i,j\in\{1,\ldots,k\},i\neq j$.  Let $\sigma\colon\R^{n}\to\R^{n}$ denote reflection across $B_{i}\cap B_{j}$.  Since $B_{i}=\sigma(B_{j})$, by \eqref{three2.5}, it suffices to find $i,j\in\{1,\ldots,k\}$ and $x\in B_{i}$ such that $\rho^{-1}LT_{\rho}1_{B_{i}}(x)<\rho^{-1}LT_{\rho}1_{B_{j}}(x)$.  By replacing $\{B_{i}\}_{i=1}^{k}$ with $\{\tau B_{i}\}_{i=1}^{k}$ for $\tau\colon\R^{n}\to\R^{n}$ a rotation, we may assume that $\mathrm{span}\{z_{i}\}_{i=1}^{k}=\mathrm{span}\{e_{1},e_{2}\}$.  Moreover, we may assume $B_{i}\cap B_{j}\subset\{x\in\R^{n}\colon x_{1}=0\}$ and $B_{i}\subset\{x\in\R^{n}\colon x_{1}\geq0\}$.  Let $y\colonequals(\sqrt{3}/2)e_{1}+(1/2)e_{2}$, $\widetilde{y}\colonequals-(1/2)e_{1}+(\sqrt{3}/2)e_{2}$.  Fix $x\in B_{i}$ with $\langle x,\widetilde{y}\rangle>0$ also fixed.  From \eqref{four2} and the fact that $\rho<0$, there exists $c=c(\langle x,\widetilde{y}\rangle)>0$ such that
\begin{equation}\label{seven1}
\left\langle x,\frac{1}{\rho}\nabla T_{\rho}(1_{B_{i}}-1_{B_{j}})(x)\right\rangle
=-\langle x,\widetilde{y}\rangle(c+O(e^{-\langle x,y\rangle^{2}/2})).
\end{equation}

For $x\in\R^{n}$ with $\langle x,\widetilde{y}\rangle=0$, we have, as in Lemma \ref{lemma6.1}, and Lemma \ref{lemma3},
\begin{flalign*}
&\bigg|\int_{\R^{n}}\bigg(\sum_{i=1}^{n}(1-y_{i}^{2})(1_{B_{i}}-1_{B_{j}})(x\rho+y\sqrt{1-\rho^{2}})d\gamma_{n}(y)\bigg)\bigg|\\
&\qquad\leq2\bigg|\int_{B(0,\rho\vnorm{x}_{2})}\sum_{i=1}^{n}(1-y_{i}^{2})d\gamma_{n}(y)\bigg|
\leq200(n+1)!((\rho\vnorm{x}_{2})^{n}+1)e^{-\rho^{2}\vnorm{x}_{2}^{2}}.
\end{flalign*}
So, a derivative bound as in the proof of \eqref{four0.5} shows
\begin{equation}\label{seven2}
\begin{aligned}
&\bigg|\int_{\R^{n}}\bigg(\sum_{i=1}^{n}(1-y_{i}^{2})(1_{B_{i}}-1_{B_{j}})(x\rho+y\sqrt{1-\rho^{2}})d\gamma_{n}(y)\bigg)\bigg|\\
&\qquad\leq \rho\langle x,\widetilde{y}\rangle200(n+2)!
+200(n+1)!((\rho\vnorm{x}_{2})^{n}+1)e^{-\rho^{2}\vnorm{x}_{2}^{2}}.
\end{aligned}
\end{equation}

Then,
\begin{equation}\label{seven3}
\begin{aligned}
&\rho^{-1}LT_{\rho}(1_{B_{i}}-1_{B_{j}})(x)
\stackrel{\eqref{six1.1}}{=}\frac{1}{\rho}(\langle x,\nabla T_{\rho}(1_{B_{i}}-1_{B_{j}})(x)\rangle-\Delta T_{\rho}(1_{B_{i}}-1_{B_{j}})(x))\\
&=\langle x,T_{\rho}(\nabla (1_{B_{i}}-1_{B_{j}}))(x)\rangle
+\frac{\rho}{1-\rho^{2}}\int_{\R^{n}}\bigg(\sum_{i=1}^{n}(1-y_{i}^{2})(1_{B_{i}}-1_{B_{j}})(x\rho+y\sqrt{1-\rho^{2}})\bigg)d\gamma_{n}(y).
\end{aligned}
\end{equation}
So, choose $\rho<(c/8)(200(n+2)!)^{-1}$, then choose $\langle x,y\rangle$ sufficiently large, and then combine \eqref{seven1},\eqref{seven2} and \eqref{seven3} to get
$$
\rho^{-1}LT_{\rho}(1_{B_{i}}-1_{B_{j}})(x)<-\langle x,\widetilde{y}\rangle\frac{c}{4}.
$$
\end{proof}

\section{Open Problems}

There are two problems that are left open in this work.  First, Conjecture \ref{SSC} remains entirely open for $k\geq4$ partition elements.  Some of the results of this work hold for the case $k\geq4$, and some do not.  The first variation in Lemma \ref{lemma1} holds for all $k\geq4$.  Strictly speaking, the argument of Lemma \ref{lemma1} may not hold for $\rho<0$ since it is not clear whether or not the functional \eqref{three1.5} is convex.  Also, the technical error estimate from Lemma \ref{lemma7} holds.  One of the main issues for the case $k\geq4$ is that Lemma \ref{lemma0} is no longer available.  Moreover, the stability estimate in Lemma \ref{lemma5} would be needed for $k\geq4$. The following conjecture summarizes the main technical issue in proving an analogue of Lemma \ref{lemma0} for $k=4$, $n=3$.  If we could have a stability estimate for Conjecture \ref{c1} below, resembling the estimate of Lemma \ref{lemma5}, then in principle the proof of the Main Lemma, Lemma \ref{lemma8} would go through, and therefore Theorem \ref{thm2} would hold for $k\geq4$ as well.  Before we state the conjecture, recall Definition \ref{dkedef}, \eqref{three1.5} and \eqref{six3}.
\begin{conj}\label{c1}
Let $k=4$, $n=3$.  Suppose $\{A_{i}\}_{i=1}^{k}\subset\R^{n}$ satisfies
$$
\psi_{0}(1_{A_{1}},\ldots,1_{A_{k}})=\sup_{(f_{1},\ldots,f_{k})\in\Delta_{k}^{0}(\gamma_{n})}\psi_{0}(f_{1},\ldots,f_{k}).
$$
Then $\{A_{i}\}_{i=1}^{k}$ is a simplicial conical partition.
\end{conj}
This result is known to be true if we replace $\Delta_{k}^{0}(\gamma_{n})$ with $\Delta_{k}(\gamma_{n})$, by \cite[Lemma 3.3]{khot09}.  However, the volume constraint of $\Delta_{k}^{0}(\gamma_{n})$ causes difficulties for the methods of \cite{khot09,khot11}.

The second problem that remains open is Conjecture \ref{SSC} for $\rho<0$ or for $\rho$ positive and much larger than $0$.  We have already discussed the issues for $\rho<0$ in Theorem \ref{thm3}, where it is shown that our proof strategy surprisingly fails for $\rho<0$.  For $\rho$ with, e.g. $\rho\in(1/2,1)$, the error bounds that we use in the proof of Theorem \ref{thm2} seem to break down, especially when we apply Lemma \ref{lemma8}, Lemma \ref{lemma4}, and \eqref{four0.5}.  There is nothing special about our choice of $1/2$ here, other than that it is a positive number that is sufficiently far from $0$.  So, it seems that our method is not applicable for $\rho\in(1/2,1)$.  For example, Lemma \ref{lemma7} has an error term which is estimated by Lemma \ref{lemma6.1}.  However, the error estimate of Lemma \ref{lemma6.1} grows exponentially in $n$.  And to compensate for this error, we need to choose $\rho$ to decrease exponentially in $n$.  Even before we apply the Main Lemma \ref{lemma8}, there is also a loss in \eqref{two2}, where we essentially need a very specific $L_{\infty}$ bound on the Gaussian heat kernel (or Mehler kernel).  We used the rather crude method of bounding each Hermite polynomial separately in Lemma \ref{lemma6}, and then summing up these polynomials.  In principle, both of these losses could be avoided with dimension independent error estimates in Lemmas \ref{lemma7} and Lemma \ref{lemma6.1}.  However, this seems to be a difficult task.

However, since the case $\rho\in(1/2,1)$ relates to geometric multi-bubble problems, whereas the case of small $\rho$ seems to concern entirely different geometric information, it is unclear whether or not a single method could simultaneously solve or interpolate between different values of $\rho$ in Conjecture \ref{SSC}.

Finally, a new open problem has emerged subsequent to this work.  It turns out that if we modify the measure restriction
 in Conjecture \ref{SSC} in any way, then the analogue of Conjecture \ref{SSC} is false \cite{heilman14}.
 To be precise, in the statement of Conjecture \ref{SSC}, let
 $(a_{1},\ldots,a_{k})$ with $0<a_{i}<1$ for all $i=1,\ldots,k$, and such that $\sum_{i=1}^{k}a_{i}=1$.
 Assume that $(a_{1},\ldots,a_{k})\neq(1/k,\ldots,1/k)$.  Then, the partition $\{A_{i}\}_{i=1}^{k}\subset\R^{n}$
 which optimizes the noise stability \eqref{six1.5}
 subject to the constraint $\gamma_{n}(A_{i})=a_{i}$ for all $i=1,\ldots,k$ is not any translation of a regular simplicial
 conical partition.  In fact, the optimal partition $\{A_{i}\}_{i=1}^{k}$ has essentially no elementary description
 using simplices.  See \cite[Theorem 2.6]{heilman14} for a precise statement.  So, for example, the following question is
 open.
 \begin{question}
 Let $\rho>0$, $k=3$, $n=2$, and let $(a_{1},a_{2},a_{3})\in(0,1)^{3}$ with $\sum_{i=1}^{3}a_{i}=1$.  What is the
 partition $\{A_{i}\}_{i=1}^{3}$ maximizing the noise stability \eqref{six1.5} subject to the constraint $\gamma_{2}(A_{i})=a_{i}$
 for all $i=1,2,3$?
 \end{question}

\medskip

\noindent{\textbf{Acknowledgement.}} Thanks to Assaf Naor for guidance and encouragement, and for helpful comments concerning Corollary \ref{thm0.1}.  Thanks to Elchanan Mossel for helpful comments concerning the details of Corollary \ref{thm0.1} and the vacuous cases of Conjecture \ref{PS}.
Thanks to Oded Regev for reading the manuscript thoroughly and providing helpful comments.
Thanks also to the anonymous reviewers' thorough reading and helpful comments.

\bibliographystyle{abbrv}
\bibliography{12162011}

\begin{thebibliography}{10}

\bibitem{andrews99}
G.~E. Andrews, R.~Askey, and R.~Roy.
\newblock {\em Special functions}, volume~71 of {\em Encyclopedia of
  Mathematics and its Applications}.
\newblock Cambridge University Press, Cambridge, 1999.

\bibitem{austrin10}
P.~Austrin.
\newblock Towards sharp inapproximability for any 2-{CSP}.
\newblock {\em SIAM J. Comput.}, 39(6):2430--2463, 2010.

\bibitem{borell85}
C.~Borell.
\newblock Geometric bounds on the {O}rnstein-{U}hlenbeck velocity process.
\newblock {\em Z. Wahrsch. Verw. Gebiete}, 70(1):1--13, 1985.

\bibitem{braverman11}
M.~Braverman, K.~Makarychev, Y.~Makarychev, and A.~Naor.
\newblock The {G}rothendieck constant is strictly smaller than {K}rivine's
  bound.
\newblock {\em Forum Math. Pi}, 1:e4, 42, 2013.

\bibitem{chat06}
S.~Chatterjee.
\newblock A generalization of the {L}indeberg principle.
\newblock {\em Ann. Probab.}, 34(6):2061--2076, 2006.

\bibitem{corneli08}
J.~Corneli, I.~Corwin, S.~Hurder, V.~Sesum, Y.~Xu, E.~Adams, D.~Davis, M.~Lee,
  R.~Visocchi, and N.~Hoffman.
\newblock Double bubbles in {G}auss space and spheres.
\newblock {\em Houston J. Math.}, 34(1):181--204, 2008.

\bibitem{eldan13}
R.~Eldan.
\newblock A two-sided estimate for the {G}aussian noise stability deficit.
\newblock Preprint, \href{http://arxiv.org/abs/1307.2781}{arXiv:1307.2781},
  2013.

\bibitem{frieze95}
A.~Frieze and M.~Jerrum.
\newblock Improved approximation algorithms for {MAX} {$k$}-{CUT} and {MAX}
  {BISECTION}.
\newblock In {\em Integer programming and combinatorial optimization
  ({C}openhagen, 1995)}, volume 920 of {\em Lecture Notes in Comput. Sci.},
  pages 1--13. Springer, Berlin, 1995.

\bibitem{gross75}
L.~Gross.
\newblock Logarithmic {S}obolev inequalities.
\newblock {\em Amer. J. Math.}, 97(4):1061--1083, 1975.

\bibitem{heilman11}
S.~Heilman, A.~Jagannath, and A.~Naor.
\newblock Solution of the propeller conjecture in {$\mathbb{R}^{3}$}.
\newblock {\em Discrete \& Computational Geometry}, 50(2):263--305, 2013.

\bibitem{heilman14}
S.~Heilman, E.~Mossel, and J.~Neeman.
\newblock Standard simplices and pluralities are not the most noise stable.
\newblock preprint, \href{http://arxiv.org/abs/1403.0885}{arXiv:1403.0885},
  2014.

\bibitem{isaksson11}
M.~Isaksson and E.~Mossel.
\newblock Maximally stable {G}aussian partitions with discrete applications.
\newblock {\em Israel J. Math.}, 189:347--396, 2012.

\bibitem{khot07}
S.~Khot, G.~Kindler, E.~Mossel, and R.~O'Donnell.
\newblock Optimal inapproximability results for {MAX}-{CUT} and other
  2-variable {CSP}s?
\newblock {\em SIAM J. Comput.}, 37(1):319--357, 2007.

\bibitem{khot09}
S.~Khot and A.~Naor.
\newblock Approximate kernel clustering.
\newblock {\em Mathematika}, 55(1-2):129--165, 2009.

\bibitem{khot12}
S.~Khot and A.~Naor.
\newblock Grothendieck-type inequalities in combinatorial optimization.
\newblock {\em Comm. Pure Appl. Math.}, 65(7):992--1035, 2012.

\bibitem{khot11}
S.~Khot and A.~Naor.
\newblock Sharp kernel clustering algorithms and their associated grothendieck
  inequalities.
\newblock {\em Random Structures \& Algorithms}, 42(3):269--300, 2013.

\bibitem{ledoux96}
M.~Ledoux.
\newblock Isoperimetry and {G}aussian analysis.
\newblock In {\em Lectures on probability theory and statistics
  ({S}aint-{F}lour, 1994)}, volume 1648 of {\em Lecture Notes in Math.}, pages
  165--294. Springer, Berlin, 1996.

\bibitem{mossel12}
E.~Mossel and J.~Neeman.
\newblock Robust optimality of {G}aussian noise stability.
\newblock Preprint, \href{http://arxiv.org/abs/1210.4126}{arXiv:1210.4126},
  2012.

\bibitem{mossel10}
E.~Mossel, R.~O'Donnell, and K.~Oleszkiewicz.
\newblock Noise stability of functions with low influences: invariance and
  optimality.
\newblock {\em Ann. of Math. (2)}, 171(1):295--341, 2010.

\bibitem{rag09}
P.~Raghavendra and D.~Steurer.
\newblock Towards computing the {G}rothendieck constant.
\newblock In {\em Proceedings of the {T}wentieth {A}nnual {ACM}-{SIAM}
  {S}ymposium on {D}iscrete {A}lgorithms}, pages 525--534, Philadelphia, PA,
  2009. SIAM.

\bibitem{stein70}
E.~M. Stein.
\newblock {\em Singular integrals and differentiability properties of
  functions}.
\newblock Princeton Mathematical Series, No. 30. Princeton University Press,
  Princeton, N.J., 1970.

\end{thebibliography}

\section{Appendix: Differentiation of the Ornstein Uhlenbeck Semigroup}\label{secapp}

We prove \eqref{six1.1} and \eqref{six1.2}.  Let $\rho\in(-1,1)$ and let $f\colon\R^{n}\to\R$.  In the following calculations, we use integration by parts freely, and we use differentiation in the distributional sense.  We first calculate derivatives of $T_{\rho}f(x)$ with respect to $x\in\R^{n}$, and then we calculate the derivative of $T_{\rho}f(x)$ with respect to $\rho$.
%

\begin{flalign}
\frac{\partial}{\partial x_{i}}T_{\rho}f(x)
&\stackrel{\eqref{six0}}{=}\int\frac{\partial}{\partial x_{i}}[f(x\rho+y\sqrt{1-\rho^{2}})]d\gamma_{n}(y)\nonumber\\
&=\int\frac{\partial f(x\rho+y\sqrt{1-\rho^{2}})}{\partial z_{i}}d\gamma_{n}(y)\rho\nonumber\\
&=\int\frac{\partial}{\partial y_{i}}[f(x\rho+y\sqrt{1-\rho^{2}})]d\gamma_{n}(y)\frac{\rho}{\sqrt{1-\rho^{2}}}\nonumber\\
&=-\int f(x\rho+y\sqrt{1-\rho^{2}})\frac{\partial}{\partial y_{i}}[d\gamma_{n}(y)]\frac{\rho}{\sqrt{1-\rho^{2}}}\nonumber\\
&=\frac{\rho}{\sqrt{1-\rho^{2}}}\int y_{i}f(x\rho+y\sqrt{1-\rho^{2}})d\gamma_{n}(y).\label{nine1}
\end{flalign}

\begin{flalign}
\frac{\partial^{2}}{\partial x_{i}^{2}}T_{\rho}f(x)
&\stackrel{\eqref{nine1}}{=}\frac{\rho}{\sqrt{1-\rho^{2}}}\int\frac{\partial}{\partial x_{i}}[f(x\rho+y\sqrt{1-\rho^{2}})]y_{i}d\gamma_{n}(y)\nonumber\\
&=\frac{\rho}{\sqrt{1-\rho^{2}}}\int\frac{\partial f(x\rho+y\sqrt{1-\rho^{2}})}{\partial z_{i}}y_{i}d\gamma_{n}(y)\rho\nonumber\\
&=\frac{\rho^{2}}{1-\rho^{2}}\int\frac{\partial}{\partial y_{i}}[f(x\rho+y\sqrt{1-\rho^{2}})]y_{i}d\gamma_{n}(y)\nonumber\\
&=-\frac{\rho^{2}}{1-\rho^{2}}\int f(x\rho+y\sqrt{1-\rho^{2}})\frac{\partial}{\partial y_{i}}[y_{i}d\gamma_{n}(y)]\nonumber\\
&=-\frac{\rho^{2}}{1-\rho^{2}}\int f(x\rho+y\sqrt{1-\rho^{2}})(-y_{i}^{2}+1)d\gamma_{n}(y)\nonumber\\
&=\frac{\rho^{2}}{1-\rho^{2}}\int (y_{i}^{2}-1)f(x\rho+y\sqrt{1-\rho^{2}})d\gamma_{n}(y).\label{nine2}
\end{flalign}

\begin{flalign*}
&\frac{d}{d\rho}T_{\rho}f(x)
=\frac{d}{d\rho}\int_{\R^{n}}f(x\rho+y\sqrt{1-\rho^{2}})d\gamma_{n}(y)\\
&=\int_{\R^{n}}\sum_{i=1}^{n}\frac{\partial f(x\rho+y\sqrt{1-\rho^{2}})}{\partial z_{i}}
\left(x_{i}-y_{i}\frac{\rho}{\sqrt{1-\rho^{2}}}\right)d\gamma_{n}(y)\\
&=\int_{\R^{n}}\sum_{i=1}^{n}\frac{\partial}{\partial y_{i}}[f(x\rho+y\sqrt{1-\rho^{2}})]
\left(x_{i}-y_{i}\frac{\rho}{\sqrt{1-\rho^{2}}}\right)\frac{d\gamma_{n}(y)}{\sqrt{1-\rho^{2}}}\\
&=-\int_{\R^{n}}f(x\rho+y\sqrt{1-\rho^{2}})\sum_{i=1}^{n}\frac{\partial}{\partial y_{i}}
\left[\left(x_{i}-y_{i}\frac{\rho}{\sqrt{1-\rho^{2}}}\right)\frac{d\gamma_{n}(y)}{\sqrt{1-\rho^{2}}}\right]\\
&=-\int_{\R^{n}}f(x\rho+y\sqrt{1-\rho^{2}})\sum_{i=1}^{n}
\left[\left(x_{i}-y_{i}\frac{\rho}{\sqrt{1-\rho^{2}}}\right)(-y_{i})-\frac{\rho}{\sqrt{1-\rho^{2}}}\right]\frac{d\gamma_{n}(y)}{\sqrt{1-\rho^{2}}}\\
&=-\int_{\R^{n}}f(x\rho+y\sqrt{1-\rho^{2}})\sum_{i=1}^{n}
\left[(y_{i}^{2}-1)\frac{\rho}{\sqrt{1-\rho^{2}}}-x_{i}y_{i}\right]\frac{d\gamma_{n}(y)}{\sqrt{1-\rho^{2}}}\\
&=\frac{1}{\rho}\bigg[\frac{\rho}{\sqrt{1-\rho^{2}}}\left\langle x,\int_{\R^{n}}yf(x\rho+y\sqrt{1-\rho^{2}})d\gamma_{n}(y)\right\rangle\\
&\qquad\qquad\qquad
+\frac{\rho^{2}}{1-\rho^{2}}\int_{\R^{n}}\left(\sum_{i=1}^{n}(1-y_{i}^{2})\right)f(x\rho+y\sqrt{1-\rho^{2}})d\gamma_{n}(y)\bigg]\\
&\stackrel{\eqref{nine1}\wedge\eqref{nine2}}{=}\frac{1}{\rho}\left(\langle x,\nabla T_{\rho}f(x)\rangle-\Delta T_{\rho}f(x)\right).
\end{flalign*}

\end{document}